\documentclass[12pt]{article}
\usepackage[margin=1.25in]{geometry}
\usepackage{amsmath,amsthm,amssymb}
\usepackage{times}
\usepackage{graphicx}
\usepackage{enumerate}
\usepackage{setspace}
\usepackage[colorlinks,
linkcolor=blue,
anchorcolor=blue,
citecolor=blue
]{hyperref}
\usepackage{xurl}
\usepackage{float}
\usepackage[title]{appendix}
\usepackage{tikz}
\usetikzlibrary{decorations.pathreplacing}
\usetikzlibrary{arrows.meta,bending,positioning}

\usepackage{natbib}

\usepackage{multibib}
\newcites{supp}{References}

\newcommand{\R}{\mathbb{R}}

\makeatletter

\newcommand{\Rmnum}[1]{\expandafter\@slowromancap\romannumeral #1@}
\makeatother
\allowdisplaybreaks

\newtheorem{theorem}{Theorem}
\newtheorem{proposition}{Proposition}
\newtheorem{lemma}{Lemma}
\newtheorem{definition}{Definition}
\newtheorem{corollary}{Corollary}
\newtheorem{example}{Example}

\newtheorem{remark}{Remark}

\newtheorem*{lemma*}{Lemma}

\title{Persuasion with Ambiguous Communication\thanks{We thank Dilip Abreu, Ian Ball, Piotr Dworczak, Huiyi Guo, Alex Jakobsen, Jian Li, Alessandro Pavan, Eran Shmaya and seminar audiences at Texas A\&M, Washington University in St. Louis, University of California, Berkeley, RUD 2024, D-TEA 2024, EC'24, Northwestern University, University of Southern California, Queen Mary University of London, University of Notre Dame, HKUST, HKU Business School, 4th Durham Economic Theory Conference, Arizona State University, NYU, University of Manchester, Princeton University, and CETC 2025 for helpful feedback. A two-page abstract of an earlier version of this paper appeared in EC'24. Part of this work was carried out during Klibanoff’s Hallsworth Visiting Professor appointment at the University of Manchester. Mukerji acknowledges financial support from ESRC Grant ES/S015299/1.}}

\date{This Version: November 10, 2026}

\author{Xiaoyu Cheng\thanks{Department of Economics, Florida State University, Tallahassee, FL USA. E-mail: xcheng@fsu.edu.} \and Peter Klibanoff\thanks{Managerial Economics and Decision Sciences, Kellogg School of Management at Northwestern University, Evanston, IL USA. E-mail: peterk@kellogg.northwestern.edu.} \and Sujoy Mukerji\thanks{School of Economics and Finance, Queen Mary University of London, London, UK. E-mail: s.mukerji@qmul.ac.uk. } \and Ludovic Renou\thanks{Department of Economics, Arizona State University, Tempe, AZ USA, and CEPR. E-mail: $\text{lrenou.econ@gmail.com}$.}}
\begin{document}
\maketitle
\enlargethispage{\baselineskip}
\begin{abstract}
We explore whether ambiguous communication can be beneficial to the sender in persuasion games when the receiver is ambiguity averse and dynamically consistent. Splitting of obedient experiments into a convex combination of two Pareto-ranked experiments (those for which both players agree on a payoff ranking) plays a key role. That some obedient experiment can be split in this way is necessary for the sender to benefit. If an optimal Bayesian persuasion experiment can be split in this way, then any not-too-ambiguity-averse sender benefits. The sender never benefits from ambiguous communication when the receiver has only two actions. We identify a class of games in which our conditions fully characterize when the sender benefits and are generically satisfied when there are at least three actions.
\end{abstract}

\newpage 
\section{Introduction}	
\begin{quote}
\textit{``If I seem unduly clear to you, you must have misunderstood what I said.''}

Alan Greenspan, Speaking to a Senate Committee in 1987, as quoted in the Guardian Weekly, November 4, 2005.\end{quote}

This paper considers the problem of a sender who wishes to favorably influence, through strategic communication of information, the action taken by a receiver. As in the large literature on Bayesian persuasion following \citet{kamenica2011bayesian} (see also \citet{rayo2010optimal} and surveys by \citet{bergemann2019information} and \citet{kamenica2019bayesian}), we model the sender as committing to a communication strategy and the receiver as best responding to that strategy. A communication strategy for the sender is usually described as a \emph{statistical experiment}, a function mapping from payoff-relevant states to probability distributions over messages (or signals).\footnote{In Bayesian persuasion, it is without loss of generality to assume that the messages are action recommendations. We show (Proposition \ref{rev-principle}) that this remains true in our setting.} Our first key departure from most of the literature is that we enlarge the set of the sender's communication strategies to include ambiguous strategies. These are strategies generating, from the perspective of both players, subjective uncertainty about which statistical experiment will be used to generate the message. Our second key departure is that the receiver (and possibly the sender) treats this uncertainty as ambiguity and is ambiguity averse, and that the receiver is dynamically consistent. The main motivation for dynamic consistency is that we want to study whether and how the sender benefits from ambiguous communication even if the channel of dynamic \emph{in}consistency with the receiver's given ex-ante preferences — the channel at work in nearly all previous literature on mechanism or information design with ambiguity — is shut down.

Would the sender ever benefit from intentionally using an ambiguous communication strategy? We show that even with a consistent receiver the answer can be yes, and provide understanding of the circumstances under which this occurs and the nature of these beneficial strategies.

In which kind of persuasion games can ambiguity benefit the sender? We show that if an optimal Bayesian persuasion experiment can be written as a convex combination of two Pareto-ranked experiments (i.e., two experiments such that both players agree that the first is better than the second assuming action recommendations are followed), this is sufficient for an ambiguity-neutral sender to benefit from ambiguous communication (Theorem \ref{sufficient_BP}). These sufficient conditions are also necessary in a class of \emph{belief-threshold games} (Section \ref{sec_threshold_games}, Theorem \ref{thm:belief_threshold_iff}). For such games, these conditions are easy to check, and are generically satisfied when there are three or more actions. In these games, the sender has state-independent preferences and there is a \emph{default} action that is strictly worse for the sender than any other action. The receiver has a distinct most-preferred action for each state. For all but the default action, an action is expected payoff maximizing for the receiver only when their belief on the corresponding state exceeds a threshold larger than one-half; otherwise, the default action is better. These preferences capture a variety of economic settings studied in the communication literature, for example, seller-buyer and lobbyist-legislator interactions, see e.g., \citep{lipnowski2020}.

For such games, whenever an ambiguity-neutral sender can benefit from ambiguous communication, a particularly simple ambiguous communication strategy can improve over the optimal Bayesian persuasion strategy: take two Pareto-ranked experiments whose convex combination gives the optimal Bayesian persuasion experiment and make them the support of the ambiguous strategy. The source of the gain is in allowing the Pareto-superior experiment to be used more often than in Bayesian persuasion. The role of the ambiguity is to ensure that the receiver still obeys the recommendations (i.e., the ambiguous experiment is \emph{obedient}). An ambiguity-averse receiver desires to hedge more than an ambiguity-neutral receiver against the possibility that recommendations are coming from the inferior experiment. This allows the superior experiment to be used more often without inducing the receiver to depart from the recommended actions. An additional feature of this ambiguous strategy is that even the receiver’s payoff improves on their payoff under the optimal Bayesian persuasion strategy \textcolor{black}{(see Theorem \ref{general_improvement}(ii))}.

The sender's benefit from ambiguity beyond such threshold games mentioned above is also linked to Pareto-ranked convex combinations: the existence of a convex splitting of some (unambiguous) obedient experiment into two Pareto-ranked experiments is necessary for ambiguous communication to benefit the sender (Theorem \ref{thm_necessary}). This condition fails in any persuasion game in which the receiver has only two available actions. Thus it is \emph{never} possible for the sender to gain from ambiguous communication in such games (Corollary \ref{cor_binary_action}). In the same class of belief-threshold games, these necessary conditions are also sufficient. \textcolor{black}{Pareto-ranked experiments also play an essential role in the form of beneficial ambiguous communication: Any ambiguous communication strategy delivering the sender more than they receive under Bayesian persuasion must include in its support some pair of experiments that are Pareto-ranked (Theorem \ref{thm_necessary_condition_existence_of_Pareto_ranked_splitting}).}

We also analyze the sender's optimal persuasion strategy. We find that Pareto-ranked experiments are an essential feature of any optimal persuasion (Theorem \ref{thm_optimal_persuasion_pareto_ranked_experiments}). Furthermore, we provide a concavification-like characterization (Proposition \ref{thm_optimal_persuasion_characterization_experiment_space} and Corollary \ref{cor_optimal_persuasion_concavification_experiment_space}) of the optimum. 

We close this section with a brief discussion of a few closely related papers. A more extensive discussion can be found in Sections \ref{sec_comparison_BLL} and \ref{sec_related_literature}. \citet{beauchene2019ambiguous} (BLL henceforth) were first to study strategic use of ambiguous communication in persuasion (see also \citet{cheng2022relative}). The key difference in assumptions between BLL and our paper is how the receiver best responds given the sender's ambiguous experiment. We assume the receiver chooses an ex-ante optimal message-contingent strategy. Equivalently, we assume the receiver behaves as in a sequentially optimal equilibrium. This corresponds \citep[Theorem 1]{hanany2020incomplete} to the receiver using a simple dynamically consistent belief updating rule and maximizing their interim preferences given each message (see Remark \ref{rem_dynamic_consistency} in Section \ref{sec_incentive_compatibility} for updating details). BLL assume the receiver chooses, given each message, actions maximizing interim preferences formed using a belief updating rule that leads to dynamic inconsistency with their ex-ante preference. Thus, a key contribution of our paper is establishing and analyzing benefits of ambiguous persuasion that do not stem from receiver's behavior that is suboptimal with respect to their given ex-ante preferences. In Section \ref{sec_comparison_BLL}, we further compare with BLL, including establishing a sharp contrast in when and how the sender can gain from ambiguous communication in a class of belief-threshold games depending on whether the receiver behaves as in our model or as in BLL.\footnote{For discussion and approaches to dynamic consistency issues in decision-making under ambiguity more broadly see e.g., \citet{klibanoff2007updating, hanany2009updating} and \citet{siniscalchi2011dynamic}. \citet{hanany2020incomplete} address dynamic games of incomplete information with ambiguity averse players.} The bulk of BLL's analysis imposes the infinitely ambiguity-averse extreme for both the sender and receiver -- a polar case of our model, though they show that their approach extends more broadly. \cite{cheng2020ambiguous} shows that all benefits from ambiguous communication identified by BLL in the case of such a sender disappear if the receiver is assumed, as in our paper, to maximize their given ex-ante preference. In light of \cite{cheng2020ambiguous}'s result, it is essential that we allow at least the sender to be less than infinitely ambiguity averse for benefits from ambiguous communications to possibly exist. Our analysis allows for varying degrees of ambiguity aversion for both the sender and the receiver.

To illustrate our findings and provide intuition, we turn to an introductory example.

\section{An Introductory Example}\label{running_example}

There is a sender and a receiver, three actions $a_1,a_2$ and $a_3$, and two payoff-relevant states $\omega_1$ and $\omega_2$, with equal prior probabilities $p=(1/2,1/2)$. The sender influences the action the receiver takes with the release of information. The payoffs are: 

\begin{table}[H]
\begin{center}
\begin{tabular}{|c|c|c|c|c|}
\hline
$(u_{s}, u_{r})$ & $a_{1}$ & $a_{2}$ & $a_3$\\ 
\hline 
$\omega_{1}$ & $1, 1$ &  $-1, -1$ & $-4, 2$ \\ \hline
$\omega_2$ & $0,0$ & $2,2 $ & $-4,-4$\\ \hline
\end{tabular}	
\caption{Payoff table (first coordinate is the sender's payoff)} 
\end{center}
\end{table}

Both the sender and receiver prefer $a_2$ in state $\omega_2$. The receiver prefers $a_3$ in state $\omega_1$, while the sender prefers $a_1$ in that state. This is the conflict of interest in this example. The receiver prefers $a_1$ when their beliefs about $\omega_2$ are intermediate (i.e., in $[1/5,1/2]$), $a_2$ when their beliefs are higher than $1/2$, and $a_3$ when they are lower than $1/5$. The sender's challenge is to choose a communication strategy to better coordinate the receiver's play of actions $a_1$ and $a_2$ with the state, without inducing action $a_3$.

\begingroup
\color{black}
An interpretation of this example in the context of stress testing and banking regulation is as follows: Think of the sender as a banking regulatory authority (``the regulator") who must design, conduct and communicate the results of stress testing of the banking sector (``the bank"). Imagine the receiver as a representative investor (``the investor") choosing among alternative investments, $a_{i}$, whose payoff depends on the realization of the state $\omega$. Think of the $\omega$ as investment-relevant information about the health of the banking sector, with $\omega_{1}$ and $\omega_{2}$ associated with ``bad" and ``good" health, respectively. Actions $a_{1}$ and $a_{2}$ are socially-productive investments (i.e., productive from the viewpoint of the economy as a whole, a viewpoint that we assume the regulator adopts). Action $a_{3}$ is a socially-detrimental, purely speculative investment. The regulator's choice of communication strategy can be seen as their choice of rules/specifications for the stress tests.\footnote{\textcolor{black}{See \cite{bergemann2016informationbp} for another example of Bayesian persuasion in the context of stress testing.}} The regulator's challenge is to design and communicate stress tests so as to better coordinate investment behavior with the health of the banking sector, without diverting investments to the speculative activity, which is always socially detrimental but beneficial for the investor when the health of the banking sector is bad.
\par
\endgroup

Under Bayesian persuasion, the sender's optimal payoff is the concavification of their indirect utility at the prior \citep{kamenica2011bayesian}. In our example, this value is $5/4$ as illustrated in Figure \ref{fig:intro_example_1}. The figure plots the receiver's expected payoff associated with  each of the three actions as dotted lines -- each line is labelled with its action; the sender's indirect utility, i.e., the expected payoff the sender obtains when the receiver best responds, as a thick solid curve, and its concavification as a thick dashed curve. 

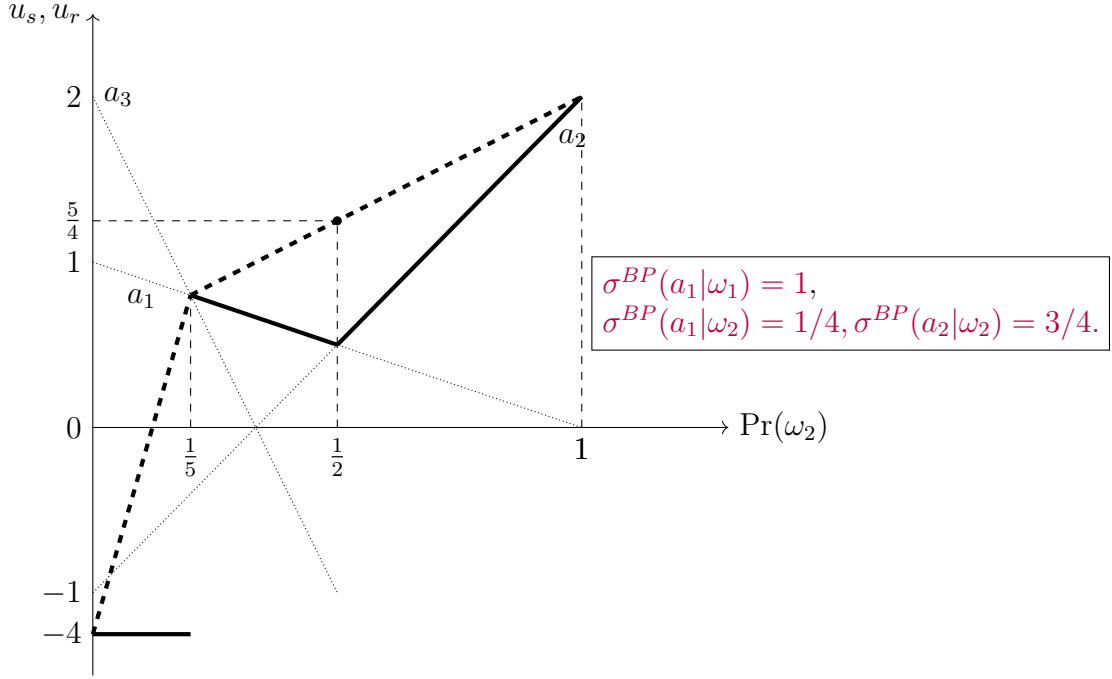
\begin{figure}[h]
\centering
\begin{tikzpicture}[x= 6.5 cm, y = 5.5cm] 
\draw[ ->] (0,0) -- (1.3,0);
\draw[ ->] (0, -0.6) -- (0,1);
\node [left] at (0, 0) {$0$};
\node [left] at (0, 1) {$u_s, u_r$};
\node [right] at (1.3,0) {$\mathrm{Pr}(\omega_{2})$};
\draw[densely dotted] (0, 0.4) -- (1, 0);
\node [below] at (0.1, 0.36) {$a_{1}$};
\node [left] at (0, 0.4) {$1$};
\draw[densely dotted] (0, -0.4) -- (1, 0.8);
\node [below] at (1, 0) {$1$};
\node [right] at (0.93, 0.7) {$a_{2}$};
\node [left] at (0, -0.4) {$-1$};
\draw[densely dotted] (0, 0.8) -- (0.5, -0.4);
\node [right] at (0, 0.8) {$a_{3}$};
\node [left] at (0, .8) {$2$};
\draw [dashed] (0.5, 0) -- (0.5, 0.2); 
\node [below] at (0.5, 0) {$\frac{1}{2}$};
\draw[dashed] (1,0) -- (1, 0.8);
\node [below] at (1, 0) {$1$};
\draw [ultra thick] (0, -0.5) -- (0.2, -0.5);
\node [left] at (0, -0.5) {$-4$};
\draw [dashed] (0.2, 0.0) -- (0.2, 0.32);
\draw [ultra thick] (0.2, 0.32) -- (0.5, 0.2);
\draw [ultra thick] (0.5, 0.2) -- (1, 0.8);
\node [below] at (0.2, 0) {$\frac{1}{5}$};
\draw [ultra thick, dashed] (0, -0.5) -- (0.2, 0.32);
\draw [ultra thick, dashed] (0.2, 0.32) -- (1, 0.8);
\draw [dashed] (0.5, 0.2) -- (0.5, 0.5);
\filldraw [] (0.5, 0.5) circle (1.5pt);
\draw [dashed] (0, 0.5) -- (0.5, 0.5); 
\node [left] at (0, 0.5) {\textcolor{black}{$\frac{5}{4}$}};
\node[draw,text width= 6.6 cm] at (1.55, 0.3) {$\textcolor{purple}{\sigma^{BP}(a_{1}|\omega_{1}) = 1},$\\ $\textcolor{purple}{\sigma^{BP}(a_{1}|\omega_{2}) = 1/4, \sigma^{BP}(a_{2}|\omega_{2}) = 3/4.}$};
\end{tikzpicture}
\caption{Sender's indirect utility (thick curve) and its concavification (thick dashed curve)}
\label{fig:intro_example_1}
\end{figure}

It is immediate to verify that the experiment $\sigma^{BP}$ defined in Figure \ref{fig:intro_example_1} attains this optimal payoff using the messages ``$a_1$'' and ``$a_2$''. The message ``$a_2$'' reveals that the state is $\omega_2$. Intuitively, since the preferences are perfectly aligned when the state is $\omega_2$, the sender wants the receiver to learn it. At the same time, the sender does not want the receiver to be too pessimistic about $\omega_2$ when the state is $\omega_1$, as the receiver chooses $a_3$ at all beliefs less than $1/5$ on $\omega_2$. The optimal experiment $\sigma^{BP}$ balances these two forces by making the likelihood of message ``$a_1$'' in state $\omega_2$ just large enough to induce action $a_1$. We note that the experiment $\sigma^{BP}$ is \emph{canonical}, that is, the messages are action recommendations, and, in the example, it is also \emph{obedient}, that is, the receiver finds it optimal to obey the recommendations. In Bayesian persuasion, the restriction to canonical and obedient experiments is without loss, and we prove (Proposition \ref{rev-principle}) that this remains true in our generalization.

Now, suppose that the sender can design \emph{ambiguous} experiments. These are communication strategies that leave some uncertainty about which statistical experiment will generate the messages, and, in this sense, may not completely pin-down the message likelihoods. We model ambiguous experiments as subjective (finite-support) distributions, $\mu$, over statistical experiments. The key implicit assumption is that the sender and receiver share this common view of the uncertainty over these experiments, would use the weights $\mu$ if they were ambiguity neutral, and that the receiver (and possibly the sender as well) is ambiguity averse and treats this uncertainty about the experiments as ambiguity. Conceptually, for example, such ambiguity could be generated by determining which experiment will be used by methods such as the outcome of artificial devices like draws from an Ellsberg urn, or by a third party who selects and conducts the experiment under public instructions by the sender that leave some room for discretion. By varying the artificial device or the instructions to the third party, the sender can vary $\mu$. Our theory is agnostic about why there is a common view of this uncertainty. Such a view may arise from shared (but limited) historical data or experience with the third party or, particularly in the case of artificially generated ambiguity, from symmetry or other logical considerations. Alternatively, the common-view assumption can simply be regarded as a convenient modeling benchmark.

\begin{figure}[h]
\centering
\begin{tikzpicture}[x= 6.5 cm, y = 5.5 cm] 
\draw[ ->] (0,0) -- (1.3,0);
\draw[ ->] (0, -0.6) -- (0,1);
\node [left] at (0, 0) {$0$};
\node [left] at (0, 1) {$u_s, u_r$};
\node [right] at (1.3,0) {$\mathrm{Pr}(\omega_{2})$};
\draw[densely dotted] (0, 0.4) -- (1, 0);
\node [below] at (0.1, 0.36) {$a_{1}$};
\node [left] at (0, 0.4) {$1$};
\draw[densely dotted] (0, -0.4) -- (1, 0.8);
\node [below] at (1, 0) {$1$};
\node [right] at (0.93, 0.7) {$a_{2}$};
\node [left] at (0, -0.4) {$-1$};
\draw[densely dotted] (0, 0.8) -- (0.5, -0.4);
\node [right] at (0, 0.8) {$a_{3}$};
\node [left] at (0, .8) {$2$};
\draw [dashed] (0.5, 0) -- (0.5, 0.2); 
\node [below] at (0.5, 0) {$\frac{1}{2}$};
\draw[dashed] (1,0) -- (1, 0.8);
\node [below] at (1, 0) {$1$};
\draw [ultra thick] (0, -0.5) -- (0.2, -0.5);
\node [left] at (0, -0.5) {$-4$};
\draw [dashed] (0.2, 0.0) -- (0.2, 0.32);
\draw [ultra thick] (0.2, 0.32) -- (0.5, 0.2);
\draw [ultra thick] (0.5, 0.2) -- (1, 0.8);
\node [below] at (0.2, 0) {$\frac{1}{5}$};
\draw [ultra thick, dashed] (0, -0.5) -- (0.2, 0.32);
\draw [ultra thick, dashed] (0.2, 0.32) -- (1, 0.8);
\draw [dashed] (0.5, 0.2) -- (0.5, 0.5);
\filldraw [] (0.5, 0.5) circle (2pt);
\draw [dashed] (0, 0.5) -- (0.5, 0.5); 
\node [left] at (0, 0.5) {\textcolor{black}{$\frac{5}{4}$}};
\draw [dashed] (0.5, 0.5) -- (0.5, 0.6);
\filldraw [blue] (0.5, 0.6) circle (2pt);
\draw [dashed] (0, 0.6) -- (0.5, 0.6);
\node [left] at (0, 0.6) {\textcolor{blue}{$\frac{3}{2}$}};   
\filldraw [orange] (0.5, 0.2) circle (2pt);
\draw [dashed] (0, 0.2) -- (0.5, 0.2);
\node [left] at (0, 0.2) {\textcolor{orange}{$\frac{1}{2}$}};
\draw [decorate,decoration={brace,amplitude=5pt, raise=3ex}] (0,0.2) -- (0,0.5);
\node [left] at (-0.1, 0.35) {\textcolor{blue}{$\mu_{\overline{\theta}} = \frac{3}{4}$}};
\draw [-Stealth, line width=1mm, blue] (-0.35, 0.35) -- (-0.35, 0.47); 
\draw [decorate,decoration={brace,amplitude=5pt, raise=3ex}] (0,0.5) -- (0,0.6);
\node [left] at (-0.1, 0.55) {\textcolor{orange}{$\frac{1}{4}$}};
\filldraw [blue] (0.35, 0.26) circle (1.5pt);
\draw [-Stealth, line width=1mm, blue] (0.35, 0.26) -- (0.2, 0.32);

\node[draw,text width= 2.5 cm] at (0.7, 0.9) {$\textcolor{blue}{\sigma_{\overline{\theta}}(a_{1}|\omega_{1}) = 1,}$\\ $\textcolor{blue}{\sigma_{\overline{\theta}}(a_{2}|\omega_{2}) = 1.}$};
\node[draw,text width= 2.5 cm] at (0.7, -0.24) {$\textcolor{orange}{\sigma_{\underline{\theta}}(a_{1}|\omega_{1}) = 1,}$\\ $\textcolor{orange}{\sigma_{\underline{\theta}}(a_{1}|\omega_{2}) = 1.}$};
\node[draw,text width= 6.6 cm] at (1.55, 0.3) {$\textcolor{purple}{\sigma^{BP}(a_{1}|\omega_{1}) = 1},$\\ $\textcolor{purple}{\sigma^{BP}(a_{1}|\omega_{2}) = 1/4, \sigma^{BP}(a_{2}|\omega_{2}) = 3/4.}$};
\end{tikzpicture}
\caption{Construction of the ambiguous experiment}
\label{fig:intro_example_2}
\end{figure}
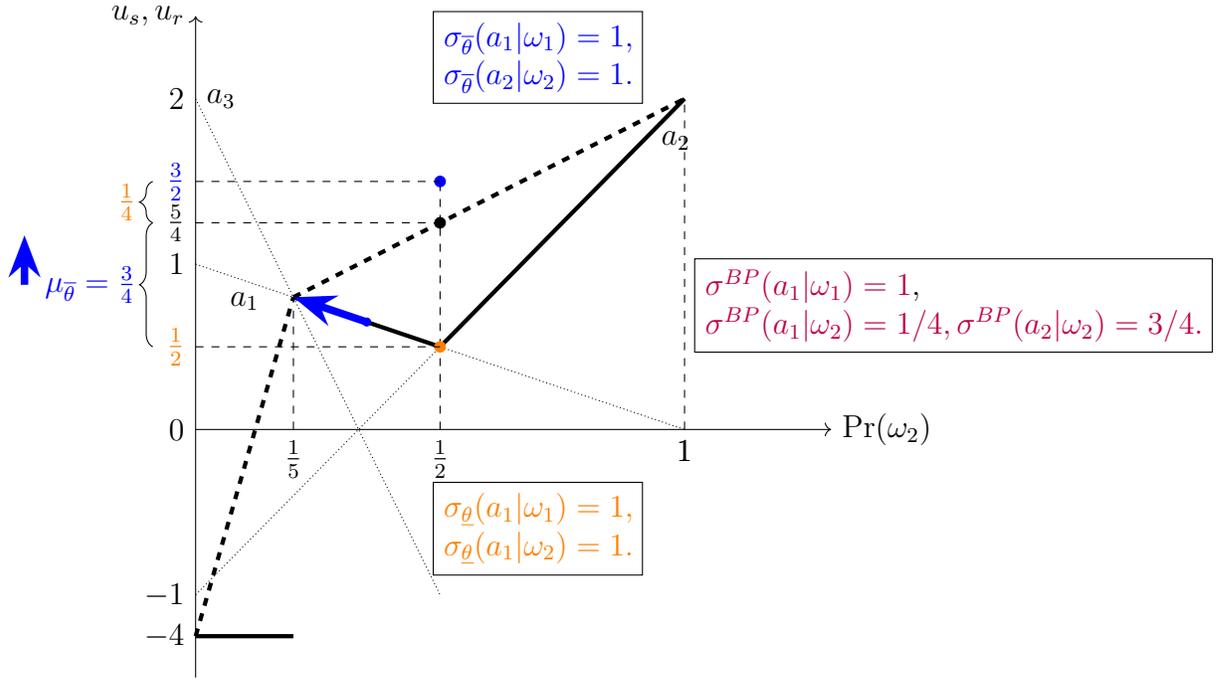

We now illustrate (see Figure \ref{fig:intro_example_2}) how ambiguous experiments can benefit an ambiguity-neutral (SEU) sender in the example.\footnote{Similar arguments remain valid as long as the sender is not too ambiguity averse. See the discussion of sufficient conditions for the sender to benefit from ambiguous communication in Section \ref{sec_improvement_binary_ambiguous_communication}, including a demonstration that the sender may benefit even if they are more ambiguity averse than the receiver.} Consider an ambiguous experiment such that one of two experiments $\sigma_{\underline{\theta}}$ and $\sigma_{\overline{\theta}}$ (see the figure) could be generating the messages, where $\theta$ is used to index experiments in the support of $\mu$. The experiment $\sigma_{\underline{\theta}}$ is uninformative, while the experiment $\sigma_{\overline{\theta}}$ is fully informative. The interpretation of the message $a_2$ is unambiguous:  the receiver learns that the state is $\omega_2$. The interpretation of the message $a_1$ is, however, ambiguous:  either it means that the state is $\omega_1$ (if $\sigma_{\overline{\theta}}$ generated the message)  or it is uninformative (if $\sigma_{\underline{\theta}}$ generated the message). The associated expected payoff profiles for these two experiments assuming action recommendations are obeyed are $(u_s, u_r)= ( 3/2,3/2)$ for $\sigma_{\overline{\theta}}$ and $(u_s, u_r)= (1/2 ,1/2)$ for $\sigma_{\underline{\theta}}$. Observe that $(3/4) \sigma_{\overline{\theta}}+ (1/4) \sigma_{\underline{\theta}} = \sigma^{BP}$ --  we call such a configuration a splitting of $\sigma^{BP}$. Thus, if an ambiguity-neutral sender assigns $\mu$-probability more than $3/4$ to the experiment $\sigma_{\overline{\theta}}$ generating the messages, their expected payoff would be strictly higher than their Bayesian persuasion payoff.

We now argue that the sender can simultaneously assign $\mu$-probability above $3/4$ to $\sigma_{\overline{\theta}}$ \emph{and} guarantee that an ambiguity-averse receiver is willing to obey the action recommendations. If the receiver were ambiguity-neutral, they would evaluate the ambiguous experiment using $\mu$, just like the sender. Ambiguity aversion, however, leads the receiver to behave as if they shift some of the weight on $\sigma_{\overline{\theta}}$ to the lower expected payoff experiment $\sigma_{\underline{\theta}}$. We refer to this ``as if'' weighting as the receiver's \emph{effective measure} on the experiments. This relative pessimism of the effective measure reflects the value an ambiguity-averse receiver places on robustness across probabilities over the experiments. As a result, the receiver chooses a strategy whose expected payoff is less sensitive across experiments than the ambiguity-neutral best response. The more ambiguity averse the receiver, the stronger this relative pessimism.

Since the receiver is obedient when facing $\sigma^{BP}$, the receiver is obedient when their effective measure places weight $3/4$ on $\sigma_{\overline{\theta}}$. In the figure, the thick arrow moving along the sender's indirect utility curve indicates the movement of the receiver's effective weight on $\sigma_{\overline{\theta}}$ towards $3/4$ from below as the probability $\mu_{\overline{\theta}}$ on $\sigma_{\overline{\theta}}$ increases above $3/4$ (along the thick arrow next to $\mu_{\overline{\theta}}$). Thus, the ambiguous communication strategy allows the sender to place more weight on the better experiment $\sigma_{\overline{\theta}}$, while maintaining obedience, than would be possible with unambiguous communication. This is how ambiguous communication provides benefits.

An important observation is that $\sigma_{\overline{\theta}}$ and $\sigma_{\underline{\theta}}$ are Pareto-ranked -- both players prefer $\sigma_{\overline{\theta}}$ assuming action recommendations are followed. If they were not Pareto-ranked, then ambiguity aversion would push the receiver's effective measure in a direction that would \emph{hurt} rather than help the sender -- for example, if the receiver thought instead that $\sigma_{\underline{\theta}}$ were the better experiment this would cause their effective weight on $\sigma_{\overline{\theta}}$ to \emph{exceed} $\mu_{\overline{\theta}}$, leading any $\mu_{\overline{\theta}} \geq 3/4$ to cause the receiver to defect to the action $a_3$ when the message is $a_1$. The ambiguity generated using Pareto-ranked experiments serves to beneficially misalign the (endogenous) effective beliefs of the sender and receiver.\footnote{A player's effective measure is endogenous in the sense that it depends on their strategy's expected payoff profile across the possible experiments in addition to their ambiguity aversion. This endogeneity  of the players' misalignment distinguishes our model from models with exogenously fixed heterogeneous priors, e.g., \cite{alonso2016bayesian}, \cite{laclau2017public} and \cite{galperti2019persuasion}.}

The remainder of the paper is organized as follows. The next section presents the model and two key preliminary results -- a revelation principle and an incentive-compatibility lemma. Main results are in Sections \ref{sec_ambiguous_communication_benefit_sender} through \ref{sec_properties_optimal_ambiguous_communication}. Section \ref{sec_discussion} contains further discussion. \textcolor{black}{Proofs and additional materials are in Appendix.}

\section{Modeling Persuasion with Ambiguous Communication}\label{sec:model}
\subsection{The Model} 
We consider a persuasion game between a sender and a receiver, where the sender can choose  \emph{ambiguous} experiments if they wish. There is a finite set $\Omega$ of payoff-relevant states $\omega$, with common full-support prior $p \in \Delta(\Omega)$. There is a finite set $A$ of actions the receiver can choose from. If the receiver chooses $a \in A$, the payoff to the sender (resp., receiver) is $u_s(a,\omega) \in \R$ (resp., $u_r(a,\omega) \in \R$), when the state is $\omega$. A \emph{statistical experiment} is a map $\sigma$ from $\Omega$ to $\Delta(M)$ where $M$ is a finite set of messages, and we write $\sigma(m|\omega)$ for the probability of $m$ given $\omega$. Let $\Sigma$ denote the set of all statistical experiments.

We model an \emph{ambiguous experiment} as a finite-support distribution over statistical experiments, and denote it by $\mu \in \Delta_{f}(\Sigma)$, where the uncertainty described by $\mu$ about which statistical experiment is generating the messages is viewed as ambiguity by the players. We index the experiments in the support of $\mu$ (denoted $\text{supp}(\mu)$), and use $\theta$ as a generic index. Thus $\sigma_{\theta}$ is an experiment in $\text{supp}(\mu)$ and $\mu_{\theta}$ is the probability $\mu$ assigns to $\sigma_{\theta}$. The sender may choose any ambiguous experiment.\footnote{It is without loss to assume that all statistical experiments in $\text{supp}(\mu)$ share the same message space.}\textsuperscript{,}\footnote{The sender chooses and commits to the ambiguous experiment \emph{before} the statistical experiment and $\omega$ are realized. This parallels standard Bayesian persuasion where the sender chooses and commits to $\sigma$ before $\omega$ is realized. See Section \ref{sec_cheap_talk} for a discussion of how things would change under the alternative assumption that it is common knowledge that the sender privately learns the realization of $\mu$ before committing to it.} Henceforth, whenever we use the term ``experiment'' without a modifier, it refers to a standard, unambiguous statistical experiment. Our model enlarges the sender's strategy space relative to the standard Bayesian persuasion model in the sense that any experiment $\sigma$ can be viewed as an ambiguous experiment with a degenerate $\mu$ concentrated on $\sigma$, i.e., $\mu = \delta_{\sigma}$. Of special interest in some of our later constructions are \emph{binary ambiguous experiments}, those in which $\text{supp}(\mu)$ is a pair of experiments. For $\kappa \in (0,1)$ and $\overline{\sigma}, \underline{\sigma} \in \Sigma$, we denote the binary ambiguous experiment $\kappa \delta_{\overline{\sigma}} + (1-\kappa)\delta_{\underline{\sigma}}$ by $\mu^{((\overline{\sigma}, \underline{\sigma}), \kappa)}$.

\begin{remark}\label{foundations_ambiguous_experiment}\emph{What is a possible foundation for the availability of these ambiguous experiments $\mu$ that are treated as commonly perceived ambiguity? One possibility is that the desired $\mu$ results from the sender conditioning the statistical experiment on the realization of a suitably chosen finite partition of a \emph{source of ambiguity}. The idea is that the sender commits to allowing the realization of external, payoff-irrelevant ambiguous events to determine which experiment will generate the messages. Formally, a source of ambiguity is a continuum, $\mathcal{A}$, of payoff-irrelevant ambiguous states, $\alpha$, along with a continuous $\rho \in \Delta(\mathcal{A})$ that represents the common view of uncertainty over $\mathcal{A}$ that an ambiguity-neutral (SEU) player would use to compute their expected utility. Let $B$ be any finite partition of $\mathcal{A}$ and define $\tilde{\mu} \in \Delta(B)$ by $\tilde{\mu}(b) \equiv \int_{\alpha \in b}\rho(\alpha)d\alpha$ for each $b \in B$. By richness of $\mathcal{A}$ and continuity of $\rho$, \emph{any} finitely-supported distribution $\mu$ over experiments can be induced by the choice of some finite partition $B$ together with some mapping from $B$ to experiments.} 

\emph{More informally, and more generally, our model of ambiguous experiments is intended to capture any setting in which the sender's instructions give rise to commonly-viewed ambiguity about which experiment will generate the messages (i.e., about the likelihood-implication of the messages) and in which, by varying their choice of instructions, the sender has the ability to influence that common view. There might be some tension between the sender's ability to finely tune the uncertainty about which experiment will be used while still viewing that uncertainty as ambiguity. However, we think that both features arise, for example, in settings where the sender instructs a third party to generate the message, while leaving them with some discretion to adapt the experiment based on their own private information. By varying the extent and direction of this discretion, the sender controls both the degree and nature of the resulting ambiguity. \qedsymbol}
\end{remark}

We analyze the receiver's behavior from the perspective of their ex-ante preferences, that is, we assume that the receiver observes the sender's choice of $\mu$ and then chooses a strategy $\tau: M \rightarrow \Delta(A)$ to maximize the receiver's ex-ante preference.\footnote{For compactness, this notation suppresses the allowed dependence of $\tau$ on $\mu$.} We further assume that $\tau$ is indeed carried out after the message is observed. This follows from either assuming that the receiver can commit to $\tau$, or that the receiver is dynamically consistent.\footnote{For a simple receiver's updating rule that guarantees dynamic consistency, see Remark \ref{rem_dynamic_consistency} in Section \ref{sec_incentive_compatibility}.} Recall from the Introduction that this allows us to shut down the channel of dynamic \emph{in}consistency that is at work in nearly all previous literature on mechanism or information design with ambiguity. We refer the interested reader to Section \ref{sec_comparison_BLL} for more discussion on this point.

We write $u_i(\sigma,\tau)$ for the expected payoff of player $i \in \{s,r\}$ when the experiment generating the message $m$ is $\sigma$ and the receiver's strategy is $\tau$, that is, 
\begin{equation}\label{eq:expected_payoff}
u_i(\sigma,\tau) =\sum_{\omega,m,a}p(\omega)\sigma(m|\omega)\tau(a|m) u_i(a,\omega). 
\end{equation}

To isolate and clarify the role of intentional ambiguous communication, we work in a stylized environment where ambiguity is not payoff-relevant unless it becomes so through the sender's strategic choice of a non-degenerate $\mu$. In particular, the payoff-relevant events $\omega$ and any randomization over messages induced by an experiment are viewed as unambiguous.\footnote{Section \ref{sec_prior_ambiguity} discusses extensions to environments with pre-existing ambiguity about $\omega$.} It follows that a message $m$ is viewed as ambiguous by the sender and receiver only if experiments in the support of $\mu$ generate $m$ with different positive likelihoods. These different likelihoods may lead the expected payoff $u_{i}(\sigma_{\theta},\tau)$ to vary with $\theta$.   

How does such ambiguity enter the sender's and receiver's preferences? We assume, as in the smooth ambiguity model \citep{klibanoff2005smooth}, that player $i$ evaluates the strategy profile $(\mu,\tau)$ as
\begin{equation}\label{eq:smooth_ambiguity}
U_{i}(\mu,\tau) = \phi_{i}^{-1} \left( \sum_{\theta}\mu_{\theta}\phi_i(u_i(\sigma_{\theta},\tau)) \right), 
\end{equation}
where $\phi_i: \mathbb{R} \rightarrow \mathbb{R}$ is a weakly concave and strictly increasing and differentiable function. An affine $\phi_{i}(\cdot)$ corresponds to ambiguity neutrality, in which case the preferences reduce to SEU with belief $\mu$. Greater concavity of $\phi_{i}(\cdot)$ corresponds to greater ambiguity aversion. We do not consider ambiguity loving behavior, as this would build-in a direct  preference benefit from ambiguous communication, while with ambiguity aversion, ambiguous communication can only be valuable if it has a strategic benefit.

When the expected payoff $u_{i}(\sigma_{\theta},\tau)$ varies with $\theta$, an ambiguity-averse player $i$ responds to $\mu$ \emph{as if} they are ambiguity neutral and using an \emph{effective} measure over $\Sigma$ that is more pessimistic than $\mu$. While \eqref{eq:smooth_ambiguity} delivers a particularly tractable formula for this effective measure (Section \ref{sec_incentive_compatibility}), the key qualitative properties of comparative pessimism and endogeneity would apply to the effective measure under many models of ambiguity averse preferences. 

The following maxmin expected utility (MEU) objective can be seen as an appropriate limit of \eqref{eq:smooth_ambiguity} as ambiguity aversion goes to infinity \citep[Proposition 3]{klibanoff2005smooth}: 
\begin{equation*}
U^{MEU}_{i}(\mu,\tau) =\min_{\sigma_{\theta} \in \text{supp}(\mu)}u_i(\sigma_{\theta},\tau).
\end{equation*}
Such an MEU receiver together with an ambiguity-neutral sender is considered in \textcolor{black}{Appendix \ref{sec_maxmin}}.\footnote{Key differences from our main analysis are that (1) the sender never needs more than binary ambiguous experiments to approach their optimal payoff, and (2) the sender can approach their optimal payoff without providing the receiver any benefits from communication, and this may lead the receiver's payoff to drop discontinuously when passing to the MEU limit.} 

Writing $BR(\mu)$ for the set of best replies of the receiver (i.e., the maximizers of $U_r(\mu, \tau)$ with respect to $\tau$), the sender's problem is:
\begin{equation}\label{eq:sender_problem}
\begin{array}{l}
\max_{(\mu,\tau)} U_{s}(\mu,\tau), \\
\text{subject to\;} \tau \in BR(\mu).
\end{array}
\tag{$\mathcal{P}$}
\end{equation}

Observe that the sender's Bayesian persuasion problem \citep{kamenica2011bayesian} corresponds to the special case of our model where the sender is restricted to choosing an experiment:
\begin{equation}\label{eq:bp_problem}
\begin{array}{l}
\max_{(\sigma,\tau)}u_s(\sigma,\tau),\\
\text{subject to\;} \tau \in br(\sigma),
\end{array}
\tag{$\mathcal{P}^{BP}$}
\end{equation}
where $br(\sigma)$ denotes the set of best replies to $\sigma$, i.e., the maximizers of $u_r(\sigma,\tau)$ with respect to $\tau$. Let $u_{s}^{BP}$ denote the value of \eqref{eq:bp_problem}, i.e., the sender's payoff at a solution to \eqref{eq:bp_problem}.  

\begin{definition}\label{def_ambiguity_strictly_benefits}
\textbf{Ambiguous communication benefits the sender} if the value of \eqref{eq:sender_problem} is strictly higher than $u_s^{BP}$. 
\end{definition}

Our analysis will focus on when and how ambiguous communication may benefit the sender compared to the standard, unambiguous case of Bayesian persuasion. We also present results on optimal persuasion with ambiguous communication (the solution to \eqref{eq:sender_problem}) and its properties. 

We next present two preliminary results -- a revelation principle and a characterization of incentive compatibility for ambiguous experiments -- that play a key role in our analysis. 

\subsection{A Revelation Principle} 

\begin{definition}\label{def_canonical_experiment}
An ambiguous experiment $\mu$ is \textbf{canonical} if $M = A$. 
\end{definition}

We write $\tau^* : A \rightarrow \Delta(A)$ for the receiver's obedient strategy, that is, $\tau^*(a|a)=1$ for all $a$. We will refer to any canonical ambiguous experiment that induces such obedience as itself obedient. 
\begin{definition}\label{def_obedient_experiment}
A canonical ambiguous experiment $\mu$ is \textbf{obedient} if $\tau^{*} \in BR(\mu)$.
\end{definition}

We start with a preliminary observation: a revelation principle holds -- for payoff purposes, it is without loss of generality to restrict attention to canonical and obedient ambiguous experiments.

\begin{proposition}\label{rev-principle}
For any $(\mu, \tau)$ such that $\tau \in BR(\mu)$, there exists a canonical and obedient ambiguous experiment $\mu^{*}$ such that $U_{i}(\mu, \tau) = U_{i}(\mu^{*}, \tau^{*})$ for all $i \in \{s,r\}$. Furthermore, for each $\sigma_{\theta} \in \text{supp}(\mu)$, $\mu(\{\sigma \in \Sigma: u_{i}(\sigma, \tau) = u_{i}(\sigma_{\theta}, \tau)\}) = \mu^{*}(\{\sigma \in \Sigma: u_{i}(\sigma, \tau^{*}) =  u_{i}(\sigma_{\theta}, \tau)\})$ for all $i \in \{s,r\}$.
\end{proposition}	

It is well-known that such a revelation principle holds in the persuasion setting without ambiguity. However, one might think of at least two reasons why the same might not be true in our environment. First, an ambiguity averse receiver might strictly prefer a mixed strategy to any pure strategy for hedging reasons in the face of ambiguity. How can the receiver's desire to mix be reconciled with the revelation principle, which states that it is without loss of generality to have the receiver play the pure strategy $\tau^{*}$? The answer is that any mixing the receiver might desire to do can always be emulated through the use of experiments that mix over action recommendations. It is the standard Bayesian persuasion assumption of sender's commitment that guarantees that this emulation is always possible. Second, dynamic inconsistency, generated by ambiguity aversion together with assumptions on updating, is the main channel leading to the failure of such a revelation principle in existing literature. As previously mentioned, we shut down this channel by modeling the receiver as maximizing $U_{r}(\mu,\tau)$, their ex-ante payoff from $(\mu,\tau)$, or, equivalently, using dynamically consistent updating (see Remark \ref{rem_dynamic_consistency} in Section \ref{sec_incentive_compatibility}).  

From here on, we restrict attention to canonical experiments, and represent incentive compatibility via obedience. When writing payoffs, $U_{i}$ and $u_{i}$, if we omit the second argument, it should be understood that $\tau = \tau^{*}$. 

Given the prominent role obedient experiments play, understanding when obedience holds is important. We next present a characterization of such incentive compatibility for ambiguous experiments.

\subsection{Incentive Compatibility and Effective Measure}\label{sec_incentive_compatibility}
We present a result we repeatedly use linking the obedience of an \emph{ambiguous} experiment to the obedience of an \emph{unambiguous} experiment derived from that ambiguous experiment. To state the result, we need the following definition:

\begin{definition}\label{def_effmeasure}
Given an ambiguous experiment $\mu$, the receiver's \textbf{effective measure} $em^{\mu} \in \Delta_{f}(\Sigma)$ is given by:
\begin{equation}\label{equ_effmeasure}
em^{\mu}_{\theta} := \frac{\mu_{\theta} \phi^{'}_r(u_r(\sigma_{\theta}))}{\sum_{\tilde{\theta}}\mu_{\tilde{\theta}} \phi^{'}_r(u_r(\sigma_{\tilde{\theta}}))}, \text{\;for all\;} \sigma_{\theta} \in \text{supp}(\mu).
\end{equation}
\end{definition}        

The effective measure $em^{\mu}$ is a probability measure with the same support as $\mu$. It is equal to $\mu$ when the receiver is ambiguity neutral (i.e., $\phi_{r}$ is affine), and is more pessimistic than $\mu$ for an ambiguity averse receiver (i.e., $\phi_{r}$ concave). Pessimism here means shifting weight toward $\sigma_{\theta}$ yielding lower expected receiver's payoffs, i.e., if $u_{r}(\sigma_{\theta}) < u_{r}(\sigma_{\theta'})$, then $em^{\mu}_{\theta}/ em^{\mu}_{\theta'} > \mu_{\theta}/\mu_{\theta'}$. Notice also that the effective measure of a given $\sigma_{\theta}$ depends on the specification of the ambiguous experiment for all $\sigma_{\theta} \in \text{supp}(\mu)$.  

The next result states that $\tau^{*}$ is the receiver's best response to the ambiguous experiment $\mu$ if, and only if, it is a best response to the experiment, $\sigma^{*}$, defined below as the convex combination of experiments with weights given by the receiver's effective measure. In such a case, we say that these experiments form a \emph{splitting} of $\sigma^{*}$ with splitting weights given by the effective measure.

\begin{lemma}\label{lem_effmeasure}
The ambiguous experiment $\mu$ is obedient if, and only if, the (unambiguous) experiment $\sigma^{*}$ is obedient, where
\begin{equation*}
\sigma^{*} = \sum_{\theta} em^{\mu}_{\theta} \sigma_{\theta}. 
\end{equation*}
\end{lemma}

Lemma \ref{lem_effmeasure} follows from the first-order conditions of the receiver's maximization problem $\max_{\tau}U_r(\mu,\tau)$, evaluated at $\tau^*$. Some intuition is that obedience will differ from the best response to the experiment $\sum_{\theta} \mu_{\theta} \sigma_{\theta}$ in that it will be better hedged against uncertainty about the weights on the experiments. In our introductory example, for instance,
\begin{align*} 
2 = u_{r}(\sigma_{\overline{\theta}}, br(\mu_{\overline{\theta}}\sigma_{\overline{\theta}} + (1 - \mu_{\overline{\theta}})\sigma_{\underline{\theta}})) & > u_{r}(\sigma_{\overline{\theta}}) = 3/2 \\ &> u_{r}(\sigma_{\underline{\theta}}) = 1/2 > u_{r}(\sigma_{\underline{\theta}}, br(\mu_{\overline{\theta}}\sigma_{\overline{\theta}} + (1 - \mu_{\overline{\theta}} )\sigma_{\underline{\theta}})) = -1, 
\end{align*} 
showing that the obedience strategy $\tau^{*}$ is hedged against the uncertainty about $\mu$ more than the strategy $br(\mu_{\overline{\theta}}\sigma_{\overline{\theta}} + (1 - \mu_{\overline{\theta}})\sigma_{\underline{\theta}})$. Thus the relative pessimism of the effective measure reflects the fact that an ambiguity-averse receiver values such hedging. 

Lemma \ref{lem_effmeasure} gives rise to the following interpretation of the receiver's effective measure: It is an ``ambiguity-neutral measure supporting obedience" in the sense that \emph{if} the receiver were ambiguity neutral, the ambiguous experiment $em^{\mu}$ would be obedient.\footnote{For the benefits of working with distributions over statistical experiments rather than distributions over posteriors, see Section \ref{sec_experiments_vs_posteriors}.} 

\begin{remark} [Dynamically consistent updating]\label{rem_dynamic_consistency}
\emph{These properties of the effective measure also give rise to an updating implementation of the receiver's ex-ante optimality. A receiver who forms their effective posterior after observing a message by updating using Bayes' rule but with the effective measure in place of $\mu$ in the Bayes' formula will be dynamically consistent. Formally, such a receiver updates $em^{\mu}_{\theta} p(\omega) \sigma_{\theta}(a|\omega)$, their effective prior over $\Sigma \times \Omega \times A$, to $\frac{em^{\mu}_{\theta} p(\omega) \sigma_{\theta}(a|\omega)}{\sum_{\tilde{\theta},\tilde{\omega}} em^{\mu}_{\tilde{\theta}} p(\tilde{\omega}) \sigma_{\tilde{\theta}}(a|\tilde{\omega})}$, their effective posterior over $\Sigma \times \Omega$ after observing message $a$. \qedsymbol}
\end{remark}

Finally, for later reference, observe that we can invert \eqref{equ_effmeasure} to express $\mu$ as a function of the effective measure it generates: 
\begin{equation}\label{equ_reverse_effmeasure}
\mu_{\theta} = \frac{em^{\mu}_{\theta}/ \phi^{'}_r(u_r(\sigma_{\theta}))}{\sum_{\tilde{\theta}}em^{\mu}_{\tilde{\theta}}/ \phi^{'}_r(u_r(\sigma_{\tilde{\theta}}))}.
\end{equation}

\section{A Necessary Condition for Ambiguous Communication to Benefit The Sender}\label{sec_ambiguous_communication_benefit_sender}

This section establishes a necessary condition for ambiguous communication \textcolor{black}{to benefit} the sender, and shows that this condition implies no benefit for the sender when the receiver has only two available actions -- a common assumption in many examples and applications in the literature. The necessary condition uses the concept of Pareto-ranked experiments and splittings into such experiments. Two experiments are \emph{Pareto-ranked} if the sender and receiver agree on their strict ranking under the assumption of obedience.

\begin{definition}\label{def_pareto_ranked_splitting}
Two experiments $\overline{\sigma}$ and $\underline{\sigma}$ are \textbf{weakly Pareto-ranked} if either the two inequalities
\begin{align}\label{equ_pareto_ranked}
u_{s}(\overline{\sigma}) \geq u_{s}(\underline{\sigma}) \text{ and } u_{r}(\overline{\sigma}) \geq u_{r}(\underline{\sigma}),
\end{align}
hold or both reversed inequalities hold. They are \textbf{Pareto-ranked} if the same holds true with strict inequalities. 

A \textbf{Pareto-ranked splitting}  of the experiment $\sigma$ is a triple $(\overline{\sigma}, \underline{\sigma}, \lambda)$ such that (i) $\lambda \overline{\sigma} + (1-\lambda)\underline{\sigma} = \sigma$, (ii) $\lambda \in (0,1)$, and (iii) \eqref{equ_pareto_ranked} holds with strict inequalities, i.e., $\overline{\sigma}$ and $\underline{\sigma}$ are Pareto-ranked. 
\end{definition}

The following theorem shows that the existence of Pareto-ranked experiments (in particular, ones having sender's expected payoffs bracketing $u_{s}^{BP}$) splitting an obedient experiment is necessary for ambiguous communication to benefit the sender: 

\begin{theorem}\label{thm_necessary}
Ambiguous communication benefits the sender only if $\phi_{r}$ is not affine and there exists a Pareto-ranked splitting, $(\overline{\sigma}, \underline{\sigma}, \lambda)$, of an obedient experiment $\hat{\sigma}$ such that $u_{s}(\overline{\sigma}) > u_{s}^{BP} > u_{s}(\underline{\sigma})$.
\end{theorem}

\begin{example}[Verifying that the introductory example satisfies this necessary condition] 
In Figure \ref{fig:intro_example_2} of the introductory example, $(\sigma_{\overline{\theta}}, \sigma_{\underline{\theta}}, 3/4)$ is a Pareto-ranked splitting of $\sigma^{BP}$ and $u_{s}(\sigma_{\overline{\theta}}) > u_{s}^{BP} > u_{s}(\sigma_{\underline{\theta}})$.
\end{example}

\begin{remark} [Not necessary for $\hat{\sigma}$ to be an optimal Bayesian persuasion experiment]\label{rem_not_necessary_BP_splitting} 
\emph{The reader might wonder if a stronger version of the theorem that requires $\hat{\sigma}$ to be an optimal Bayesian persuasion experiment holds. This is false. There are examples in which the sender benefits from ambiguous communication even though no Pareto-ranked splitting of any optimal Bayesian persuasion experiment exists (as is true, for instance, whenever all such experiments are efficient). \qedsymbol} 
\end{remark} 

The conditions in Theorem \ref{thm_necessary} are deceptively powerful: From these conditions alone, strict benefit from ambiguity can be ruled out for a simple yet important class of problems -- those in which the receiver has a binary action space. 

\begin{corollary}\label{cor_binary_action}
If the receiver has only two actions, the sender cannot benefit from ambiguous communication. 
\end{corollary} 

An important step in the proof of Theorem \ref{thm_necessary} is showing that its existence condition is equivalent to the existence of Pareto-ranked experiments, $\sigma$ and $\sigma^*$ such that:  (i) $\mathrm{supp\,} \sigma(\cdot|\omega)= \mathrm{supp\,} \sigma^*(\cdot|\omega)$ for all $\omega$, (ii) $u_s(\sigma,\tau^*) >  u_s^{BP}$, and (iii) $\tau^* \in br(\sigma^*) \setminus br(\sigma)$. Given $\overline{\sigma}$, $\underline{\sigma}$, and $\hat{\sigma}$ from Theorem \ref{thm_necessary}, letting $\sigma^{*} = \hat{\sigma}$ and $\sigma$ be a well-chosen mixture of $\overline{\sigma}$ and $\hat{\sigma}$ satisfies (i) - (iii). Conversely, given $\sigma$ and $\sigma^{*}$ satisfying (i) - (iii) there exists a Pareto-ranked splitting of the obedient $\sigma^{*}$ with $\overline{\sigma} = \sigma$ that satisfies $u_{s}(\overline{\sigma}) > u_{s}^{BP} > u_{s}(\underline{\sigma})$. Necessity is then established by the following argument: Suppose that there exists a solution $\tilde{\mu}$ to the sender's program \eqref{eq:sender_problem} that benefits the sender. Construct $\sigma$ and $\sigma^{*}$ satisfying (i) - (iii) by letting $\sigma=\sum_{\theta} \tilde{\mu}_{\theta}\tilde{\sigma}_{\theta}$ and $\sigma^*=\sum_{\theta} em^{\tilde{\mu}}_{\theta} \tilde{\sigma}_{\theta}$. 

The intuition for Corollary \ref{cor_binary_action} is as follows. From the above, ambiguous communication benefiting the sender implies the existence of a non-obedient experiment $\sigma$ that is Pareto-ranked higher than some obedient experiment $\sigma^*$. Intuitively, such an improvement for the receiver is possible only when $\sigma$ is more informative for the receiver and the benefit of this extra information outweighs the cost of not best responding. When there are only two actions, taking advantage of information requires best responding. To see this, note that not best responding implies either taking the same action always (and thus ignoring any information) or always doing the opposite of what is optimal for the receiver (which hurts more when there is more information). In contrast, when there are three or more actions, it becomes possible to have some beneficial responsiveness to information without going all the way to best responding. As we saw in the introductory example, this indeed can leave scope for possible improvements.   

\textcolor{black}{The following result identifies a necessary property of any obedient ambiguous experiment that benefits the sender.}

\begingroup
\color{black}
\begin{theorem}\label{thm_necessary_condition_existence_of_Pareto_ranked_splitting}
        If an obedient ambiguous experiment $\mu$ benefits the sender, then there exist $\sigma_{\theta}, \sigma_{\theta'} \in \text{supp}(\mu)$ such that $\sigma_{\theta}$ and $\sigma_{\theta'}$ are Pareto-ranked, with
        $u_{s}(\sigma_{\theta}) > u_{s}^{BP} \geq u_{s}(\sigma_{\theta'})$. 
    \end{theorem}
\par
\endgroup

\section{A Sufficient Condition for Ambiguous Communication to Benefit the Sender}\label{sec_improvement_binary_ambiguous_communication}

This section shows that a strengthening of the necessary conditions described in Theorem \ref{thm_necessary} is sufficient for ambiguous communication to benefit the sender. Given strict ambiguity aversion of the receiver, existence of a Pareto-ranked splitting of an optimal Bayesian persuasion experiment, $\sigma^{BP}$, is sufficient for an ambiguity-neutral sender to benefit. This sufficiency extends to ambiguity-averse senders when paired with a condition giving a tight upper bound on a measure of the sender's ambiguity aversion. The measure of ambiguity aversion is the following notion of \emph{probability premium}. 

\textcolor{black}{In general, ambiguous communication may benefit the sender even when no binary ambiguous experiment does. Appendix \ref{sec_insufficiency_binary_experiment} states this result as Proposition \ref{prop_insufficiency_binary_experiment} and provides a proof sketch.}

\begin{definition}\label{def_probability_premium}
Given $\phi$, $u$,  and experiments $\overline{\sigma}$ and $\underline{\sigma}$ such that $u(\overline{\sigma}) > u(\underline{\sigma})$, the \textbf{$((\overline{\sigma}, \underline{\sigma}), \lambda)$-probability premium} required to compensate for replacing the unambiguous experiment $\sigma^{*} := \lambda \overline{\sigma} + (1-\lambda)\underline{\sigma}$ by the binary ambiguous experiment $\mu^{((\overline{\sigma}, \underline{\sigma}), \lambda)}$, assuming obedience, is:
\begin{align*}
\rho^{\phi,u}((\overline{\sigma}, \underline{\sigma}), \lambda):=\frac{\phi(u(\sigma^{*})) - \lambda \phi(u(\overline{\sigma})) - (1-\lambda)\phi(u(\underline{\sigma}))}{ \phi(u(\overline{\sigma})) - \phi(u(\underline{\sigma}))}.
\end{align*}
\end{definition}

The probability premium $\rho^{\phi,u}((\overline{\sigma}, \underline{\sigma}), \lambda)$ is exactly the $\phi$-payoff difference between $\sigma^{*}$ and the ambiguous experiment $\mu^{((\overline{\sigma}, \underline{\sigma}), \lambda)} $, normalized to lie in $[0,1]$. This premium is non-negative under ambiguity aversion, and is zero under ambiguity neutrality. Similar notions of probability premium in the context of risk go back to at least \citet{pratt1964}. Letting
\begin{equation*}
\kappa =\lambda + \rho^{\phi,u}((\overline{\sigma}, \underline{\sigma}), \lambda) \in [0,1],
\end{equation*}
\begin{align*}
U(\mu^{((\overline{\sigma}, \underline{\sigma}), \kappa)}) &=\phi^{-1}\Big(\kappa \phi(u(\overline{\sigma}))+(1-\kappa)\phi(u(\underline{\sigma}))\Big) =\phi^{-1}\Big(\phi(u(\sigma^{*}))\Big)=u(\sigma^{*}),
\end{align*}
meaning that the premium $\rho^{\phi,u}((\overline{\sigma}, \underline{\sigma}), \lambda)$ is exactly the increase in $\kappa$ above $\lambda$ needed to make the player indifferent between the ambiguous experiment $\mu^{((\overline{\sigma}, \underline{\sigma}), \kappa)}$ and $\sigma^{*}$. Thus, assuming obedience, this premium is the smallest increase in the probability on $\overline{\sigma}$ required to compensate for exposure to the ambiguous experiment.

\begin{theorem}\label{sufficient_BP}
Let $\sigma^{BP}$ be an obedient experiment such that $u_{s}(\sigma^{BP}) = u_{s}^{BP}$. If there exists a Pareto-ranked splitting of $\sigma^{BP}$, $(\overline{\sigma}, \underline{\sigma}, \lambda)$, for which $ \rho^{\phi_{s},u_{s}}((\overline{\sigma}, \underline{\sigma}), \lambda)<\kappa-\lambda$ with 
\begin{equation}\label{equ_thm_Pareto_ranked_improvement}
\kappa = \frac{\lambda \phi_{r}'(u_{r}(\underline{\sigma}))}{\lambda \phi_{r}'(u_{r}(\underline{\sigma})) + (1-\lambda) \phi_{r}'(u_{r}(\overline{\sigma}))},  
\end{equation}
then ambiguous communication benefits the sender.

Furthermore, $\rho^{\phi_{s},u_{s}}((\overline{\sigma}, \underline{\sigma}), \lambda)$ is increasing in the sender's ambiguity aversion, and $\kappa$ given by \eqref{equ_thm_Pareto_ranked_improvement} is increasing in the receiver's ambiguity aversion.
\end{theorem}

\noindent Therefore, whenever a Pareto-ranked splitting of an optimal Bayesian persuasion experiment exists, an ambiguity-neutral sender benefits from ambiguous communication if $0<\kappa - \lambda$. The latter is equivalent to  $\phi_{r}'(u_{r}(\underline{\sigma})) > \phi_{r}'(u_{r}(\overline{\sigma}))$, i.e., the receiver is, within this range of payoffs, not everywhere ambiguity neutral, which is always satisfied if $\phi_r$ is strictly concave.

The key to Theorem \ref{sufficient_BP} is to consider the binary ambiguous experiment $\mu^{((\overline{\sigma}, \underline{\sigma}), \kappa)}$. First, observe that (\ref{equ_thm_Pareto_ranked_improvement}) is equivalent to $em^{\mu^{((\overline{\sigma}, \underline{\sigma}), \kappa)}}_{\overline{\theta}}=\lambda$, ensuring that the obedience of $\sigma^{BP}$ extends to the ambiguous experiment $\mu^{((\overline{\sigma}, \underline{\sigma}), \kappa)}$ by Lemma \ref{lem_effmeasure}. Second, the inequality $\rho^{\phi_{s},u_{s}}((\overline{\sigma}, \underline{\sigma}), \lambda)<\kappa-\lambda$ characterizes exactly when the sender strictly prefers using the ambiguous experiment $\mu^{((\overline{\sigma}, \underline{\sigma}), \kappa)}$ rather than $\sigma^{BP}$. The source of the economic gain from this ambiguous communication is the greater use, as measured by $\kappa-\lambda$, of the Pareto-better experiment $\overline{\sigma}$ while maintaining obedience. This gain has to be netted-off against the probability premium, which encapsulates the cost due to the player's own ambiguity aversion of the exposure to ambiguity from the ambiguous experiment.

Notice that the gain here is not a story of a less ambiguity averse sender gaining by insuring a more ambiguity averse receiver. In particular, it is possible to satisfy the sufficient conditions in Theorem \ref{sufficient_BP} even when the sender is more ambiguity averse than the receiver, as measured by their respective probability premiums. To see this, suppose that $\phi_{r}$ is strictly concave so that $\kappa>\lambda$ and let $(\overline{\sigma}, \underline{\sigma}, \lambda)$ be a Pareto-ranked splitting of $\sigma^{BP}$. Algebra and concavity can be used to show that $\rho^{\phi_{r},u_{r}}((\overline{\sigma}, \underline{\sigma}), \lambda)<\kappa-\lambda$, and thus for any sender such that $\rho^{\phi_{s},u_{s}}((\overline{\sigma}, \underline{\sigma}), \lambda)\in (\rho^{\phi_{r},u_{r}}((\overline{\sigma}, \underline{\sigma}), \lambda), \kappa - \lambda)$  the sufficient conditions in Theorem \ref{sufficient_BP} are satisfied while the sender has the higher probability premium. For instance, in the introductory example, setting $\phi_{r}(x)=(x+4)^{1/2}$ and $\phi_{s}(x)=(x+4)^{1/4}$, and using the Pareto-ranked splitting $(\sigma_{\overline{\theta}}, \sigma_{\underline{\theta}}, 3/4)$ of $\sigma^{BP}$ as in Figure \ref{fig:intro_example_2}, yields $\rho^{\phi_{r},u_{r}}((\sigma_{\overline{\theta}}, \sigma_{\underline{\theta}}), 3/4) <0.01  < \rho^{\phi_{s},u_{s}}((\sigma_{\overline{\theta}}, \sigma_{\underline{\theta}}), 3/4) < \kappa - 3/4 \approx 0.018$, so that the sufficient conditions in Theorem \ref{sufficient_BP} are satisfied while the sender is more ambiguity averse than the receiver.

The comparative statics in the final part of the theorem deliver robustness in several respects, including to the sender underestimating the extent of ambiguity aversion of the receiver. Formally, let $(\overline{\sigma}, \underline{\sigma}, \lambda)$ be a Pareto-ranked splitting of $\sigma^{BP}$ satisfying the conditions in Theorem \ref{sufficient_BP}. Then ambiguous communication using $\mu^{((\overline{\sigma}, \underline{\sigma}), \kappa)}$ also benefits all less ambiguity averse senders. If the receiver is more ambiguity averse than stated, then if $\underline{\sigma}$ is obedient, $\mu^{((\overline{\sigma}, \underline{\sigma}), \kappa)}$ continues to benefit the sender (while otherwise only $\kappa$ may need adjustment to ensure a benefit).

\textcolor{black}{The next result extends the construction to any obedient experiment and also characterizes the receiver's gain.}

\begingroup
\color{black}
\begin{theorem}\label{general_improvement}
    Let $\sigma^{*}$ be an obedient experiment. Suppose that $(\overline{\sigma}, \underline{\sigma}, \lambda)$ is a Pareto-ranked splitting of $\sigma^{*}$ satisfying $u_s(\overline{\sigma}) > u_s(\underline{\sigma})$. 
    The binary ambiguous experiment $\mu = \mu^{((\overline{\sigma}, \underline{\sigma}), \kappa)}$, where 
    \begin{equation}\label{equ_wp_general_improvement}
        \kappa = \frac{\lambda \phi_{r}'(u_{r}(\underline{\sigma}))}{\lambda \phi_{r}'(u_{r}(\underline{\sigma})) + (1-\lambda) \phi_{r}'(u_{r}(\overline{\sigma}))},  
    \end{equation}
     satisfies the following properties: 
    \begin{enumerate}[(i)]
    \item $\mu$ is obedient, 
    \item $U_{r}(\mu) > u_{r}(\sigma^{*})$ if, and only if, $\kappa>\lambda$, 
    \item $U_{s}(\mu) > u_{s}(\sigma^{*})$ if, and only if, the sender's $((\overline{\sigma}, \underline{\sigma}), \lambda)$-probability premium is strictly less than $\kappa - \lambda$.
    \end{enumerate}
    Furthermore, the sender's $((\overline{\sigma}, \underline{\sigma}), \lambda)$-probability premium is increasing in the sender's ambiguity aversion, and $\kappa$ is increasing in the receiver's ambiguity aversion.
    \end{theorem}
    
    That a $\kappa$ satisfying \eqref{equ_wp_general_improvement} ensures that the obedience of $\sigma^{*}$ extends to the binary ambiguous experiment $\mu$ as in (i) is a straightforward consequence of Lemma \ref{lem_effmeasure}. The necessary and sufficient conditions in (ii) for the receiver to be better off when the sender communicates ambiguously using $\mu$ rather than unambiguously using $\sigma^{*}$ require some elaboration. First, the condition $\kappa > \lambda$ is equivalent to  $\phi_{r}'(u_{r}(\underline{\sigma})) > \phi_{r}'(u_{r}(\overline{\sigma}))$, i.e., the receiver is, within this range of payoffs, not everywhere ambiguity neutral. In particular, this condition is always satisfied if $\phi_r$ is strictly concave. Second, the condition $\kappa > \lambda$ can be shown to be equivalent to the receiver's $((\overline{\sigma}, \underline{\sigma}), \lambda)$-probability premium being strictly less than $\kappa - \lambda$, which, by Lemma \ref{lem_prob-premium_improvement}, characterizes when the receiver is better off under $\mu$ than under $\sigma^{*}$. The necessary and sufficient conditions in (iii) for $\mu$ to be better for the sender than $\sigma^{*}$ follow directly from Lemma \ref{lem_prob-premium_improvement}. For an ambiguity neutral sender, the probability premium is zero, and thus the condition in (iii) reduces to $\kappa > \lambda$, as in (ii). Thus, for an ambiguity neutral sender facing a strictly ambiguity averse receiver, the ambiguity introduced in $\mu$ improves on $\sigma^*$ for both sender and receiver.\footnote{\textcolor{black}{Though our result (and proof) is global in the sense that they apply to any Pareto-ranked splitting, Appendix \ref{sec_local_argument_receiver_gain} provides a local illustration that may be helpful. We thank Dilip Abreu for suggesting this line of reasoning.}}
\par
\endgroup

\textcolor{black}{The following result gives robustness properties for binary ambiguous experiments that benefit the sender.}

\begingroup
\color{black}
\begin{theorem}\label{thm_receiver_more_ambi_averse}
    Suppose the binary ambiguous experiment $\mu = \mu^{((\overline{\sigma}, \underline{\sigma}), \kappa)}$, with $\kappa \in (0,1)$, is obedient and benefits the sender (compared to $u_{s}^{BP}$). Label the experiments so that $u_s(\overline{\sigma}) > u_s(\underline{\sigma})$. 
    Then: 
    \begin{enumerate}
        \item[(i)] $\mu$ also benefits all less ambiguity averse senders, and
        \item[(ii)] for any more ambiguity averse receiver, there exists some $\tilde{\kappa} \in (0,1)$ such that $\tilde{\mu} = \mu^{((\overline{\sigma}, \underline{\sigma}), \tilde{\kappa})}$ benefits the sender (and all less ambiguity averse senders), and
        \item[(iii)] if $\underline{\sigma}$ is obedient, then $\tilde{\kappa}$ in (ii) can be set equal to $\kappa$. 
    \end{enumerate}
    \end{theorem}
\par
\endgroup

We next study a class of persuasion games where Theorem \ref{sufficient_BP} is particularly powerful.

\section{Belief-Threshold Games}\label{sec_threshold_games}

In this section, we characterize when ambiguous communication benefits an ambiguity-neutral sender in a class of persuasion games called \emph{belief-threshold games} \citep{aybas2026}. This characterization takes the form of a simple and direct condition on payoffs that is easy to satisfy when the receiver has more than two actions (and, of course, is never satisfied when the receiver has only two actions).

In a belief-threshold game, there are $n+1$ states and actions indexed by $i = 0, \ldots, n$. For $i \geq 1$, action $a_{i}$ maximizes the receiver's expected payoff if and only if the probability of state $\omega_{i}$ exceeds a threshold $T_{i} > 1/2$. When none of these thresholds are met, action $a_{0}$, called the default action, maximizes the receiver's expected payoff. The sender's payoff is state-independent and ranks $a_{0}$ below all the other actions. Formally:

\begin{definition}\label{def:belief_threshold_game}
Fix $n \geq 1$, and let $\Omega = \{\omega_{0}, \omega_{1}, \ldots, \omega_{n}\}$ and $A = \{a_{0}, a_{1}, \ldots, a_{n}\}$. A persuasion game is a \textbf{belief-threshold game} if the following conditions hold:
\begin{enumerate}[(i)]
\item The sender's payoff is state-independent and satisfies
\begin{align}\label{equ:belief_threshold_sender_payoff}
u_{s}(a_{i}, \omega) = v_{i} > v_{0} = u_{s}(a_{0}, \omega) \quad \text{for all } i = 1, \ldots, n \text{ and } \omega \in \Omega.
\end{align} 
\item There exist unique thresholds $T_{1}, \ldots, T_{n} \in (1/2, 1)$ such that, for every $q \in \Delta(\Omega)$, 
\begin{align*}
a_{i} \in \mathop{\arg\max}_{a\in A}\sum_{\omega\in\Omega}q(\omega)u_r(a,\omega) \iff q(\omega_{i}) \geq T_{i}, \text{ for all } i = 1, \ldots, n,
\end{align*}
and 
\begin{align*}
a_{0} \in \mathop{\arg\max}_{a\in A}\sum_{\omega\in\Omega}q(\omega)u_r(a,\omega) \iff q(\omega_{i}) \leq T_{i} \text{ for all } i = 1, \ldots, n.
\end{align*}
\end{enumerate} 
\end{definition}
\noindent Without loss of generality, we normalize $v_{0}$ to $0$ throughout.

Such preferences are well-suited for a variety of economic settings studied in the communication literature. Some of the seminal motivating examples in the persuasion literature, including the prosecutor-judge example in \citet{kamenica2011bayesian}, are belief-threshold games. The combination of sender state-independence (called \emph{transparent motives} following \citet{lipnowski2020}) and receiver threshold preferences models settings where the receiver seeks to match the action to an underlying state, while the sender benefits from inducing any action other than the default. These include seller–buyer and lobbyist–legislator interactions. For example, \citet{chakraborty2010} consider a seller persuading a buyer to purchase one of several products, with no purchase as the default action, while \citet{lipnowski2020} study advice-seeking environments, for instance, one in which an industry lobbyist is trying to persuade a politician to adopt one of several policy reforms rather than maintain the status quo.

In the seller-buyer application, the state $\omega$ represents characteristics of the products, with $\omega_{i}$ the characteristics making product $i$ the best option for the buyer, while $\omega_{0}$ are those indicating that the buyer should walk away. The seller makes more profit from any sale than no sale. In lobbyist-legislator, the state $\omega$ represents information about the effects of each policy on the legislator's constituents, with $\omega_{i}$ indicating that policy $i$ is the best option for them, while $\omega_{0}$ indicates that the status quo is the best option. The lobbyist is paid a (possibly policy-dependent) bonus if they convince the legislator to support any of the new policies, and thus has state-independent payoffs that rank the status quo last.

The following lemma establishes that a simple parametric representation of the receiver's payoffs characterizes belief-threshold games.

\begin{lemma}\label{lem:belief_threshold_payoff_representation}
A persuasion game is a belief-threshold game if and only if the sender's payoffs satisfy \eqref{equ:belief_threshold_sender_payoff} and the receiver's payoffs satisfy     
\begin{align}\label{equ:belief_threshold_payoff_representation}
u_{r}(a_{i},\omega_{j}) - u_{r}(a_{0}, \omega_{j}) = k_{i}(\mathbf{1}_{i = j} - T_{i}), \quad \text{ for all } j \in \{0, 1, \ldots, n\},
\end{align}
for some $k_{i} > 0$, $T_{i} \in (1/2, 1)$, $i \in \{1, \ldots, n\}$. 
\end{lemma}

Our next result establishes a simple-to-check and generically satisfied necessary and sufficient condition on payoffs for ambiguous communication to benefit an ambiguity-neutral sender in belief-threshold games for which the prior satisfies $\sum_{i=1}^{n}p(\omega_{i})/T_{i} < 1$. The proof additionally establishes that in these games, the sufficient conditions from Theorem \ref{sufficient_BP} are also necessary, and the necessary conditions from Theorem \ref{thm_necessary} are also sufficient.

\begin{theorem}\label{thm:belief_threshold_iff}
Fix any belief-threshold game with $\sum_{i=1}^{n}p(\omega_{i})/T_{i} < 1$ and suppose $\phi_{r}$ is strictly concave. Then ambiguous communication benefits an ambiguity-neutral sender if and only if there exist $i, j \in \{1, \ldots, n\}$ such that
\begin{align}\label{equ:belief_threshold_iff}
\frac{v_{i}}{k_{i}T_{i}} \neq \frac{v_{j}}{k_{j}T_{j}}.
\end{align}
\end{theorem}

Condition \eqref{equ:belief_threshold_iff} requires that the ratio of the sender's payoff from $a_{i}$ compared to $a_{0}$ and the receiver's payoff from the same in any state where $a_{i}$ is not optimal is not constant across all non-default actions. In the special case where the thresholds and $k$'s are the same for all actions, \eqref{equ:belief_threshold_iff} reduces to the condition that the sender has a strict preference between two non-default actions. Similarly, if the sender is indifferent between all non-default actions, \eqref{equ:belief_threshold_iff} reduces to the condition that the receiver's $k_{i}T_{i}$ varies across non-default actions. In general, \eqref{equ:belief_threshold_iff} is generically satisfied when $n \geq 2$ in the sense that it is satisfied for all but a measure-zero set of payoff parameters for this class of games.\footnote{For a similar class of games, \citet{arieli2024} show that $\sigma^{BP}$ is generically inefficient when $n\geq 2$. This inefficiency is a prerequisite for the existence of a Pareto-ranked splitting of $\sigma^{BP}$.}

What is the role of the condition $\sum_{i=1}^{n}p(\omega_{i})/T_{i} < 1$? First, observe that it ensures that the prior does not reach the threshold for any $a_{i}, i = 1, \ldots, n$, and thus, absent persuasion the receiver would take the worst action for the sender, $a_{0}$, creating room for the sender to possibly gain from persuasion. Second, and the reason why it cannot be replaced by the weaker condition that the receiver would take $a_{0}$ absent persuasion, the proof shows that it ensures that $\sigma^{BP}$ recommends all actions with positive probability in state $\omega_{0}$ (and recommends $a_{i}$ with probability one in state $\omega_{i}$ for $i > 0 $). This property is what allows the condition in Lemma \ref{prop:inefficiency_improvement} ensuring the existence of a Pareto-ranked splitting to be checked using the payoff-difference vectors from state $\omega_{0}$ and hence allows the necessary and sufficient condition to be expressed as the simple ratio condition \eqref{equ:belief_threshold_iff}. 

\begingroup
\color{black}
\begin{example}[The prior restriction in Theorem \ref{thm:belief_threshold_iff}]\label{ex:belief_threshold_prior_condition}
To see that one cannot relax the condition $\sum_{i=1}^{n}p(\omega_{i})/T_{i} < 1$ while still preserving the result, consider the following example. Let $n = 2$, $T_{1} = T_{2} = 3/4$, $p(\omega_{0}) = 1/4$, $p(\omega_{1}) = p(\omega_{2}) = 3/8$, $v_{1} = 5$, $v_{2} = 3$, $u_{r}(a_{0}, \omega) = 2$ for all $\omega$, $k_{1} = 2$, and $k_{2} = 1$. Then by Lemma \ref{lem:belief_threshold_payoff_representation}, the remaining receiver's payoffs are
\begin{align*}
&u_{r}(a_{1}, \omega_{0})=1/2, \quad u_{r}(a_{1}, \omega_{1}) = 5/2, \quad u_{r}(a_{1}, \omega_{2}) = 1/2,\\
&u_{r}(a_{2}, \omega_{0}) = 5/4, \quad u_{r}(a_{2}, \omega_{1}) = 5/4, \quad u_{r}(a_{2}, \omega_{2}) = 9/4.
\end{align*}
Observe that $\sum_{i=1}^{n}p(\omega_{i})/T_{i} = 1$, and the receiver would take $a_{0}$ absent persuasion. The sender cannot benefit from ambiguous communication, despite the fact that $\frac{v_{1}}{k_{1}T_{1}} = 10/3 \neq 4 = \frac{v_{2}}{k_{2}T_{2}}$.
\end{example}
\par
\endgroup

How can Theorem \ref{thm:belief_threshold_iff} be extended to the case of an ambiguity-averse sender? First, the ``only if'' direction of the theorem goes through unchanged, as the argument makes no use of the sender's ambiguity neutrality. Second, the part of the ``if" direction establishing the existence of a Pareto-ranked splitting of $\sigma^{BP}$ also goes through unchanged. The only adjustments are needed in ensuring that the splitting satisfies the inequality $\rho^{\phi_{s},u_{s}}((\overline{\sigma}, \underline{\sigma}), \lambda)<\kappa-\lambda$ from the sufficient conditions in Theorem \ref{sufficient_BP}. \textcolor{black}{Since $\kappa$ increases with the concavity of $\phi_r$ and the other terms are unaffected by $\phi_r$, fixing the splitting and $\phi_s$, there always exists a coefficient of absolute ambiguity aversion $\hat{\alpha}>0$ such that the inequality is satisfied when facing any receiver with $\phi_r$ at least as concave as $-e^{-\hat{\alpha}x}$. The following proposition establishes this fact.}

\begingroup
\color{black}
\begin{proposition}\label{prop:belief_threshold_critical_aversion}
Fix a Pareto-ranked splitting $(\overline{\sigma}, \underline{\sigma}, \lambda)$ of an obedient optimal Bayesian persuasion experiment $\sigma^{BP}$. Fix $\phi_s$ for the sender. Then there exists an $\hat{\alpha}>0$ such that, for every receiver with $\phi_r$ at least as concave as $-e^{-\hat{\alpha}x}$,
\begin{equation*}
\rho^{\phi_s,u_s}((\overline{\sigma}, \underline{\sigma}),\lambda)<\kappa-\lambda,
\end{equation*}
where $\kappa$ is given by \eqref{equ_thm_Pareto_ranked_improvement}. Consequently, the binary ambiguous experiment $\mu^{((\overline{\sigma},\underline{\sigma}),\kappa)}$ is obedient and benefits the sender.
\end{proposition}
\par
\endgroup

\section{Optimal Persuasion with Ambiguous Communication}\label{sec_properties_optimal_ambiguous_communication}

Just as Pareto-ranked splittings were key to understanding when and how ambiguous communication can benefit the sender, our next result shows that Pareto-ranked experiments are fundamental in describing properties of \emph{optimal} ambiguous communication. 

\begin{theorem}\label{thm_optimal_persuasion_pareto_ranked_experiments}
Suppose $\mu$ is a solution to \eqref{eq:sender_problem}, and that $\phi_{r}$ is strictly concave.\footnote{The only role of strict concavity of $\phi_{r}$ is to simplify the statement of the theorem. Without it, one needs to add conditions checking if $\phi'_{r}(u_{r}(\sigma_{\theta})) \neq \phi'_{r}(u_{r}(\sigma_{\theta'}))$ to each part of the theorem.} Then 
\begin{enumerate}[(i)]
\item for all $\sigma_{\theta}, \sigma_{\theta'} \in \text{supp}(\mu)$ such that $u_{s}(\sigma_{\theta})\geq U_{s}(\mu) \geq u_{s}(\sigma_{\theta'})$ with at least one inequality strict, $\sigma_{\theta}$ and $\sigma_{\theta'}$ are weakly Pareto-ranked, and
\item if $\phi_{s}$ is affine, then for all $\sigma_{\theta}\in \text{supp}(\mu)$, there does not exist a Pareto-ranked splitting of $\sigma_{\theta}$, $(\overline{\sigma}, \underline{\sigma}, \lambda)$, such that $u_{s}(\overline{\sigma}) \geq U_{s}(\mu) \geq u_{s}(\underline{\sigma})$. 
\end{enumerate}

\end{theorem}

Part (i) of the result says that if two experiments in the support of $\mu$ bracket the sender's payoff from an optimal ambiguous experiment, they must be weakly Pareto-ranked. Part (ii) shows there must not exist any opportunities to introduce additional ambiguity through Pareto-ranked splittings that bracket the sender's payoff from the ambiguous experiment. Unlike Part (i), Part (ii) requires the sender to be ambiguity neutral. However, Remark \ref{rem:linear_phi_s_relaxation} in the Appendix explains why a similar result holds whenever the sender is not too ambiguity averse compared to the receiver.

To gain intuition for part (i), first observe that if such $\sigma_{\theta}$ and $\sigma_{\theta'}$ are not weakly Pareto-ranked, then the receiver must get a strictly higher expected payoff from $\sigma_{\theta'}$ than from $\sigma_{\theta}$, while the reverse is true for the sender. Ambiguity aversion then implies that the receiver's effective measure places more weight on $\sigma_{\theta}$ relative to $\sigma_{\theta'}$ than the ambiguity neutral weights  do, i.e., $em^{\mu}_{\theta}/em^{\mu}_{\theta'} > \mu_{\theta}/\mu_{\theta'}$. If $\sigma_{\theta}$ and $\sigma_{\theta'}$ are the only two experiments in the support of $\mu$, the sender can merge them into the (unambiguous) experiment $em^{\mu}_{\theta} \sigma_{\theta} + em^{\mu}_{\theta'} \sigma_{\theta'}$. By construction, the receiver would continue to be obedient, and the sender would strictly benefit from this merging -- a profitable deviation. When $\sigma_{\theta}$ and $\sigma_{\theta'}$ are not the only two experiments in the support of $\mu$, however, this is not the complete story. Splitting or merging experiments may shift their combined weights in the effective measure relative to the other experiments. Part of the additional insight of the proof is that when $u_{s}(\sigma_{\theta})\geq U_{s}(\mu) \geq u_{s}(\sigma_{\theta'})$ holds, the merging nevertheless benefits the sender.
The intuition for part (ii) is similar. 

 Lemmas \ref{thm_inverstly_ranked_concavification_experiment_space} and \ref{thm_Pareto_ranked_concavification_experiment_space} in the Appendix show that concavity (resp. convexity) of $1/\phi'_{r}$ is sufficient to extend the conclusions to pairs of experiments on a particular side of $U_{s}(\mu)$, and assuming linearity of $1/\phi'_{r}$ leads to the following simpler necessary conditions for optimal persuasion: 

\begin{proposition}\label{cor_linear_phi_inverse_concavification}
Suppose $\mu$ is a solution to \eqref{eq:sender_problem}, and 
\begin{equation}\label{equ_linear_phi_inverse} 
\phi_{r}(x) = \ln(x + b)
\end{equation}
for some $b \in \R$ where $x + b > 0$ for $x \in [\min_{a, \omega}u_{r}(a, \omega), \max_{a, \omega}u_{r}(a, \omega)]$. Then, all experiments are weakly Pareto-ranked, that is,  for all $\sigma_{\theta}, \sigma_{\theta'} \in \text{supp}(\mu)$, $\sigma_{\theta}$ and $\sigma_{\theta'}$ are weakly Pareto-ranked. If, in addition, the sender is ambiguity neutral, no Pareto-ranked splitting of any $\sigma_{\theta}\in \text{supp}(\mu)$ exists. 
\end{proposition}

The result that all experiments used must be weakly Pareto-ranked is reminiscent of a key Pareto-ranking result \cite[Lemma 2]{rayo2010optimal} in an entirely different persuasion setting (one in which ambiguity plays no role).
The more general results of our Theorem \ref{thm_optimal_persuasion_pareto_ranked_experiments} have no obvious analogue in the setting of \cite{rayo2010optimal}.  

The proofs of Theorem \ref{thm_optimal_persuasion_pareto_ranked_experiments} and Proposition \ref{cor_linear_phi_inverse_concavification} make use of a concavification-like characterization \citep{aumann1966game} of the value of optimal persuasion with ambiguous communication. Since concavification plays a prominent role in the existing persuasion literature, we explicitly describe this characterization here.

Let $\Sigma^{*} \subseteq \Sigma$ be the set of obedient experiments. Both $\Sigma$ and $\Sigma^{*}$ are non-empty convex sets and can be embedded into an $|\Omega|(|A|-1)$-dimensional Euclidean space since an experiment specifies, for each $\omega \in \Omega$, a probability distribution over $A$. For each scalar $u \in \R$, define 
\begin{equation*}
\Phi_{u}(\sigma) := \frac{\phi_{s}(u_{s}(\sigma)) - \phi_{s}(u)}{\phi_{r}'(u_{r}(\sigma))},
\end{equation*}
and consider the following maximization problem: 
\begin{equation}\label{equ_program_Phi_star}
\begin{array}{l}
\max\limits_{\lambda\in \Delta_{f}(\Sigma)} \sum_{\theta} \lambda_{\theta} \Phi_{u}(\sigma_{\theta}), \\
\text{subject to} \sum_{\theta}\lambda_{\theta}\sigma_{\theta} \in \Sigma^{*}.
\end{array}
\tag{$\Phi^{*}_{u}$}
\end{equation}
Let $V(\Phi^{*}_{u})$ denote the value of the program \eqref{equ_program_Phi_star}.

Proposition \ref{thm_optimal_persuasion_characterization_experiment_space} states that the value of the ambiguous persuasion program \eqref{eq:sender_problem} is the unique utility level $u$ such that $V(\Phi^{*}_{u}) = 0$. An optimal ambiguous persuasion strategy can be directly constructed from a solution to \eqref{equ_program_Phi_star}, and there always exists such an optimal strategy that makes use of no more than $|\Omega| (|A|-1) + 1$ experiments. 

\begin{proposition}\label{thm_optimal_persuasion_characterization_experiment_space}
The value of \eqref{eq:sender_problem} is $u$ if and only if $V(\Phi^{*}_{u}) = 0$. Moreover, there exists a solution $\mu^{*}$ to \eqref{eq:sender_problem} such that $|\text{supp }(\mu^{*})| \leq |\Omega| (|A|-1) + 1$. 
\end{proposition}

Concavification can be used to compute $V(\Phi^{*}_{u})$: For each $u \in \R$, the program \eqref{equ_program_Phi_star} maximizes over convex combinations of points on the graph of $\Phi_{u}$, exactly the type of program that concavification characterizes. Specifically, for each $u \in \R$, let $\mathrm{cav}{\Phi}_{u}: \Sigma \rightarrow \R$ denote the concavification of $\Phi_{u}$, that is,
\begin{equation*}
\mathrm{cav}{\Phi}_{u}(\sigma)=
\left\{
\begin{array}{l}
\max\limits_{\lambda \in \Delta_{f}(\Sigma)} \sum_{\theta} \lambda_{\theta} \Phi_{u}(\sigma_{\theta}), \\
\text{subject to} \sum_{\theta}\lambda_{\theta}\sigma_{\theta} = \sigma,
\end{array}
\right.
\end{equation*}
and the maximum over $\sigma \in \Sigma^{*}$ of $\mathrm{cav}{\Phi}_{u}(\sigma)$ is $V(\Phi^{*}_{u})$. Any such maximum is achieved by a splitting of some obedient experiment, with the splitting weights given by the effective measure. The following immediate corollary of Proposition \ref{thm_optimal_persuasion_characterization_experiment_space} thus provides a concavification-like characterization of the value of \eqref{eq:sender_problem}:\footnote{See Section \ref{sec_experiments_vs_posteriors} for a discussion comparing this result with concavification in Bayesian persuasion.}

\begin{corollary}\label{cor_optimal_persuasion_concavification_experiment_space}
The value of \eqref{eq:sender_problem} is $u$ if, and only if, $\max_{\sigma \in \Sigma^{*}} \mathrm{cav}{\Phi}_{u}(\sigma) = 0$. 
\end{corollary}

We can apply Corollary \ref{cor_optimal_persuasion_concavification_experiment_space} and Proposition \ref{cor_linear_phi_inverse_concavification} to the Introductory Example:

\begin{example}[Optimal Persuasion in the Introductory Example]
Suppose $\phi_{r}(x) = \ln(x+5)$ and $\phi_{s}(x) = x$. These results allow us to show the sender's optimal strategy in the Introductory Example is the binary ambiguous experiment described in Figure \ref{fig:intro_example_2} with $\mu_{\overline{\theta}} = 39/50$. Both the sender and receiver do better than they would under Bayesian persuasion.
\end{example}

\section{Further Discussion}\label{sec_discussion}

\subsection{Comparison with BLL and persuasion of dynamically inconsistent ambiguity-averse receivers}\label{sec_comparison_BLL}

To facilitate comparison with BLL's approach, while maintaining our modeling setting as much as possible, consider the following dynamically inconsistent model of the receiver's interim preference formation: The receiver applies Bayesian updating using $\mu$ rather than the effective measure (recall the dynamically consistent updating described in Remark \ref{rem_dynamic_consistency} in Section \ref{sec_incentive_compatibility}), and chooses their message-contingent strategy to maximize these interim preferences. Formally, fixing an ambiguous experiment $\mu$, and any message $m \in M$ that it could generate, let 
\begin{align*}
q_{\theta}(\omega|m) = \frac{p(\omega)\sigma_{\theta}(m|\omega)}{\sum_{\omega'}p(\omega')\sigma_{\theta}(m|\omega')}
\end{align*}
denote this receiver's posterior belief over $\Omega$ conditional on message $m$ under experiment $\sigma_{\theta}$. Furthermore, let 
\begin{align*}
\eta(\sigma_{\theta}|m) = \frac{\mu(\sigma_{\theta})\sum_{\omega}p(\omega)\sigma_{\theta}(m|\omega)}{\sum_{\theta'}\sum_{\omega}\mu(\sigma_{\theta'})p(\omega)\sigma_{\theta'}(m|\omega)}
\end{align*}
denote this receiver's posterior belief over $\Sigma$ conditional on message $m$. Then such a receiver chooses their strategy $\tau$ such that, for each $m$, $\tau(m) \in \Delta(A)$ maximizes
\begin{align*}
\sum_{\theta}\eta(\sigma_{\theta}|m)\phi_{r} \left(\sum_{\omega,a} q_{\theta}(\omega|m)\tau(a|m)u_{r}(a, \omega)\right).
\end{align*}
We refer to this as an \emph{inconsistent} receiver.\footnote{Being precise, the inconsistent receiver's interim preferences are inconsistent with their given ex-ante preferences. As \cite{pahlke2022dynamic} points out, the interim preferences in BLL are consistent with \emph{some} ex-ante preferences, though these may be very different from the given ones.} Similarly, we refer to a receiver who behaves as in our model as a \emph{consistent} receiver. 

One key difference between inconsistent and consistent receivers relates to opting out of persuasion games with ambiguous communication. It is easy to construct examples where the sender's optimal ambiguous experiment makes an inconsistent receiver worse off than they would be under no communication (see, e.g., Section 2 of BLL). Such a receiver has an incentive to opt out of the persuasion game ex ante. In contrast, a consistent receiver will never be worse off under ambiguous communication than under no communication.  

Turning to when and how the sender benefits from ambiguous communication when facing a consistent versus an inconsistent receiver, Belief-threshold games provide a useful setting for this comparison. Consider belief-threshold games with $\sum_{i=1}^{n}p(\omega_{i})/T_{i} < 1$. Theorem \ref{thm:belief_threshold_iff} told us that an ambiguity-neutral sender benefits from ambiguous communication when facing a consistent receiver if and only if $v_{i}/k_{i}T_{i} \neq v_{j}/k_{j}T_{j}$ for some $i,j \in \{1, \ldots, n\}$, i.e., \eqref{equ:belief_threshold_iff} in Theorem \ref{thm:belief_threshold_iff} holds. Suppose that the receiver's payoff to some action is \emph{safe} in the sense that it is independent of the state. Note that the safety of any specified action is compatible with any choice of the $v_{i}$'s, $k_{i}$'s, and $T_{i}$'s, and thus does not affect whether \eqref{equ:belief_threshold_iff} holds. The next result shows that the sender benefits from ambiguous communication when facing a sufficiently ambiguity-averse inconsistent receiver if and only if this safe action is not $a_{0}$.\footnote{Belief-threshold games cannot have multiple safe actions for the receiver.} 

\begin{theorem}\label{thm:benefit_inconsistent_iff}
Fix any belief-threshold game with $\sum_{i=1}^{n}p(\omega_{i})/T_{i} < 1$ and $u_{r}(a_{\bar{i}}, \omega)$ constant across $\omega \in \Omega$ for some $\bar{i}$. If $ \bar{i} > 0$, then the sender benefits from ambiguous communication with any inconsistent receiver having sufficiently concave $\phi_{r}$, and if $\bar{i} = 0$, then the sender does not benefit from ambiguous communication with any inconsistent receiver.
\end{theorem}

A key construction in helping understand this result is the following ambiguous experiment: given an optimal Bayesian persuasion experiment $\sigma^{BP}$, let ambiguous experiment $\mu$ be such that $\mu(\sigma_{1}) = \mu(\sigma_{2}) = 1/2$, $\sigma_{1}$ and $\sigma_{2}$ send messages in $\{m_{0}, m_{j}\} \cup A$, and 
\begin{align*}
& \sigma_{1}(m_{0}|\omega_{0}) = \sigma_{2}(m_{j}|\omega_{0}) = \sigma^{BP}(a_{0}|\omega_{0}) + \sigma^{BP}(a_{j}|\omega_{0}), \\
&\sigma_{1}(m_{j}|\omega_{j}) = \sigma_{2}(m_{0}|\omega_{j}) = \sigma^{BP}(a_{j}|\omega_{j}),
\end{align*}
where $j = \bar{i}$ if $\bar{i} > 0$ and $j = 1$ otherwise, all other entries for $m_{0}$, $m_{j}$, $a_{0}$, and $a_{j}$ are zero, and for all other actions, $\sigma_{1} = \sigma_{2} = \sigma^{BP}$. Starting from $\sigma^{BP}$, this ambiguous experiment $\mu$ injects ambiguity through the messages $m_{0}$ and $m_{j}$. Observe that the Bayesian posterior beliefs over $\Omega$ conditional on $m_{0}$ under $\sigma_{1}$ are the same as those conditional on $m_{j}$ under $\sigma_{2}$, and vice versa. These two messages are an example of what BLL call strong synonyms, and such synonyms play an important role in their analysis. Intuitively, the ambiguity generated by these synonyms drives the inconsistent receiver toward the safe action $a_{\bar{i}}$ in states $\omega_{0}$ and $\omega_{j}$. When $\bar{i} > 0$, this results in $a_{\bar{i}}$ being played more often and $a_{0}$ less often than under $\sigma^{BP}$, which helps the sender, and is how that case of the theorem is proved. When $\bar{i} = 0$, however, this results in $a_{0}$ being played more often and $a_{j}$ less often than under $\sigma^{BP}$, which hurts the sender. So in this case, ambiguity generated by strong synonyms doesn't benefit the sender. The proof shows that this remains true for any other ambiguous experiment.

The combination of Theorems \ref{thm:belief_threshold_iff} and \ref{thm:benefit_inconsistent_iff} provide a simple recipe for constructing examples that contrast the sender's ability to benefit from ambiguous communication with a consistent versus an inconsistent receiver. First, consider a binary-action belief-threshold game and make action $a_1$ safe. Corollary \ref{cor_binary_action} rules out any sender benefit with a consistent receiver, while Theorem \ref{thm:benefit_inconsistent_iff} shows that the sender can benefit with an inconsistent receiver having sufficiently concave $\phi_{r}$. Conversely, consider a three-action belief-threshold game and make $a_0$ safe. This places no restrictions on the $k_i,T_i,$ or $v_i$, so consider $k_1T_1=k_2T_2$ and $v_1\neq v_2$. Thus, condition \eqref{equ:belief_threshold_iff} in Theorem \ref{thm:belief_threshold_iff} is satisfied, implying an ambiguity-neutral sender can benefit with a consistent receiver whenever $\phi_r$ is strictly concave, while the safety of $a_0$ rules out any benefit with an inconsistent receiver by Theorem \ref{thm:benefit_inconsistent_iff}.

More generally, these results show that the determinants of whether the sender can  benefit from ambiguous communication are entirely different when facing a consistent versus an inconsistent receiver. Furthermore, when benefit is possible, it can be obtained through very different forms of ambiguous communication: Pareto-ranked splittings of $\sigma^{BP}$ when facing a consistent receiver, and augmenting $\sigma^{BP}$ with strong synonyms when facing an inconsistent receiver.

Why doesn't a Pareto-ranked splitting of $\sigma^{BP}$ benefit the sender when facing an inconsistent receiver in these games? An inconsistent receiver can be thought of as using a different effective measure for each message. Recall from Section \ref{sec_threshold_games} that $\sigma^{BP}$ for such games recommends all actions with positive probability in state $\omega_{0}$ and recommends $a_{i}$ with probability one in state $\omega_{i}$ for $i > 0$. Thus, in any Pareto-ranked splitting of $\sigma^{BP}$, the recommendation $a_{i}$ may be produced only in states $\omega_{i}$ and $\omega_{0}$. Fix $i > 0$. From Lemma \ref{lem:belief_threshold_payoff_representation} and that $a_{\bar{i}}$ is a safe action, $u_{r}(a_{i}, \omega_{i}) - u_{r}(a_{i}, \omega_{0}) \geq 0$. Therefore, upon receiving the recommendation $a_{i}$, the inconsistent receiver's effective measure will place at least as much weight as an ambiguity-neutral receiver on the $\theta$ such that $\sigma_{\theta}(a_{i}|\omega_{0})$ is higher (because this will give a lower interim payoff to the receiver). For the better experiment in the splitting, it must be that enough $a_{i}$'s for $i > 0$ get more weight in state $\omega_{0}$ than they do in that state in the worse experiment (otherwise it would not be better for the sender). But it will be exactly for these $a_{i}$'s that the inconsistent receiver's effective measure fails to move incentives in the direction that would allow the sender to place more weight on the better experiment while maintaining obedience.  

Correspondingly, why doesn't the strong synonym construction benefit the sender when facing a consistent receiver? Observe that this modification of $\sigma^{BP}$ can be described in two steps. First, form an experiment $\sigma^{BP}_{m}$ by merging the two messages $a_{0}$ and $a_{j}$ into a single message $m$. Second, inject ambiguity by splitting $m$ into two messages $m_{0}$ and $m_{j}$, where under $\theta_{1}$, $m_{0}$ substitutes for $m$ in state $\omega_{0}$ and $m_{j}$ substitutes for $m$ in state $\omega_{j}$, while under $\theta_{2}$, the pattern of substitution is reversed. By the symmetry of these substitutions, the receiver's optimal response is the same to $m_{0}$ as to $m_{j}$. Thus, a consistent receiver has the same expected payoff from $\sigma_{\theta_{1}}$ as from $\sigma_{\theta_{2}}$, and therefore behaves as if facing $\sigma^{BP}_{m}$, and in particular plays $a_{0}$ in response to $m$, which is worse for the sender than $\sigma^{BP}$. In this sense, a consistent receiver ``sees through'' the ambiguity generated by the strong synonyms. 

\begin{remark}
\emph{There are games outside the belief-threshold class in which the sender might benefit from using a Pareto-ranked splitting of $\sigma^{BP}$ when facing an inconsistent receiver. The introductory example is such a game. The splitting used there has the property that for every message in the support of both experiments in the splitting, the receiver's expected payoff is higher under $\sigma_{\overline{\theta}}$ than under $\sigma_{\underline{\theta}}$, and thus all of the receiver's interim effective measures move in a direction favorable for the sender. Even so, belief-threshold games make clear that the existence of Pareto-ranked splittings with this strong property is not a necessary condition for the sender to benefit from ambiguous communication when facing an inconsistent receiver. \qedsymbol}
\end{remark}

\subsection{What if the sender knows the resolution of ambiguity in advance?}\label{sec_cheap_talk}

\begingroup
\color{black}
\textcolor{black}{We have assumed that the sender commits to an ambiguous experiment $\mu$ without knowing in advance how this ambiguity will resolve.} Here we consider the alternative
assumption that the sender privately observes the resolution of the source of ambiguity, $\alpha \in \mathcal{A}$, before committing to $\mu$ (and this is common knowledge).
In this case, the ``type $\alpha$'' sender privately knows that
the experiment generating the message will be $\sigma_{\theta}$ for the
$\theta$ corresponding to the partition element from $\mathcal{A}$ containing $\alpha$. The receiver remains uncertain. We now
argue that, in any perfect Bayesian equilibrium of this modified game,
the sender does \emph{not} benefit from ambiguous communication.\footnote{\textcolor{black}{A (pure) strategy for the sender maps from $\mathcal{A}$ to ambiguous
experiments, whereas a (pure) strategy for the receiver is a mapping
from ambiguous experiments and messages to actions. Upon observing the sender's choice of an ambiguous experiment, the receiver may revise their beliefs about
the sender's private information, as in cheap-talk games. We assume that
whenever the receiver faces an \emph{unambiguous} experiment, the
receiver's action is optimal given that experiment and the realized
message, regardless of whether it is on or off-path.}} First observe that the payoff of any privately informed sender must
be at least $u_{s}^{BP}$ in any equilibrium. To see this, note that
a sender can always offer an \emph{unambiguous} optimal Bayesian persuasion
experiment and, thereby, guarantee a payoff of $u_{s}^{BP}$.

One subtlety is that observing the sender's choice of $\mu$ may lead the receiver to update their beliefs
about the underlying state $\alpha$ (and thus the experiment $\sigma_{\theta}$) because senders informed of different
$\alpha$'s might make different choices in equilibrium. Let $\beta^{\mu}\in\Delta(\mathrm{supp}(\mu))$ denote this possibly updated belief over experiments. Rather than reacting to $\mu$,
the receiver views the ambiguous experiment as $\beta^{\mu}$.
Let $\tau \in BR(\beta^{\mu})$ be the receiver's equilibrium response. The same first-order argument as in Lemma \ref{lem_effmeasure} implies that
$\tau \in br(\sum_{\theta}\zeta_{\theta}^{\mu}\sigma_{\theta})$, where
\begin{equation*}
\zeta_{\theta}^{\mu}
=\frac{\beta_{\theta}^{\mu}\phi'_r(u_r(\sigma_{\theta},\tau))}
{\sum_{\tilde{\theta}}\beta_{\tilde{\theta}}^{\mu}\phi'_r(u_r(\sigma_{\tilde{\theta}},\tau))}.
\end{equation*}
For all senders informed of $\hat{\alpha}$ corresponding to some $\sigma_{\widehat{\theta}}\in\mathrm{supp}(\beta^{\mu})$
committing to $\mu$, the sender's expected payoff
is
\begin{equation}
u_s(\sigma_{\widehat{\theta}},\tau)\geq u_s^{BP},\label{equ_sender_payoff_learn_theta}
\end{equation}
where the inequality follows from the preceding observation that $u_{s}^{BP}$
is a lower bound on the informed sender's equilibrium payoff.

Since $\mathrm{supp}(\zeta^{\mu})\subseteq\mathrm{supp}(\beta^{\mu})$,
we have
\begin{align*}
u_s^{BP}
&\geq u_s\left(\sum_{\theta}\zeta_{\theta}^{\mu}\sigma_{\theta},\tau\right)=\sum_{\widehat{\theta}}\zeta_{\widehat{\theta}}^{\mu}u_s(\sigma_{\widehat{\theta}},\tau)
\geq u_s^{BP},
\end{align*}
where the first inequality follows since $u_{s}^{BP}$ is the sender's
best expected payoff from any unambiguous experiment to which the
receiver best responds, the equality from linearity of $u_{s}$ in
its first argument, and the second inequality from \eqref{equ_sender_payoff_learn_theta}.
\textcolor{black}{One lesson is that to benefit from ambiguous communication, the sender needs to use a source of ambiguity the resolution of which they do not already know at the time of committing to the ambiguous experiment.}
\par
\endgroup

\subsection{Robustness to Small Misperceptions}\label{sec_robust_benefit}

\begingroup
\color{black}
\textcolor{black}{We have assumed that if the sender designs the ambiguous experiment $\mu$, the receiver perceives it correctly, and that the sender knows $\phi_{r}$ exactly.} More realistically,  
the receiver might have a somewhat different perception of the experiment than the one the sender intends to convey. After all, conveying the exact specifications of an experiment is a complex task, let alone of an ambiguous one. \textcolor{black}{We show that if the sender can benefit from ambiguous communication, this property is robust to small misperceptions of the experiment or of $\phi_{r}$.}  

\begin{proposition}\label{prop_robust_benefit}
    Suppose ambiguous communication benefits the sender and that the set of obedient experiments has a non-empty interior. Then, there exists a non-empty open set of obedient ambiguous experiments that benefit the sender.
\end{proposition}

A sketch of the argument (full details are in Section \ref{sec_proofs_auxiliary_results}) is to show the persuasion problem is sufficiently continuous to guarantee the existence of an open set of obedient ambiguous experiments that benefit the sender. In fact, this continues to be true even under small perturbations of $\phi_r$, so that the existence of benefits does not rest on exact knowledge of the receiver's ambiguity aversion.   

\begin{corollary}\label{coro_robust_benefit}
    Suppose ambiguous communication benefits the sender and that the set of obedient experiments has a non-empty interior. Then, there exists an ambiguous experiment that benefits the sender, and continues to do so under small perturbations of $\phi_{r}$.
\end{corollary}
\par
\endgroup

\subsection{Robustness to Pre-Existing Ambiguity about Payoff-Relevant States}\label{sec_prior_ambiguity}

\begingroup
\color{black}
\textcolor{black}{To isolate the role of ambiguous communication, we have assumed that there is no pre-existing ambiguity about the payoff-relevant states. A full analysis of persuasion with ambiguous communication given arbitrary pre-existing ambiguity over states is beyond the scope of this paper and is an interesting topic for future research. Nonetheless, many results and concepts emphasized in this paper remain relevant. \textcolor{black}{Here} we establish \textcolor{black}{(Theorem \ref{thm:binary_improvement_under_prior_ambiguity})} the continued role of Pareto-ranked splittings of obedient experiments in generating benefits from ambiguous communication even under pre-existing ambiguity. Additionally, the conditions under which the sender benefits from ambiguous communication are robust to the introduction of a small amount of pre-existing ambiguity \textcolor{black}{(Theorem \ref{thm:sender_benefit_prior_ambiguity})}.}

In the context of our model, pre-existing ambiguity over the payoff-relevant state space $\Omega$ can be represented by a common subjective distribution over priors $\eta \in \Delta(\Delta(\Omega))$. Denote by $u_i(p,\sigma)$ and $U_i(p,\mu)$ the extensions of $u_i(\sigma)$ and $U_i(\mu)$ to allow for different priors $p \in \mathrm{supp}(\eta)$, continuing to suppress the obedient strategy $\tau^*$. Theorem \ref{thm:binary_improvement_under_prior_ambiguity} establishes the benefit of Pareto-ranked splittings of obedient experiments where both obedience and Pareto ranking are satisfied for each $p \in \mathop{supp}(\eta)$ rather than simply the unique $p$ as in the rest of the paper.

\begin{theorem}\label{thm:binary_improvement_under_prior_ambiguity}
    Let $\sigma^{*}$ be obedient under all $p \in \mathop{supp}(\eta)$. Suppose that $\phi_{r}$ is strictly concave and $(\overline{\sigma}, \underline{\sigma}, \lambda)$ is a Pareto-ranked splitting of $\sigma^{*}$ satisfying $u_s(p, \overline{\sigma}) > u_s(p, \underline{\sigma})$ for all $p \in \mathop{supp}(\eta)$ and $\underline{\sigma}$ is obedient under all $p \in \mathop{supp}(\eta)$. Then the binary ambiguous experiment $\mu^{((\overline{\sigma},\underline{\sigma}),\kappa)}$, with $\kappa$ given by
    \begin{equation}\label{equ_thm_binary_improvement_under_prior_ambiguity}
        \kappa = \min\limits_{p \in \mathop{supp}(\eta)} \frac{\lambda \phi_{r}'(u_{r}(p, \underline{\sigma}))}{\lambda \phi_{r}'(u_{r}(p, \underline{\sigma})) + (1-\lambda) \phi_{r}'(u_{r}(p, \overline{\sigma}))} > \lambda,  
    \end{equation}
    is obedient and an ambiguity neutral sender does strictly better using $\mu^{((\overline{\sigma},\underline{\sigma}),\kappa)}$ than using $\sigma^{*}$. 
\end{theorem}

\begin{theorem}\label{thm:sender_benefit_prior_ambiguity}
    Suppose an ambiguity neutral sender benefits from ambiguous communication when the prior is $p \in \Delta(\Omega)$ and that the set of distributions over states and actions for which obedience is a best response for the receiver has a non-empty interior. Then, the sender continues to benefit from ambiguous communication under any small enough pre-existing ambiguity, $\eta$, around $p$. Specifically, there exists $\delta > 0$ such that the sender benefits from ambiguous communication whenever $\mathbb{E}[\eta] = p$ and $\mathop{supp}(\eta) \subseteq \{q \in \Delta(\Omega): \Vert q - p \Vert < \delta\}$. 
\end{theorem}

For both theorems, similar results continue to hold as long as the sender is not too ambiguity averse.
\par
\endgroup

\subsection{Why experiments rather than distributions over posteriors?}\label{sec_experiments_vs_posteriors}
Our analysis works directly with statistical experiments rather than the corresponding Bayes-plausible distributions over posteriors frequently used in the literature. Experiments prove more useful in our setting: convex combinations of experiments, which we exploit repeatedly (e.g., in defining effective measure and Pareto-ranked splittings), need not correspond to a convex combination of the corresponding distributions over posteriors, because the latter ignores any overlap in the messages sent by the experiments in the support. 

Similarly, our concavification result in Section \ref{sec_properties_optimal_ambiguous_communication} is another instance where experiments prove more amenable than distributions. The splittings in Corollary \ref{cor_optimal_persuasion_concavification_experiment_space} are on the higher-dimensional space of experiments, rather than on the space of priors as in \cite{kamenica2011bayesian}, and its extension to allow for exogenously heterogeneous priors in \cite{alonso2016bayesian}. We now explain why this is the case.

Suppose we try to write a program in which the sender maximizes with respect to splittings of the prior. Consider the simplest case of an ambiguity-neutral sender. Any ambiguous experiment $\mu$ induces a distribution over the receiver's \emph{effective} posteriors, that is, the posteriors that the ``effective'' experiment, $\sum_{\theta}em^{\mu}_{\theta}\sigma_{\theta}$, induces. Thus, the splitting $\sum_{\theta}em^{\mu}_{\theta}\sigma_{\theta}$ induces may differ from the splitting the experiment $\sum_{\theta} \mu_{\theta}\sigma_{\theta}$ induces. The latter is the one the ambiguity-neutral sender uses to evaluate their payoff. To be amenable to a concavification approach on this space, the sender's objective function would therefore need to be, as in \cite{alonso2016bayesian}, an increasing transformation of a linear function of the distribution over the receiver's effective posteriors. However, the highly non-linear relationship between $em^{\mu}_{\theta}$ and $\mu$ makes this impossible. Economically, this non-linearity has its source in the fact that ambiguity aversion  causes the effective measure to be proportional to the product of the ambiguity neutral probability $\mu_{\theta}$ and the marginal utility  $\phi'_{r}(u_{r}(\sigma_{\theta}))$ and the latter is non-separable across action recommendations. Thus, it is this non-separability across action recommendations in determining obedience under ambiguity aversion that prevents adopting the strategies of \cite{kamenica2011bayesian} and \cite{alonso2016bayesian}. 

\subsection{Related Literature}\label{sec_related_literature}

In addition to the papers cited in the introduction, the following are also at the intersection of Bayesian persuasion (BP) and ambiguity. \citet{kosterina2022persuasion} studies BP when an MEU sender is ambiguous about the receiver's prior, while in \citet{dworczak2022robustbp} an MEU sender (who also has a preference for selecting among MEU-optimal strategies those that perform best under a baseline conjecture) is ambiguous about the exogenous information a receiver might learn.\footnote{\citet{dworczak2022robustbp}'s model is not restricted to single-receiver persuasion settings.}  \citet{NIKZAD2021144} studies BP when the receiver is MEU and has ambiguity about the prior over states. \citet{hedlund2020persuasion} studies BP in problems with two states and two actions, when the receiver has $\alpha$-MEU preferences \citep{ghirardato2004differentiating} and considers an interval of priors and the sender has state-independent preferences over the action taken by the agent and is ambiguity neutral. In all of these papers, the sender is limited to standard, unambiguous experiments, and thus any ambiguity is exogenous. This stands in contrast to the endogeneity of ambiguity in our setting, where it becomes payoff-relevant only through the intentional communication choices of the sender.

\citet{kellner2018endogenous} study cheap talk communication assuming that the receiver has MEU preferences and the sender can choose to communicate ambiguously. The key difference between cheap talk and persuasion is the sender's inability to commit to a communication strategy. Their receiver uses the same dynamically inconsistent update rule as in BLL. They find that both sender and receiver may benefit from the sender choosing to communicate ambiguously. \citet{kellner2017exogenous} studies cheap talk communication with purely exogenous ambiguity.

Papers studying mechanism design with ambiguity include \citet{bose2014mechanism}, \citet{wolitzky2016mechanism}, \citet{ditillio2017design}, \citet{guo2019mechanism} and \citet{tang2021maxmin}, among others. All but \citet{wolitzky2016mechanism} consider ambiguity that arises intentionally through design of the mechanism. \cite{dutting2024ambiguous} allow a principal to offer ambiguous contracts to an MEU agent and show how the principal may benefit and that optimal contracts have a simple form. All gains from ambiguous contracting disappear in their model if the agent can hedge against ambiguity by randomizing. 

\appendix
\appendixpage

\numberwithin{definition}{section}
\numberwithin{theorem}{section}
\numberwithin{lemma}{section}
\numberwithin{proposition}{section}
\numberwithin{equation}{section}
\numberwithin{remark}{section}

\section{Proofs of Main Results}\label{sec_proof_results}

\subsection{Proofs for Section \ref{sec:model}}
\begin{proof}[Proof of Proposition \ref{rev-principle}] 

Fix $\mu$ and $\tau \in BR(\mu)$. For each $\sigma_{\theta} \in \text{supp}(\mu)$, define $\sigma_\theta^*$ by $\sigma_\theta^*(a|\omega) =\sum_m\tau(a|m)\sigma_\theta(m|\omega)$ for all $a\in A$ and $\omega\in\Omega$. For each $\theta$ and $i\in\{s,r\}$, $u_i(\sigma_\theta,\tau)=\sum_{\omega,a}p(\omega)u_i(a,\omega)\sum_m\tau(a|m)\sigma_\theta(m|\omega)=u_i(\sigma_\theta^*,\tau^*)$. Define the canonical ambiguous experiment $\mu^* =\sum_{\theta}\mu_\theta\delta_{\sigma_\theta^*}$. Then $U_i(\mu^*,\tau^*)=U_i(\mu,\tau)$ for each $i\in\{s,r\}$.

Suppose $\mu^{*}$ is not obedient. Then there exists $\tilde{\tau}: A \rightarrow \Delta(A)$ such that $U_{r}(\mu^{*}, \tau^{*}) < U_{r}(\mu^{*}, \tilde{\tau})$. Then, define $\tau'(a|m)= \sum_{a'}\tilde{\tau}(a|a') \tau(a'|m)$, and note that for all $\theta$,
\begin{align*}
u_r(\sigma_{\theta},\tau')  & = \sum_{\omega,a}p(\omega)u_r(a,\omega)\sum_{a'} \tilde{\tau}(a|a')\sum_{m}\tau(a'|m)\sigma_{\theta}(m|\omega) \\
& = \sum_{\omega, a, a'}p(\omega)u_r(a,\omega)\tilde{\tau}(a|a')\sigma_{\theta}^{*}(a'|\omega) = u_{r}(\sigma_{\theta}^{*}, \tilde{\tau}).
\end{align*} 
Thus, $U_{r}(\mu, \tau') = U_{r}(\mu^{*},\tilde{\tau}) > U_{r}(\mu^{*}, \tau^{*}) = U_{r}(\mu, \tau)$, implying $\tau \notin BR(\mu)$. 
\end{proof}	

\begin{proof}[Proof of Lemma \ref{lem_effmeasure}] Let $\mu$ be an ambiguous experiment. We show that $\tau^* \in br(\sigma^*)$, where $\sigma^{*} = \sum_{\theta} em^{\mu}_{\theta} \sigma_{\theta}$, if and only if $\tau^* \in BR(\mu)$. $\tau^* \in br(\sigma^*)$ is equivalent to 
\begin{equation}\label{equ_lem_effmeasure}
u_{r}(\sigma^{*}, \tau) - u_{r}(\sigma^{*}, \tau^{*}) \leq 0, \quad \forall \tau. 
\end{equation}  

The concavity and differentiability of $\phi_{r}$ implies that $\phi_{r}(U_{r}(\mu, \tau))$ is concave and differentiable in $\tau$. Thus, $\tau^{*} \in BR(\mu)$ is equivalent to
\begin{equation}\label{lem_effmeasure_inequality}
\langle \nabla_{\tau} \phi_{r}(U_{r}(\mu, \tau^{*})), \tau - \tau^{*}\rangle \leq 0, \quad \forall \tau.
\end{equation}
Using the formula \eqref{equ_effmeasure} for the effective measure $em^{\mu}_{\theta}$, we have
\begin{align*}
    \langle \nabla_{\tau} \phi_{r}(U_{r}(\mu, \tau^{*})), \tau - \tau^{*}\rangle = \left(\sum_{\theta}\mu_{\theta}\phi'_{r}(u_{r}(\sigma_{\theta})) \right) (u_{r}(\sigma^{*}, \tau) - u_{r}(\sigma^{*}, \tau^{*})).
\end{align*}
By the positivity of $\sum_{\theta}\mu_{\theta}\phi'_{r}(u_{r}(\sigma_{\theta}))$, \eqref{lem_effmeasure_inequality} if and only if \eqref{equ_lem_effmeasure}.
\end{proof}

\subsection{Proofs for Section \ref{sec_ambiguous_communication_benefit_sender}}

\begin{proof}[Proof of Theorem \ref{thm_necessary}] 
We prove the theorem by first showing the following lemma and then establishing that the conditions in Theorem \ref{thm_necessary} are equivalent to the conditions in the lemma.
\begin{lemma}\label{lem:necessary_equivalent_condition}
Ambiguous communication benefits the sender only if $\phi_{r}$ is not affine and there exist Pareto-ranked experiments, $\sigma$ and $\sigma^*$ such that:  (i) $\mathrm{supp\,} \sigma(\cdot|\omega)= \mathrm{supp\,} \sigma^*(\cdot|\omega)$ for all $\omega$, (ii) $u_s(\sigma) >  u_s^{BP}$, and (iii) $\tau^* \in br(\sigma^*) \setminus br(\sigma)$.
\end{lemma}    

\begin{proof}[Proof of Lemma \ref{lem:necessary_equivalent_condition}]

Let $\mu^*$ solve \eqref{eq:sender_problem} and benefit the sender.
Define $\sigma:=\sum_\theta\mu^*_\theta\sigma_\theta$
and $\sigma^*:=\sum_\theta em^{\mu^*}_\theta\sigma_\theta$.
Since $\mu^*$ and $em^{\mu^*}$ have the same support,
$\sigma$ and $\sigma^*$ satisfy (i).
Lemma \ref{lem_effmeasure} gives $\tau^*\in br(\sigma^*)$,
and concavity of $\phi_s$ yields $u_s(\sigma^*)\leq u_s^{BP}<U_s(\mu^*)
\leq\sum_\theta\mu^*_\theta u_s(\sigma_\theta)
=u_s(\sigma)$. Thus $\tau^*\notin br(\sigma)$, proving (ii) and (iii).
Also, $\sigma\neq\sigma^*$ implies $em^{\mu^*}\neq\mu^*$
and hence $\phi_r$ is not affine.
To establish $u_r(\sigma)>u_r(\sigma^*)$, we use the
following lemma.

\begin{lemma}\label{lem_effective_measure_reduces_payoff}
Fix two monotonic sequences $x_1 \geq \dots \geq x_n$, $0< y_1 \leq \dots \leq y_n$, and $\mu \in \Delta(\{1, 2, \dots,n\})$. If there exist indices $i^* <j^*$ such that $\mu_{i^*} > 0$, $\mu_{j^*} > 0$, $x_{i^*} > x_{j^*}$ and $y_{i^*} < y_{j^*}$. Then $(\sum_{i=1}^n x_i \mu_i y_i)/(\sum_{j=1}^n \mu_j y_j) < \sum_{i=1}^n x_i \mu_i$.\footnote{This lemma seems like a known result, but we could not locate a reference and so include a proof in the \textcolor{black}{Appendix \ref{sec_proofs_auxiliary_results}} for completeness.}
\end{lemma}

Order $u_r(\sigma_\theta)$ decreasingly. Since $\phi_r$ is strictly increasing and concave, the corresponding derivatives $\phi_r'(u_r(\sigma_\theta))$ are positive and nondecreasing. They are nonconstant on $\text{supp}(\mu^*)$ because $em^{\mu^*}\neq\mu^*$. Hence some $\sigma_\theta,\sigma_{\theta'}\in \text{supp}(\mu^*)$ satisfy $u_r(\sigma_\theta)>u_r(\sigma_{\theta'})$ and $\phi_r'(u_r(\sigma_\theta)) <\phi_r'(u_r(\sigma_{\theta'}))$.
Lemma \ref{lem_effective_measure_reduces_payoff} therefore gives $u_r(\sigma)>u_r(\sigma^*)$.
\end{proof}

\noindent We show conditions in Lemma \ref{lem:necessary_equivalent_condition} imply those in Theorem \ref{thm_necessary} with the following lemma:

\begin{lemma}\label{lem_existence_splitting}
Fix $\sigma$. If there exists $\hat{\sigma}$ such that $u_{s}(\hat{\sigma}) > u_{s}(\sigma)$, $u_{r}(\hat{\sigma}) > u_{r}(\sigma)$ and for all $\omega \in \Omega$, $\text{supp}(\hat{\sigma}(\cdot|\omega)) \subseteq \text{supp}(\sigma(\cdot|\omega))$, then there exists a Pareto-ranked splitting of $\sigma$, $(\overline{\sigma}, \underline{\sigma}, \lambda)$ with $\overline{\sigma} = \hat{\sigma}$.\footnote{The proof of Lemma \ref{lem_existence_splitting} is in \textcolor{black}{Appendix \ref{sec_proofs_auxiliary_results}}.} 
\end{lemma}

Applying Lemma \ref{lem_existence_splitting} to the experiments in Lemma \ref{lem:necessary_equivalent_condition} gives a Pareto-ranked splitting $(\overline{\sigma},\underline{\sigma},\lambda)$ of the obedient $\sigma^*$ with $\overline{\sigma}=\sigma$.
Conditions (ii) and (iii) then give $u_s(\overline{\sigma})>u_s^{BP} \geq u_s(\sigma^*)>u_s(\underline{\sigma})$, as required.

Finally, we show the conditions in Theorem \ref{thm_necessary} imply those in Lemma \ref{lem:necessary_equivalent_condition}. Suppose there exists an obedient experiment $\hat{\sigma}$ satisfying the conditions in Theorem \ref{thm_necessary}. Because $u_{s}(\overline{\sigma}) > u_{s}^{BP}$, there exists $\underline{\gamma} \in (\lambda,1)$ such that for all $\gamma \in (\underline{\gamma}, 1)$, $\gamma u_{s}(\overline{\sigma})+ (1-\gamma)u_{s}(\underline{\sigma}) > u_{s}^{BP}$, and thus $\tau^{*} \notin br( \gamma \overline{\sigma} + (1-\gamma)\underline{\sigma})$. Fix any such $\gamma$ and let $\sigma = \gamma \overline{\sigma} + (1-\gamma)\underline{\sigma}$. Let $\sigma^{*} = \hat{\sigma}$. It follows that $\tau^{*} \in br(\sigma^{*})$ since $\hat{\sigma}$ is obedient, and thus $u_{s}^{BP} \geq u_{s}(\sigma^{*})$. Since both $\sigma$ and $\sigma^{*}$ are strict mixtures of $\overline{\sigma}$ and $\underline{\sigma}$, they satisfy the common support condition (i). By the definition of Pareto-ranked splitting and $\gamma > \lambda$, $u_{r}(\sigma) > u_{r}(\sigma^{*})$. \end{proof} 

\begin{proof}[Proof of Corollary \ref{cor_binary_action}]
Suppose $\sigma$ and $\sigma^*$ satisfy the conditions in
Lemma \ref{lem:necessary_equivalent_condition}.
Since $\tau^*\notin br(\sigma)$, assume without loss of
generality that obedience fails after recommendation $a_1$.
Together with $\tau^*\in br(\sigma^*)$, this gives $\sum_\omega p(\omega)\sigma(a_1|\omega)
[u_r(a_1,\omega)-u_r(a_2,\omega)]<0$ and $\sum_\omega p(\omega)\sigma^*(a_1|\omega)
[u_r(a_1,\omega)-u_r(a_2,\omega)]\geq0$.
Using
$\sigma(a_2|\omega)-\sigma^*(a_2|\omega)
=\sigma^*(a_1|\omega)-\sigma(a_1|\omega)$,
we obtain $u_r(\sigma)-u_r(\sigma^*) = \sum_\omega p(\omega)[u_r(a_1,\omega)-u_r(a_2,\omega)]
[\sigma(a_1|\omega)-\sigma^*(a_1|\omega)]<0$, contradicting $u_r(\sigma)>u_r(\sigma^*)$.
\end{proof}

\begingroup
\color{black}
\begin{proof}[Proof of Theorem \ref{thm_necessary_condition_existence_of_Pareto_ranked_splitting}]
Suppose that $\mu$ is obedient and $U_s(\mu)>u_s^{BP}$. Let $\sigma^*:=\sum_\theta em^\mu_\theta\sigma_\theta$. By Lemma \ref{lem_effmeasure}, $\sigma^*$ is obedient. Concavity of $\phi_s$ therefore gives
\begin{equation*}
u_s(\sigma^*)\leq u_s^{BP}<U_s(\mu)\leq\sum_\theta\mu_\theta u_s(\sigma_\theta).
\end{equation*}
In particular, some component in $\mathrm{supp}(\mu)$ gives the sender strictly more than $u_s^{BP}$.

Suppose, toward a contradiction, that no pair in $\mathrm{supp}(\mu)$ satisfying $u_s(\sigma_\theta)>u_s^{BP}\geq u_s(\sigma_{\theta'})$ is Pareto-ranked. Then every such pair satisfies $u_r(\sigma_\theta)\leq u_r(\sigma_{\theta'})$. By concavity of $\phi_r$ and \eqref{equ_effmeasure}, the ratio $em^\mu_\theta/\mu_\theta$ is therefore at least as large for every component above the sender's Bayesian persuasion payoff as for every component weakly below it. Let $c>0$ be the smallest of these ratios among components with $u_s(\sigma_\theta)>u_s^{BP}$. Thus, on $\mathrm{supp}(\mu)$, $em^\mu_\theta\geq c\mu_\theta$ when $u_s(\sigma_\theta)>u_s^{BP}$ and $em^\mu_\theta\leq c\mu_\theta$ otherwise. It follows that
\begin{align*}
u_s(\sigma^*)-u_s^{BP}=\sum_\theta em^\mu_\theta\bigl[u_s(\sigma_\theta)-u_s^{BP}\bigr]\geq c\sum_\theta\mu_\theta\bigl[u_s(\sigma_\theta)-u_s^{BP}\bigr]>0,
\end{align*}
contradicting $u_s(\sigma^*)\leq u_s^{BP}$.
\end{proof}
\par
\endgroup

\subsection{Proofs for Section \ref{sec_improvement_binary_ambiguous_communication}}
The following lemma is used in proving Theorem \ref{sufficient_BP}:
\begin{lemma}\label{lem_prob-premium_improvement}
Let $\overline{\sigma}$ and $\underline{\sigma}$ be experiments such that $u_{i}(\overline{\sigma}) > u_{i}(\underline{\sigma})$.  For all $\kappa,\lambda \in [0,1]$, $U_{i}(\mu^{((\overline{\sigma}, \underline{\sigma}), \kappa)}) > u_{i}(\lambda \overline{\sigma} + (1-\lambda)\underline{\sigma})$ if and only if player i's $((\overline{\sigma}, \underline{\sigma}), \lambda)$-probability premium is strictly less than $\kappa - \lambda$.
\end{lemma}

\begin{proof}[Proof of Lemma \ref{lem_prob-premium_improvement}]
\begin{align*}
& \frac{\phi_{i}(u_{i}(\lambda \overline{\sigma} + (1-\lambda)\underline{\sigma})) - (\lambda \phi_{i}(u_{i}(\overline{\sigma})) + (1-\lambda)\phi_{i}(u_{i}(\underline{\sigma})))}{ \phi_{i}(u_{i}(\overline{\sigma})) - \phi_{i}(u_{i}(\underline{\sigma}))} < \kappa - \lambda\\
\Longleftrightarrow ~& \phi_{i}(u_{i}(\lambda \overline{\sigma} + (1-\lambda)\underline{\sigma}))  < \kappa \phi_{i}(u_{i}(\overline{\sigma}))  + (1-\kappa)\phi_{i}(u_{i}(\underline{\sigma}))
\end{align*}	
Thus, equivalently, $u_{i}(\lambda \overline{\sigma} + (1-\lambda)\underline{\sigma}) < U_{i}(\mu^{((\overline{\sigma}, \underline{\sigma}), \kappa)})$.
\end{proof}

\begin{proof}[Proof of Theorem \ref{sufficient_BP}]
Let $\sigma^{BP}$ be obedient and satisfy $u_{s}(\sigma^{BP}) = u_{s}^{BP}$. Suppose that $(\overline{\sigma}, \underline{\sigma}, \lambda)$ is a Pareto-ranked splitting of $\sigma^{BP}$, $u_{i}(\overline{\sigma}) > u_{i}(\underline{\sigma})$ for $i \in \{s, r\}$. For $\kappa$ given by \eqref{equ_thm_Pareto_ranked_improvement}, consider the binary ambiguous experiment $\mu^{((\overline{\sigma}, \underline{\sigma}), \kappa)}$. Its effective weight on $\overline{\sigma}$ is exactly $\lambda$, Lemma \ref{lem_effmeasure} then implies $\mu^{((\overline{\sigma}, \underline{\sigma}), \kappa)}$ is obedient since $\sigma^{BP}$ is. Then $\rho^{\phi_s, u_s}((\overline{\sigma}, \underline{\sigma}), \lambda) < \kappa - \lambda$ implies $U_{s}(\mu^{((\overline{\sigma}, \underline{\sigma}), \kappa)}) > u_{s}^{BP}$ by Lemma \ref{lem_prob-premium_improvement}, thus the sender benefits.

An increase in ambiguity aversion corresponds to replacing $\phi$ by $\tilde{\phi}=\varphi\circ\phi$, where $\varphi$ is strictly increasing, concave, and differentiable
\citep{klibanoff2005smooth}. For the sender, concavity of $\varphi$ gives
\begin{equation*}
\frac{\tilde{\phi}_s(u_s(\sigma^{BP}))
-\tilde{\phi}_s(u_s(\underline{\sigma}))}
{\tilde{\phi}_s(u_s(\overline{\sigma}))
-\tilde{\phi}_s(u_s(\underline{\sigma}))}
\geq
\frac{\phi_s(u_s(\sigma^{BP}))
-\phi_s(u_s(\underline{\sigma}))}
{\phi_s(u_s(\overline{\sigma}))
-\phi_s(u_s(\underline{\sigma}))}.
\end{equation*}
The two sides equal the probability premiums under $\tilde{\phi}_s$ and $\phi_s$, respectively, plus $\lambda$. 

Finally, $\kappa$ in \eqref{equ_thm_Pareto_ranked_improvement} is strictly
increasing in the ratio $\phi_r'(u_r(\underline{\sigma}))/\phi_r'(u_r(\overline{\sigma}))$, and\begin{align*}
\frac{\tilde{\phi}_r'(u_r(\underline{\sigma}))}
{\tilde{\phi}_r'(u_r(\overline{\sigma}))}
=
\frac{\varphi'(\phi_r(u_r(\underline{\sigma})))}
{\varphi'(\phi_r(u_r(\overline{\sigma})))}
\frac{\phi_r'(u_r(\underline{\sigma}))}
{\phi_r'(u_r(\overline{\sigma}))}\geq
\frac{\phi_r'(u_r(\underline{\sigma}))}
{\phi_r'(u_r(\overline{\sigma}))}.
\end{align*}
Thus $\kappa$ increases with the receiver's ambiguity aversion.
\end{proof}

\begingroup
\color{black}
\begin{proof}[Proof of Theorem \ref{general_improvement}]
Let $\mu=\mu^{((\overline{\sigma},\underline{\sigma}),\kappa)}$, with $\kappa$ given by \eqref{equ_wp_general_improvement}. Its effective weights are $\lambda$ and $1-\lambda$, so its effective average is $\sigma^*$. Lemma \ref{lem_effmeasure} then implies that $\mu$ is obedient, proving (i).

For (ii), concavity of $\phi_r$ gives
\begin{align*}
\phi_r(U_r(\mu))-\phi_r(u_r(\sigma^*))\geq\sum_\theta\mu_\theta\phi_r'(u_r(\sigma_\theta))
\bigl[u_r(\sigma_\theta)-u_r(\sigma^*)\bigr] =0,
\end{align*}
where the equality follows from the definition of the effective measure and its average $\sigma^*$. Since both components have positive probability and $u_r(\sigma^*)$ lies strictly between their payoffs, equality in the concavity inequality holds if and only if $\phi_r$ is affine on $[u_r(\underline{\sigma}),u_r(\overline{\sigma})]$. Here differentiability at $u_r(\sigma^*)$ ensures that the two affine pieces must have the same slope. Equivalently, $\phi_r'(u_r(\underline{\sigma}))=\phi_r'(u_r(\overline{\sigma}))$, which by \eqref{equ_wp_general_improvement} is equivalent to $\kappa=\lambda$. As concavity implies $\kappa\geq\lambda$, this proves (ii). Part (iii) follows directly from Lemma \ref{lem_prob-premium_improvement}.

The comparative statics follow from the argument in the proof of Theorem \ref{sufficient_BP}, with $\sigma^*$ in place of $\sigma^{BP}$. That argument uses only the fixed Pareto-ranked splitting and does not require $\sigma^{BP}$ to be optimal for Bayesian persuasion.
\end{proof}
\par
\endgroup

\begingroup
\color{black}
\begin{proof}[Proof of Theorem \ref{thm_receiver_more_ambi_averse}]
        That $\mu$ benefits the sender means that
        \begin{equation}\label{eq:benefit}
            \phi_{s}^{-1}\left(\kappa\phi_{s}(u_{s}(\overline{\sigma})) + (1-\kappa)\phi_{s}(u_{s}(\underline{\sigma}))\right) > u_{s}^{BP}.
        \end{equation}
        Any less ambiguity averse sender will have $\tilde{\phi}_{s}$ weakly less concave than $\phi_{s}$, increasing the left-hand side of (\ref{eq:benefit}), while leaving the right-hand side unchanged. This proves (i).
    
        Let $\sigma^{*} = em^{\mu}_{\overline{\theta}}\overline{\sigma}+(1-em^{\mu}_{\overline{\theta}})\underline{\sigma}$. By Lemma \ref{lem_effmeasure},  $\tau^{*} \in br(\sigma^{*})$. By (\ref{eq:benefit}) and Theorem \ref{thm_necessary_condition_existence_of_Pareto_ranked_splitting}, $\overline{\sigma}$ and $\underline{\sigma}$ must be Pareto-ranked. By our labeling, $\overline{\sigma}$ is the better one. A weakly more ambiguity averse receiver will have a $\tilde{\phi}_{r}=\varphi\circ\phi_{r}$ for some increasing, differentiable and concave $\varphi$, resulting in an effective measure $\tilde{em}^{\mu}$ such that $\tilde{em}^{\mu}_{\overline{\theta}} \leq em^{\mu}_{\overline{\theta}}$. 

        We next prove (iii). Since $\tau^{*} \in br(\underline{\sigma})$, $\tau^{*} \in br(\sigma^{*}) = br(em^{\mu}_{\overline{\theta}} \overline{\sigma} + (1-em^{\mu}_{\overline{\theta}})\underline{\sigma})$ implies $\tau^{*} \in br(\tilde{em}^{\mu}_{\overline{\theta}} \overline{\sigma} + (1-\tilde{em}^{\mu}_{\overline{\theta}})\underline{\sigma})$. By Lemma \ref{lem_effmeasure}, $\tau^{*} \in BR(\mu)$ for any such $\tilde{\phi}_{r}$, proving that $\mu$ continues to benefit the sender (and all less ambiguity averse senders). To prove (ii), if $\tau^{*} \notin br(\underline{\sigma})$, define $\tilde{\mu} = \mu^{((\overline{\sigma}, \underline{\sigma}), \tilde{\kappa})}$ by
        \begin{align*}
            \tilde{\kappa} &= \frac{ em^{\mu}_{\overline{\theta}} \tilde{\phi}_{r}'(u_{r}(\underline{\sigma}))}{ em^{\mu}_{\overline{\theta}}\tilde{\phi}_{r}'(u_{r}(\underline{\sigma})) + (1- em^{\mu}_{\overline{\theta}}) \tilde{\phi}_{r}'(u_{r}(\overline{\sigma}))},
        \end{align*}
        so that $\tilde{em}^{\tilde{\mu}} = em^{\mu}$. Given $\tilde{\phi}_{r}$, Lemma \ref{lem_effmeasure} and $\tau^{*} \in br(\sigma^{*})$ imply $\tau^{*} \in BR(\tilde{\mu})$. Since $u_{s}(\overline{\sigma}) > u_{s}(\underline{\sigma})$, $\tilde{\kappa} \geq \kappa$ ensures (ii).
        \begin{align*}
            \tilde{\kappa}& = \frac{em^{\mu}_{\overline{\theta}} \phi_{r}'(u_{r}(\underline{\sigma}))}{em^{\mu}_{\overline{\theta}} \phi_{r}'(u_{r}(\underline{\sigma})) + (1-em^{\mu}_{\overline{\theta}}) \frac{\varphi'(\phi_{r}(u_{r}(\overline{\sigma})))}{\varphi'(\phi_{r}(u_{r}(\underline{\sigma})))}\phi_{r}'(u_{r}(\overline{\sigma}))}\\
            & \geq \frac{em^{\mu}_{\overline{\theta}} \phi_{r}'(u_{r}(\underline{\sigma}))}{em^{\mu}_{\overline{\theta}} \phi_{r}'(u_{r}(\underline{\sigma})) + (1-em^{\mu}_{\overline{\theta}}) \phi_{r}'(u_{r}(\overline{\sigma}))} = \kappa,
        \end{align*}
        where the inequality follows from $\varphi'(\phi_{r}(u_{r}(\underline{\sigma}))) \geq \varphi'(\phi_{r}(u_{r}(\overline{\sigma})))$.
    \end{proof}
\par
\endgroup

\subsection{Proofs for Section \ref{sec_threshold_games}}

\begin{proof}[Proof of Lemma \ref{lem:belief_threshold_payoff_representation}]

\textbf{Only If.} Fix any $i > 0$. By part (ii) of Definition \ref{def:belief_threshold_game}, the receiver is indifferent between $a_{i}$ and $a_{0}$ at any belief $q$ with $q(\omega_{i}) = T_{i}$. This indifference must hold in particular for $q = T_{i} \delta_{\omega_{i}} + (1 - T_{i})\delta_{\omega_{j}}$ for every $j \neq i$. This is equivalent to
\begin{align*}
T_{i}(u_{r}(a_{i}, \omega_{i}) - u_{r}(a_{0}, \omega_{i})) + (1 - T_{i})(u_{r}(a_{i}, \omega_{j}) - u_{r}(a_{0}, \omega_{j})) = 0.
\end{align*}
Since $a_{i}$ is strictly optimal at $\omega_{i}$, \eqref{equ:belief_threshold_payoff_representation} holds with $k_{i}$ defined as
\begin{align*}
k_{i} := \frac{u_{r}(a_{i}, \omega_{i}) - u_{r}(a_{0}, \omega_{i})}{1-T_{i}} = - \frac{u_{r}(a_{i}, \omega_{j}) - u_{r}(a_{0}, \omega_{j})}{T_{i}} > 0.
\end{align*}

\noindent \textbf{If.} Fix $k_{i} > 0$, $T_{i} \in (1/2, 1)$, $i> 0$ and suppose the receiver's payoffs satisfy \eqref{equ:belief_threshold_payoff_representation}. Then observing that maximizing $\sum_{\omega\in\Omega}q(\omega)u_r(a,\omega)$ over $a \in A$ is equivalent to maximizing $\sum_{\omega\in\Omega}q(\omega)(u_r(a,\omega) - u_r(a_{0}, \omega))$, and substituting \eqref{equ:belief_threshold_payoff_representation} yields Definition \ref{def:belief_threshold_game}. 
\end{proof}

\begin{proof}[Proof of Theorem \ref{thm:belief_threshold_iff}]
\textbf{If.} For any obedient experiment $\sigma$, 
\begin{align*}
q(\omega_{i}|a_{i}) = \frac{p(\omega_{i})\sigma(a_{i}|\omega_{i})}{\sum_{\omega}p(\omega)\sigma(a_{i}|\omega)} \geq T_{i}.
\end{align*}
Since $\sigma(a_{i}|\omega_{i}) \leq 1$, it follows that $\sum_{\omega}p(\omega)\sigma(a_{i}|\omega) \leq p(\omega_{i})/T_{i}$. Therefore, every obedient experiment $\sigma$ satisfies $u_{s}(\sigma, \tau^{*}) = \sum_{i=1}^{n}v_{i}\sum_{\omega}p(\omega)\sigma(a_{i}|\omega) \leq \sum_{i=1}^{n}v_{i}p(\omega_{i})/T_{i}$. When $\sum_{i=1}^{n}p(\omega_{i})/T_{i} < 1$, this upper bound can be achieved by $\sigma^{BP}$ defined as follows:
\begin{align*}
& \sigma^{BP}(a_{i} | \omega_{i}) = 1, \quad \sigma^{BP}(a_{i}| \omega_{0}) = \frac{p(\omega_{i})(1 - T_{i})}{p(\omega_{0})T_{i} }, \quad \text{for all } i \geq 1, \\
& \sigma^{BP}(a_{0}|\omega_{0}) = 1 - \sum_{i=1}^{n}\frac{p(\omega_{i})(1 - T_{i})}{p(\omega_{0})T_{i} },
\end{align*}
and all other probabilities are zero. To see this is feasible, notice that
\begin{align*}
&\sum_{i=1}^{n}\frac{p(\omega_{i})(1 - T_{i})}{p(\omega_{0})T_{i} } < 1
\Longleftrightarrow \sum_{i=1}^{n}\frac{p(\omega_{i})}{T_{i}} - \sum_{i=1}^{n}p(\omega_{i}) < p(\omega_{0})
\Longleftrightarrow \sum_{i=1}^{n}\frac{p(\omega_{i})}{T_{i}} < 1.
\end{align*}
Moreover, $\sigma^{BP}$ is obedient because $q(\omega_{i}|a_{i}) = T_{i}$ for all $i\geq 1$, and satisfies
\begin{align*}
u_{s}(\sigma^{BP}) = \sum_{i=1}^{n}v_{i}\sum_{\omega}p(\omega)\sigma^{BP}(a_{i}|\omega) = \sum_{i=1}^{n}v_{i}\frac{p(\omega_{i})}{T_{i}}.
\end{align*}

We next use the following lemma to show that a Pareto-ranked splitting of $\sigma^{BP}$ exists. 
\begin{lemma}\label{prop:inefficiency_improvement}
Given any experiment $\sigma^{*}$, fix, for each $\omega \in \Omega$, $a_{\omega} \in \mathop{supp}(\sigma^{*}(\cdot|\omega))$ and consider the following set of vectors, 
\begin{equation}\label{equ:inefficiency_improvement}
\left \{ 
\begin{bmatrix}
u_{s}(a, \omega) - u_{s}(a_{\omega}, \omega)\\
u_{r}(a, \omega) - u_{r}(a_{\omega}, \omega)
\end{bmatrix}: a \in \mathop{supp}(\sigma^{*}(\cdot|\omega)), \omega \in \Omega\right \}.
\end{equation}
If this set spans $\R^{2}$, then there exists a Pareto-ranked splitting of $\sigma^{*}$. 
\end{lemma}

Let $a_{\omega_{0}} = a_{0}$, the payoff-difference vectors in state $\omega_{0}$ are positive multiples of
\begin{align*}
\begin{bmatrix}
v_{i}\\
-k_{i}T_{i}
\end{bmatrix} \text{ for } i \geq 1.
\end{align*}
Two such vectors are linearly independent and hence span $\R^{2}$ if and only if there exist $i, j$ such that $v_{i}/k_{i}T_{i} \neq v_{j}/k_{j}T_{j}$. Therefore, a Pareto-ranked splitting of $\sigma^{BP}$ exists. Note that such existence requires $n \geq 2$, so there is no conflict with Corollary \ref{cor_binary_action}. Since $\phi_{r}$ is strictly concave, the sufficient conditions in Theorem \ref{sufficient_BP} are satisfied. Therefore, ambiguous communication benefits the sender.

\noindent\textbf{Only if.} Suppose $v_{i}/k_{i}T_{i} = c$ for all $i > 0$. Consider any non-obedient experiment $\hat{\sigma}$ with $u_{s}(\hat{\sigma}) > u_{s}^{BP}$. Then using $v_{i}/k_{i}T_{i} = c$, we have
\begin{align*}
\sum_{i=1}^{n}\sum_{\omega}p(\omega)\hat{\sigma}(a_{i}|\omega)v_{i} > \sum_{i=1}^{n}\frac{p(\omega_{i})}{T_{i}}v_{i} \implies \sum_{i=1}^{n}k_{i}T_{i}\sum_{\omega}p(\omega)\hat{\sigma}(a_{i}|\omega) > \sum_{i=1}^{n}p(\omega_{i})k_{i}. 
\end{align*}

We next show $u_{r}(\hat{\sigma})$ is strictly worse than always taking action $a_{0}$:
\begin{align*}
u_{r}(\hat{\sigma})  - \sum_{\omega}p(\omega)u_{r}(a_{0}, \omega) & = \sum_{i=1}^{n}\sum_{\omega}p(\omega)\hat{\sigma}(a_{i}|\omega)(u_{r}(a_{i}, \omega) - u_{r}(a_{0}, \omega)) \\
& \leq \sum_{i=1}^{n}p(\omega_{i})k_{i} - \sum_{i=1}^{n}k_{i}T_{i}\sum_{\omega}p(\omega)\hat{\sigma}(a_{i}|\omega) < 0.
\end{align*}
Since the receiver can always choose action $a_{0}$, any obedient $\sigma$ must satisfy $u_{r}(\sigma) \geq \sum_{\omega}p(\omega)u_{r}(a_{0}, \omega)$. Thus no experiment giving the sender more than $u_s^{BP}$ can Pareto-dominate an obedient experiment, violating the necessary condition in Theorem \ref{thm_necessary}.
\end{proof}

\begin{proof}[Proof of Lemma \ref{prop:inefficiency_improvement}]
Let $\Sigma_{\sigma^{*}} \subseteq \Sigma$ denote the set of experiments having the same support as $\sigma^{*}$ for each $\omega \in \Omega$. For each $\omega \in \Omega$, fix any $a_{\omega} \in \mathop{supp}(\sigma^{*}(\cdot|\omega))$. For any $\sigma \in \Sigma_{\sigma^{*}}$, by substituting $\sigma(a_{\omega}|\omega) = 1 - \sum_{a \neq a_{\omega}} \sigma(a|\omega)$, we have, for $i \in \{s,r\}$, 
\begin{align*}
u_{i}(\sigma, \tau^{*}) & = \sum_{\omega,a} p(\omega)\sigma(a|\omega)u_{i}(a,\omega) = \sum_{\omega \in \Omega_{\sigma^{*}}, a} p(\omega)\sigma(a|\omega)(u_{i}(a,\omega) - u_{i}(a_{\omega}, \omega)) + \sum_{\omega} p(\omega)u_{i}(a_{\omega}, \omega),
\end{align*}
where $\Omega_{\sigma^{*}} \subseteq \Omega$ denotes the set of $\omega$ such that $|\mathop{supp}(\sigma^{*}(\cdot|\omega))| > 1$. Given any $\sigma \in \Sigma_{\sigma^{*}}$, we use $\tilde{\sigma} \in \R^{\sum_{\omega \in \Omega_{\sigma^{*}}}(|\mathop{supp}(\sigma^{*}(\cdot|\omega))|-1)} =: \R^{\sigma^{*}}$ to denote the vector of those components of $\sigma(a|\omega)$ with $\omega \in \Omega_{\sigma^{*}}$ and $a \in \mathop{supp}(\sigma^{*}(\cdot|\omega)) \setminus \{a_{\omega}\}$. Thus, we can write 
\begin{align*}
\begin{bmatrix}
u_{s}(\sigma, \tau^{*})\\
u_{r}(\sigma, \tau^{*})
\end{bmatrix} 
& = \begin{bmatrix}
p(\omega)(u_{s}(a, \omega) - u_{s}(a_{\omega}, \omega)) & \cdots \\
p(\omega)(u_{r}(a, \omega) - u_{r}(a_{\omega}, \omega)) & \cdots
\end{bmatrix} \tilde{\sigma} +         \begin{bmatrix}
    \sum_{\omega} p(\omega)u_{s}(a_{\omega}, \omega)\\
    \sum_{\omega} p(\omega)u_{r}(a_{\omega}, \omega)
\end{bmatrix}
\end{align*}
Notice that any non-zero vectors in \eqref{equ:inefficiency_improvement} are positive multiples of non-zero columns of the first matrix on the right-hand side above. When the former set spans $\R^{2}$, the latter matrix has full rank and thus the affine mapping from $\R^{\sigma^{*}}$ to $\R^{2}$ defined above is surjective. Since $\{\tilde{\sigma}: \sigma \in \Sigma_{\sigma^{*}}\} \ni \tilde{\sigma^{*}}$ is open in $\R^{\sigma^{*}}$, by the open mapping theorem, $\{(u_{s}(\sigma, \tau^{*}), u_{r}(\sigma, \tau^{*})): \sigma \in \Sigma_{\sigma^{*}}\}$ is open. Thus, there exists $\hat{\sigma} \in \Sigma_{\sigma^{*}}$ such that $u_{i}(\hat{\sigma}, \tau^{*}) > u_{i}(\sigma^{*}, \tau^{*})$ for $i \in \{s,r\}$. By Lemma \ref{lem_existence_splitting}, there exists a Pareto-ranked splitting of $\sigma^{*}$ with $\overline{\sigma} = \hat{\sigma}$. 
\end{proof}

\begin{remark}[Non-necessity]\label{rem:spanning_condition_not_necessary}
\emph{That the spanning condition in Lemma \ref{prop:inefficiency_improvement} is not necessary for the existence of a Pareto-ranked splitting can be seen from our introductory example. For $\sigma^{*} = \sigma^{BP}$, \eqref{equ:inefficiency_improvement} does not span $\R^{2}$, but there are, as depicted in Figure \ref{fig:intro_example_2}, Pareto-ranked splittings of $\sigma^{BP}$. \qedsymbol}
\end{remark}

\begingroup
\color{black}
\begin{proof}[Proof of Example \ref{ex:belief_threshold_prior_condition}]
For any obedient experiment $\sigma$, the probability of recommending $a_i$ is at most $p(\omega_i)/T_i=1/2$ for each $i\in\{1,2\}$. Thus, the sender's expected payoff is at most $5/2+3/2=4$. The optimal Bayesian persuasion experiment $\sigma^{BP}$ is given by
\begin{align*}
    & \sigma^{BP}(a_{1} | \omega_{1}) = 1, \quad \sigma^{BP}(a_{2} | \omega_{2}) = 1,\\
    & \sigma^{BP}(a_{1}| \omega_{0}) = 1/2, \quad \sigma^{BP}(a_{2}| \omega_{0}) = 1/2,
\end{align*}
with all other probabilities zero. This experiment is obedient and gives expected payoffs $u_s(\sigma^{BP})=u_s^{BP}=4$ and $u_r(\sigma^{BP})=2$.

To see that the payoffs under $\sigma^{BP}$ are ex-ante efficient, observe that, for any experiment $\sigma$,
\begin{align*}
u_r(\sigma)+\frac{3}{8}u_s(\sigma)
&\leq \sum_{\omega}p(\omega)\max_{a\in A}\left\{u_r(a,\omega)+\frac{3}{8}u_s(a,\omega)\right\}\\
&=\frac{7}{2}=u_r(\sigma^{BP})+\frac{3}{8}u_s(\sigma^{BP}).
\end{align*}
Consider any experiment $\hat{\sigma}$ that is better for the sender than $\sigma^{BP}$. By the preceding inequality, it must be worse than $\sigma^{BP}$ for the receiver. However, under $\sigma^{BP}$, the receiver's expected payoff is the same as they would get from playing $a_{0}$, which they could do with no information from the sender. Therefore, $\hat{\sigma}$ cannot be better for the receiver than any obedient experiment. It follows that the necessary condition in Theorem \ref{thm_necessary} cannot be satisfied, and thus the sender cannot benefit from ambiguous communication.
\end{proof}
\par
\endgroup

\begingroup
\color{black}
\begin{proof}[Proof of Proposition \ref{prop:belief_threshold_critical_aversion}]
Fix a Pareto-ranked splitting $(\overline{\sigma},\underline{\sigma},\lambda)$ of $\sigma^{BP}$ and let $\kappa$ be given by \eqref{equ_thm_Pareto_ranked_improvement}. By Theorem \ref{sufficient_BP}, the sender benefits from ambiguous communication if $\rho^{\phi_s,u_s}((\overline{\sigma},\underline{\sigma}),\lambda)<\kappa-\lambda$. Let $R=\phi_r'(u_r(\underline{\sigma}))/\phi_r'(u_r(\overline{\sigma}))$. Then this inequality is equivalent to
\begin{equation}\label{equ:belief_threshold_critical_aversion}
\frac{R-1}{\lambda R+(1-\lambda)}>
\frac{\rho^{\phi_s,u_s}((\overline{\sigma},\underline{\sigma}),\lambda)}{\lambda(1-\lambda)}.
\end{equation}
Since $\phi_s$ is strictly increasing and $u_s(\sigma^{BP})<u_s(\overline{\sigma})$, the right-hand side is strictly bounded above by
\begin{align*}
&\frac{\phi_s(u_s(\overline{\sigma}))-\lambda\phi_s(u_s(\overline{\sigma}))-(1-\lambda)\phi_s(u_s(\underline{\sigma}))}
{\lambda(1-\lambda)\big(\phi_s(u_s(\overline{\sigma}))-\phi_s(u_s(\underline{\sigma}))\big)}
=\frac{1}{\lambda}.
\end{align*}
The left-hand side of \eqref{equ:belief_threshold_critical_aversion} is strictly increasing in $R$ and approaches $1/\lambda$ as $R\rightarrow\infty$. If $\phi_r(x)=-e^{-\alpha x}$, then
\begin{equation*}
R=e^{\alpha\left(u_r(\overline{\sigma})-u_r(\underline{\sigma})\right)}\rightarrow\infty
\quad\text{as }\alpha\rightarrow\infty.
\end{equation*}
Thus, there exists $\hat{\alpha}>0$ such that the strict inequality \eqref{equ:belief_threshold_critical_aversion} holds for all $\alpha\geq\hat{\alpha}$. Finally, $\kappa$ is increasing in the receiver's ambiguity aversion by Theorem \ref{sufficient_BP}. Therefore, the same strict inequality holds whenever $\phi_r$ is at least as concave as $-e^{-\hat{\alpha}x}$. Theorem \ref{sufficient_BP} and Lemma \ref{lem_effmeasure} then imply that $\mu^{((\overline{\sigma},\underline{\sigma}),\kappa)}$ is obedient and benefits the sender.
\end{proof}
\par
\endgroup

\subsection{Proofs for Section \ref{sec_properties_optimal_ambiguous_communication}}\label{sec_proof_properties_optimal_ambiguous_communication}

We prove Proposition  \ref{thm_optimal_persuasion_characterization_experiment_space} first and then use it to prove Theorem \ref{thm_optimal_persuasion_pareto_ranked_experiments}.

\begin{proof}[Proof of Proposition \ref{thm_optimal_persuasion_characterization_experiment_space}]

There is a unique $u$ such that $V(\Phi^{*}_{u}) = 0$:
\begin{lemma}\label{lem_single_crossing_experiment_space}
$V(\Phi^{*}_{u})$ satisfies single-crossing, i.e., for any $u > u'$, if $V(\Phi^{*}_{u}) \geq 0$, then $V(\Phi^{*}_{u'}) > 0$. Thus, there exists a unique $u \in \R$ such that $V(\Phi^{*}_{u}) = 0$. 
\end{lemma}

\begin{proof}[Proof of Lemma \ref{lem_single_crossing_experiment_space}]

$V(\Phi^{*}_{u}) \geq 0 $ and $\phi_{s}, \phi_{r}$ strictly increasing imply $V(\Phi^{*}_{u'}) > 0$ for $u' < u$. Thus, there is a unique $u$ such that $V(\Phi^{*}_{u}) = 0$. By Caratheodory's theorem, the maximum in \eqref{equ_program_Phi_star} can be achieved within a compact set, continuity of $\Phi_{u}$ and Berge's Maximum Theorem imply $V(\Phi^{*}_{u})$ is continuous. Finally, $V(\Phi^{*}_{u}) < 0$ for $u > \max u_{s}(a,\omega)$ and $V(\Phi^{*}_{u}) > 0$ for $u < \min u_{s}(a,\omega)$ imply the existence of $u$ such that $V(\Phi^{*}_{u}) = 0$.
\end{proof}

By Proposition \ref{rev-principle}, \eqref{eq:sender_problem} is equivalent to $\max_{\mu} U_{s}(\mu)$ subject to $\tau^{*} \in BR(\mu)$.

\begin{lemma}\label{lem_characterization_improving_u}
For each $u \in \R$, the value of \eqref{eq:sender_problem} is greater than $u$ if and only if $V(\Phi^{*}_{u}) > 0$. 
\end{lemma}

\begin{proof}[Proof of Lemma \ref{lem_characterization_improving_u}]
By Lemma \ref{lem_effmeasure} and
\eqref{equ_reverse_effmeasure}, $\lambda=em^\mu$
gives a one-to-one correspondence between obedient $\mu$
and feasible $\lambda$ for \eqref{equ_program_Phi_star}.
For each corresponding pair,
\begin{equation*}
\phi_s(U_s(\mu))-\phi_s(u)
=
\frac{\sum_{\theta}\lambda_{\theta}\Phi_u(\sigma_{\theta})}
{\sum_{\theta}\lambda_{\theta}/\phi_r'(u_r(\sigma_{\theta}))}.
\end{equation*}
Thus $U_s(\mu)>u$ if and only if $\sum_{\theta}\lambda_{\theta}\Phi_u(\sigma_{\theta})>0$.
Existence of such feasible choices therefore gives the value of \eqref{eq:sender_problem} above $u$ if and only if $V(\Phi_{u}^*)>0$.
\end{proof}

Let $u$ be the unique zero of $V(\Phi^*_{u})$. By Lemmas \ref{lem_single_crossing_experiment_space} and \ref{lem_characterization_improving_u}, for every $u'\in\R$, the value of \eqref{eq:sender_problem} is greater than $u' \Longleftrightarrow V(\Phi^*_{u'})>0 \Longleftrightarrow u>u'$. Hence the value of \eqref{eq:sender_problem} is $u$ if and only if $V(\Phi^*_{u})=0$.

Let $u$ be the value of \eqref{eq:sender_problem} and choose a maximizer $\lambda$ of $\eqref{equ_program_Phi_star}$ with $\sum_\theta\lambda_\theta\sigma_\theta=\sigma^*\in\Sigma^*$. Since $V(\Phi^*_{u})=0$, $(\sigma^*,0)$ is a boundary point of the convex hull of the graph of $\Phi_u$ in $\R^{|\Omega|(|A|-1)+1}$. Every graph point receiving positive weight under $\lambda$ lies on a supporting hyperplane through $(\sigma^*,0)$. Caratheodory's theorem applied within this hyperplane therefore gives a representation using at most $|\Omega|(|A|-1)+1$ graph points.
Applying \eqref{equ_reverse_effmeasure} yields an optimal $\mu^*$ with $|\text{supp}(\mu^*)|\leq|\Omega|(|A|-1)+1$.
\end{proof}

\begin{proof}[Proof of Theorem \ref{thm_optimal_persuasion_pareto_ranked_experiments}]
Suppose $\mu$ is optimal and $\phi_r$ is strictly concave. Let $\hat{u}:=U_s(\mu)$ and $\lambda:=em^\mu$. By the preceding characterization, $\sum_\theta\lambda_\theta\Phi_{\hat{u}}(\sigma_\theta)=0=V(\Phi^*_{\hat{u}})$, so $\lambda$ is a solution. Since $\lambda$ and $\mu$ have the same support, strict concavity of $\phi_r$ and the corresponding parts of the following lemma (and its swapped version) imply parts
(i) and (ii). 
\begin{definition}\label{def_big_sigmas}
Let $\Sigma_{+}(u) = \{\sigma \in \Sigma: u_{s}(\sigma) > u\}$ and $\Sigma_{-}(u) = \{\sigma \in \Sigma: u_{s}(\sigma) \leq u\}$. 
\end{definition}

\begin{lemma}\label{lem:inversely_ranked_pareto_ranked}
Let $\lambda\in \Delta_{f}(\Sigma)$ be a solution to \eqref{equ_program_Phi_star}. Then the following hold: 
\begin{enumerate}[(i)]
\item For all $\sigma_{\theta}, \sigma_{\theta'} \in \text{supp}(\lambda)$ such that $(\sigma_{\theta}, \sigma_{\theta'}) \in \Sigma_{+}(u) \times \Sigma_{-}(u)$, if $\phi'_{r}(u_{r}(\sigma_{\theta})) \neq \phi'_{r}(u_{r}(\sigma_{\theta'}))$, then they are Pareto-ranked,
\item If $\phi_{s}$ is affine, then for all $\sigma_{\theta} \in \text{supp}(\lambda)$, there does not exist a Pareto-ranked splitting of $\sigma_{\theta}$, $(\overline{\sigma}, \underline{\sigma}, \alpha)$ with $(\overline{\sigma}, \underline{\sigma}) \in \Sigma_{+}(u) \times \Sigma_{-}(u)$ and $\phi'_{r}(u_{r}(\overline{\sigma})) < \phi_r'(u_r(\sigma_{\theta}))< \phi'_{r}(u_{r}(\underline{\sigma}))$.
\end{enumerate}
\end{lemma}

A swapped version of Lemma \ref{lem:inversely_ranked_pareto_ranked} with $\Sigma_{+}(u)$ and $\Sigma_{-}(u)$ defined by swapping the strict and weak inequalities in Definition \ref{def_big_sigmas} also holds, and the proof is identical.
\end{proof}

\begin{proof}[Proof of Lemma \ref{lem:inversely_ranked_pareto_ranked}]
Fix $u \in \R$. Let $\lambda \in \Delta_{f}(\Sigma)$ be feasible for \eqref{equ_program_Phi_star}. To prove (i), suppose there exist $\sigma_{\theta}, \sigma_{\theta'} \in \text{supp}(\lambda)$ and $\alpha \in (0,1)$ such that $\Phi_{u} (\alpha\sigma_{\theta} + (1-\alpha)\sigma_{\theta'}) > \alpha \Phi_{u}(\sigma_{\theta}) + (1-\alpha) \Phi_{u}(\sigma_{\theta'})$. Then, $\lambda$ cannot be a solution to \eqref{equ_program_Phi_star}. To see this, if $\lambda_{\theta}/\alpha \leq \lambda_{\theta'}/(1-\alpha)$, then replacing $\sigma_{\theta}$ by $\alpha\sigma_{\theta} + (1-\alpha)\sigma_{\theta'}$, letting $\hat\lambda_{\theta}=\lambda_{\theta}/\alpha$, and $\hat\lambda_{\theta'}=\lambda_{\theta'}-(1-\alpha)\lambda_{\theta}/\alpha$ yields a strict improvement. The other case is symmetric.

Towards a contradiction, suppose in the solution there exists $(\sigma_{\theta}, \sigma_{\theta'}) \in \Sigma_{+}(u) \times \Sigma_{-}(u)$ with $\phi'_{r}(u_{r}(\sigma_{\theta})) \neq \phi'_{r}(u_{r}(\sigma_{\theta'}))$ and they are not weakly Pareto-ranked, i.e., $u_{s}(\sigma_{\theta}) > u \geq u_{s}(\sigma_{\theta'})$, and $u_{r}(\sigma_{\theta}) < u_{r}(\sigma_{\theta'})$. Then there exists $\alpha \in (0,1)$ such that $ \phi_{r}'(u_{r}(\sigma_{\theta})) > \phi_{r}'(u_{r}(\alpha \sigma_{\theta} + (1-\alpha) \sigma_{\theta'}))$ (by differentiability of $\phi_{r}$), and  $\Phi_{u}(\alpha \sigma_{\theta} + (1-\alpha)\sigma_{\theta'})> \alpha \Phi_{u}(\sigma_{\theta}) + (1-\alpha) \Phi_{u}(\sigma_{\theta'})$. To see the last point, notice that 
\begin{align*}
\Phi_{u}(\alpha \sigma_{\theta} + (1-\alpha)\sigma_{\theta'}) & = \frac{\phi_{s}(u_{s}(\alpha \sigma_{\theta} + (1-\alpha) \sigma_{\theta'})) - \phi_{s}(u)}{\phi_{r}'(u_{r}(\alpha \sigma_{\theta} + (1-\alpha) \sigma_{\theta'}))}\\
& \geq \frac{\alpha [\phi_{s}(u_{s}( \sigma_{\theta})) - \phi_{s}(u)] + (1-\alpha) [\phi_{s}(u_{s}( \sigma_{\theta'})) - \phi_{s}(u)]}{\phi_{r}'(u_{r}(\alpha \sigma_{\theta} + (1-\alpha) \sigma_{\theta'}))}\\
& > \alpha \Phi_{u}(\sigma_{\theta}) + (1-\alpha) \Phi_{u}(\sigma_{\theta'}), 
\end{align*}
where the first inequality follows from concavity of $\phi_{s}$, the second follows from concavity of $\phi_{r}$, $\Phi_{u}(\sigma_{\theta}) > 0$, $\Phi_{u}(\sigma_{\theta'}) \leq 0$, and $ \phi_{r}'(u_{r}(\sigma_{\theta})) > \phi_{r}'(u_{r}(\alpha \sigma_{\theta} + (1-\alpha) \sigma_{\theta'}))$.

To prove (ii), if there exist $\sigma_{\theta} \in \text{supp}(\lambda)$ and two experiments $\sigma$ and $\sigma'$ such that $\sigma_{\theta} = \alpha \sigma + (1-\alpha) \sigma'$ for some $\alpha \in (0,1)$ and $\Phi_{u} (\alpha \sigma + (1-\alpha)\sigma')< \alpha \Phi_{u}(\sigma) + (1-\alpha) \Phi_{u}(\sigma')$, then $\lambda$ cannot be a solution to \eqref{equ_program_Phi_star}. This follows by noting that splitting $\sigma_{\theta}$ into $\sigma$ with probability $\alpha \lambda_{\theta}$ and $\sigma'$ with probability $(1-\alpha)\lambda_{\theta}$ induce a strict improvement. 

Towards a contradiction, suppose there exists such a Pareto-ranked splitting $(\overline{\sigma}, \underline{\sigma},\alpha)$ with $\alpha \overline{\sigma} + (1-\alpha)\underline{\sigma} \in \Sigma_{+}(u)$ (the other case is symmetric), then we have 
\begin{align*}
&\alpha\Phi_u(\overline{\sigma})
+(1-\alpha)\Phi_u(\underline{\sigma})
-\Phi_u(\sigma_\theta)\\
\geq &
\frac{\alpha}{\phi_r'(u_r(\overline{\sigma}))}
\left[
\phi_s(u_s(\overline{\sigma}))
-\phi_s(u_s(\sigma_\theta))
\right]+
\frac{1-\alpha}{\phi_r'(u_r(\sigma_\theta))}
\left[
\phi_s(u_s(\underline{\sigma}))
-\phi_s(u_s(\sigma_\theta))
\right]\\
\geq&
\alpha(1-\alpha)
\bigl(u_s(\overline{\sigma})-u_s(\underline{\sigma})\bigr)
\left[
\frac{\phi_s'(u_s(\overline{\sigma}))}
{\phi_r'(u_r(\overline{\sigma}))}
-
\frac{\phi_s'(u_s(\underline{\sigma}))}
{\phi_r'(u_r(\sigma_\theta))}
\right]>0,
\end{align*}
where the first inequality follows from $u_s(\sigma_\theta)>u\geq u_s(\underline{\sigma})$
and concavity of $\phi_r$, the second from concavity of $\phi_s$ and $\sigma_\theta=\alpha\overline{\sigma}+(1-\alpha)\underline{\sigma}$, and the last from
linearity of $\phi_s$ and $\phi_r'(u_r(\overline{\sigma}))<\phi_r'(u_r(\sigma_\theta))$.
\end{proof}

\begin{remark}\label{rem:linear_phi_s_relaxation}
\emph{From the proof of Lemma \ref{lem:inversely_ranked_pareto_ranked} (ii), linearity of $\phi_{s}$ can be relaxed to $\frac{\phi'_{s}(u_{s}(\overline{\sigma}, \tau^{*}))}{\phi'_{s}(u_{s}(\underline{\sigma}, \tau^{*}))} > \frac{\phi'_{r}(u_{r}(\overline{\sigma}, \tau^{*}))}{\phi'_{r}(u_{r}(\sigma_{\theta}, \tau^{*}))}$, if $\sigma_{\theta} \in \Sigma_{+}(u)$, and $\frac{\phi'_{s}(u_{s}(\overline{\sigma}, \tau^{*}))}{\phi'_{s}(u_{s}(\underline{\sigma}, \tau^{*}))} > \frac{\phi'_{r}(u_{r}(\sigma_{\theta}, \tau^{*}))}{\phi'_{r}(u_{r}(\underline{\sigma}, \tau^{*}))}$, if $\sigma_{\theta} \in \Sigma_{-}(u)$. \qedsymbol}
\end{remark}

\begin{proof}[Proof of Proposition \ref{cor_linear_phi_inverse_concavification}]

Proposition \ref{cor_linear_phi_inverse_concavification} is implied by the following two auxiliary lemmas that extend Lemma \ref{lem:inversely_ranked_pareto_ranked}. The proofs of these lemmas are similar to that of Lemma \ref{lem:inversely_ranked_pareto_ranked} and may be found in \textcolor{black}{Appendix \ref{sec_proofs_auxiliary_results}}. 

\begin{lemma}\label{thm_inverstly_ranked_concavification_experiment_space}
Let $\lambda \in \Delta_{f}(\Sigma)$ be a solution to \eqref{equ_program_Phi_star}. Then: 
\begin{enumerate}[(i)]                
\item If $1/\phi'_{r}$ is concave, then for all $\sigma_{\theta}, \sigma_{\theta'} \in \text{supp}(\lambda)$ such that $(\sigma_{\theta}, \sigma_{\theta'}) \in \Sigma_{+}(u) \times \Sigma_{+}(u)$, if $u_{s}(\sigma_{\theta}) \neq  u_{s}(\sigma_{\theta'})$ and $\phi'_{r}(u_{r}(\sigma_{\theta})) \neq \phi'_{r}(u_{r}(\sigma_{\theta'}))$, they are Pareto-ranked.

\item If $1/\phi'_{r}$ is convex, then for all $\sigma_{\theta}, \sigma_{\theta'} \in \text{supp}(\lambda)$ such that $(\sigma_{\theta}, \sigma_{\theta'}) \in \Sigma_{-}(u) \times \Sigma_{-}(u)$, if $u_{s}(\sigma_{\theta}) \neq  u_{s}(\sigma_{\theta'})$ and $\phi'_{r}(u_{r}(\sigma_{\theta})) \neq \phi'_{r}(u_{r}(\sigma_{\theta'}))$, they are Pareto-ranked. 
\end{enumerate}
\end{lemma}

\begin{lemma}\label{thm_Pareto_ranked_concavification_experiment_space}
Let $\lambda \in \Delta_{f}(\Sigma)$ be a solution to \eqref{equ_program_Phi_star}. Then:  
\begin{enumerate}[(i)]                
\item If $1/\phi'_{r}$ is concave, then for all $\sigma_{\theta} \in \text{supp}(\lambda)$, there does not exist a Pareto-ranked splitting of $\sigma_{\theta}$, $(\overline{\sigma}, \underline{\sigma}, \alpha)$ with $(\overline{\sigma}, \underline{\sigma}) \in \Sigma_{-}(u) \times \Sigma_{-}(u)$, $\phi'_{r}(u_{r}(\overline{\sigma})) < \phi'_{r}(u_{r}(\underline{\sigma}))$, and $\phi'_{s}(u_{s}(\overline{\sigma}))/\phi'_{s}(u_{s}(\underline{\sigma})) > \phi'_{r}(u_{r}(\overline{\sigma}))/\phi'_{r}(u_{r}(\underline{\sigma}))$. 

\item If $1/\phi'_{r}$ is convex, then for all $\sigma_{\theta} \in \text{supp}(\lambda)$, there does not exist a Pareto-ranked splitting of $\sigma_{\theta}$, $(\overline{\sigma}, \underline{\sigma}, \alpha)$ with $(\overline{\sigma}, \underline{\sigma}) \in \Sigma_{+}(u) \times \Sigma_{+}(u)$, $\phi'_{r}(u_{r}(\overline{\sigma})) < \phi'_{r}(u_{r}(\underline{\sigma}))$, and $\phi'_{s}(u_{s}(\overline{\sigma}))/\phi'_{s}(u_{s}(\underline{\sigma})) > \phi'_{r}(u_{r}(\overline{\sigma}))/\phi'_{r}(u_{r}(\underline{\sigma}))$. 
\end{enumerate}
\end{lemma}
With their swapped versions, when $1/\phi'_{r}$ is affine, Lemma \ref{lem:inversely_ranked_pareto_ranked} (i) and \ref{thm_inverstly_ranked_concavification_experiment_space} imply the first conclusion. If $\phi_{s}$ is also affine, Lemmas \ref{lem:inversely_ranked_pareto_ranked} (ii) and \ref{thm_Pareto_ranked_concavification_experiment_space} imply the second conclusion.
\end{proof}

\subsection{Proofs for Section \ref{sec_discussion}}

\begin{proof}[Proof of Theorem \ref{thm:benefit_inconsistent_iff}]
Fix belief-threshold game with $\sum_{i=1}^{n}p(\omega_{i})/T_{i} < 1$. Let $\sigma^{BP}$ be defined as in the proof of Theorem \ref{thm:belief_threshold_iff}. Let $a_{\bar{i}}$ be the safe action and normalize $u_{r}(a_{\bar{i}}, \omega) = d$. 

\textbf{If $\bar{i} > 0$.} Without loss of generality, let $\bar{i} = 1$. We show the sender benefits by using the ambiguous experiment $\mu$ such that $\mu(\sigma_{\theta_{1}}) = \mu(\sigma_{\theta_{2}}) = 1/2$, where
\begin{align*}
& \sigma_{\theta_{1}}(a_{i}|\omega_{j}) = \sigma_{\theta_{2}}(a_{i}|\omega_{j}) = \sigma^{BP}(a_{i}|\omega_{j}), \quad \text{for all } i = 2, \ldots, n, \text{ and } j = 0, \ldots, n,\\
& \sigma_{\theta_{1}}(m_{0}|\omega_{0}) = \sigma^{BP}(a_{0}|\omega_{0}) + \sigma^{BP}(a_{1}|\omega_{0}), \quad \sigma_{\theta_{1}}(m_{1}|\omega_{1}) = 1,\\
& \sigma_{\theta_{2}}(m_{1}|\omega_{0}) = \sigma^{BP}(a_{0}|\omega_{0}) + \sigma^{BP}(a_{1}|\omega_{0}), \quad \sigma_{\theta_{2}}(m_{0}|\omega_{1}) = 1,
\end{align*}
and all other probabilities are zero. 

The receiver's posterior over $\omega$ conditional on $a_{i}$ for $i > 1$ is the same as $\sigma^{BP}$. Thus, the receiver's optimal action remains the corresponding $a_{i}$. By symmetry, the receiver's optimal action is the same after observing $m_{0}$ or $m_{1}$. Since an interim best response can put positive probability on at most one $a_{i}$ with $i > 0$ (see the proof of Lemma \ref{lem:interim_best_response}) and $a_{0}$ dominates other $a_{i}$ with $i > 1$ in states $\omega_{0}$ and $\omega_{1}$, the receiver chooses $\alpha a_{1} + (1-\alpha) a_{0}$ for some $\alpha \in [0,1]$. Thus, the sender benefits if $\alpha$ is strictly greater than 
\begin{align*}
\alpha^{BP} := \frac{p(\omega_{1}) + p(\omega_{0})\sigma^{BP}(a_{1}|\omega_{0})}{p(\omega_{1}) + p(\omega_{0})(\sigma^{BP}(a_{1}|\omega_{0})+ \sigma^{BP}(a_{0}|\omega_{0}))}.
\end{align*}

The receiver's posterior beliefs over $\omega$ conditional on message $m_{0}$ and $\sigma_{\theta}$ are given by $q_{\theta_{1}}(\omega_{0}|m_{0}) = 1$ and $q_{\theta_{2}}(\omega_{1}|m_{0}) = 1$, with posterior beliefs over $\Sigma$ given by
\begin{align*}
&\eta(\sigma_{\theta_{1}}|m_{0}) = \frac{p(\omega_{0})(\sigma^{BP}(a_{0}|\omega_{0}) + \sigma^{BP}(a_{1}|\omega_{0}))}{p(\omega_{0})(\sigma^{BP}(a_{0}|\omega_{0}) + \sigma^{BP}(a_{1}|\omega_{0})) + p(\omega_{1})} \text{ and } \eta(\sigma_{\theta_{2}}|m_{0}) = 1 - \eta(\sigma_{\theta_{1}}|m_{0}).
\end{align*}
The receiver's expected payoff from taking $\alpha a_{1} + (1-\alpha) a_{0}$ is given by
\begin{align*}
V(\alpha) := &\frac{p(\omega_{0})(\sigma^{BP}(a_{0}|\omega_{0}) + \sigma^{BP}(a_{1}|\omega_{0}))}{p(\omega_{0})(\sigma^{BP}(a_{0}|\omega_{0}) + \sigma^{BP}(a_{1}|\omega_{0})) + p(\omega_{1})}\phi_{r}(\alpha d + (1-\alpha)(d + k_{1}T_{1}))\\
&+ \frac{p(\omega_{1})}{p(\omega_{0})(\sigma^{BP}(a_{0}|\omega_{0}) + \sigma^{BP}(a_{1}|\omega_{0})) + p(\omega_{1})}\phi_{r}(\alpha d + (1-\alpha)(d - k_{1}(1-T_{1}))).
\end{align*}
Since $V(\alpha)$ is concave in $\alpha$, it holds that the optimal choice of $\alpha$ is strictly greater than $\alpha^{BP}$ if the first-order condition at $\alpha^{BP}$ is strictly positive, i.e., $V'(\alpha^{BP}) > 0$, which says
\begin{align*}
\frac{\phi'_{r}(\alpha^{BP} d + (1-\alpha^{BP})(d - k_{1}(1-T_{1})))}{\phi'_{r}(\alpha^{BP} d + (1-\alpha^{BP})(d + k_{1}T_{1}))} & > \frac{p(\omega_{0})(\sigma^{BP}(a_{0}|\omega_{0}) + \sigma^{BP}(a_{1}|\omega_{0}))}{p(\omega_{1})}\frac{T_{1}}{1-T_{1}}\\
& = 1 + \frac{T_{1}}{p(\omega_{1})(1-T_{1})} \left(1 - \sum_{i=1}^{n} \frac{p(\omega_{i})}{T_{i}} \right) > 1.
\end{align*}
This is satisfied when $\phi_{r}$ is sufficiently concave in the relevant range.

\textbf{If $\bar{i} = 0$.} We first establish the following lemma.

\begin{lemma}\label{lem:interim_best_response}
Suppose $u_{r}(a_{0}, \omega)$ is constant across $\omega \in \Omega$. If an inconsistent receiver's interim best response after message $m$ puts positive probability on action $a_{i}$ for some $i > 0$, then the average posterior probability of $\omega_{i}$ conditional on $m$ must be at least $T_{i}$, i.e., $\overline{q}(\omega_{i}|m) := \sum_{\theta}\eta(\sigma_{\theta}|m)q_{\theta}(\omega_{i}|m) \geq T_{i}$.
\end{lemma}
\begin{proof}[Proof of Lemma \ref{lem:interim_best_response}]
Any interim best response $\tau(m)$ is optimal under an
effective posterior belief $\tilde{q}\in\Delta(\Omega)$.
Assigning positive probability to $a_i$, $i>0$, therefore
requires $\tilde{q}(\omega_i)\geq T_i>1/2$, which can hold
for at most one $i$.

Suppose $\tau(m)=\alpha a_i+(1-\alpha)a_0$ is an interim
best response for some $i>0$ and $\alpha\in(0,1]$, and let
$u_r(a_0,\omega)=d$. Then
\begin{align*}
\phi_r(d)\leq
\sum_\theta\eta(\sigma_\theta|m)
\phi_r\bigl(d+\alpha k_i(q_\theta(\omega_i|m)-T_i)\bigr) \leq
\phi_r\bigl(d+\alpha k_i(\overline{q}(\omega_i|m)-T_i)\bigr),
\end{align*}
where the first inequality follows from optimality relative
to $a_0$, and the second from Jensen's inequality and
concavity of $\phi_r$. Since $\phi_r$ is strictly increasing
and $\alpha k_i>0$, it follows that
$\overline{q}(\omega_i|m)\geq T_i$.
\end{proof}

Fix any ambiguous experiment $\mu$ and the receiver's interim best response $\tau$. For each action $a_{i}$ with $i > 0$, let $M_{i} = \{m \in M: \tau(a_{i}|m) > 0\}$. By Lemma \ref{lem:interim_best_response}, for each $m \in M_{i}$, it holds that $\overline{q}(\omega_{i}|m) \geq T_{i}$. Let $\pi(m) = \sum_{\theta}\mu_{\theta}\sum_{\omega}p(\omega)\sigma_{\theta}(m|\omega)$ denote the average total probability of message $m$. Then, by Bayes plausibility,
\begin{align*}
p(\omega_{i}) = \sum_{m \in M} \pi(m) \overline{q}(\omega_{i}|m) \geq \sum_{m \in M_{i}} \pi(m) \overline{q}(\omega_{i}|m) \geq T_{i}\sum_{m \in M_{i}}\pi(m).
\end{align*}
Thus, $\sum_{m \in M_{i}}\pi(m) \leq p(\omega_{i})/T_{i}$. Notice the sender's payoff is bounded above by
\begin{align*}
\phi_{s}^{-1}\left(\sum_{\theta}\mu_{\theta}\phi_{s}\left(\sum_{i=1}^{n}v_{i}\sum_{m \in M_{i}}\sum_{\omega}p(\omega)\sigma_{\theta}(m|\omega)\right)\right) \leq \sum_{i=1}^{n}v_{i}\sum_{m \in M_{i}}\pi(m) \leq \sum_{i=1}^{n}v_{i}\frac{p(\omega_{i})}{T_{i}} = u_{s}^{BP}.
\end{align*}
Therefore, no ambiguous experiment can give the sender a strictly higher payoff than $u_{s}^{BP}$, and the sender does not benefit from ambiguous communication.
\end{proof}

\begingroup
\color{black}
Before proving Theorems \ref{thm:binary_improvement_under_prior_ambiguity} and \ref{thm:sender_benefit_prior_ambiguity}, we first characterize the receiver's obedience condition under pre-existing ambiguity about the payoff-relevant states. For any joint distribution $\pi \in \Delta(\Omega \times A)$, $\tau^{*}$ is the receiver's best response if
\begin{equation*}
\tau^{*} \in \mathop{\arg\max}_{\tau} \sum_{\omega,a,a'}\pi(\omega,a)\tau(a'|a)u_r(a',\omega).
\end{equation*}
Let $\Pi^{*} \subseteq \Delta(\Omega \times A)$ denote the set of all such distributions. Note that $\Pi^{*}$ is convex and closed. Thus, an experiment $\sigma$ is obedient under prior $p$ if the induced joint distribution $\pi^{(p,\sigma)}=p\times\sigma\in\Pi^{*}$. An argument similar to that in the proof of Lemma \ref{lem_effmeasure} can be used to show that when there is pre-existing ambiguity $\eta$ and the sender chooses an ambiguous experiment $\mu$, $\tau^{*}$ is a best response of the receiver if and only if
\begin{equation*}
\pi^{(\eta,\mu)}
:=\frac{\displaystyle\int\sum_{\theta}\mu_{\theta}\phi'_r(u_r(p,\sigma_{\theta}))(p\times\sigma_{\theta})\,\eta(dp)}
{\displaystyle\int\sum_{\theta}\mu_{\theta}\phi'_r(u_r(p,\sigma_{\theta}))\,\eta(dp)}
\in\Pi^{*},
\end{equation*}
where we explicitly include $p$ as an argument of the receiver's expected utility $u_r(p,\sigma_{\theta})$. In this case, we say the ambiguous experiment $\mu$ is obedient under $\eta$.

\begin{proof}[Proof of Theorem \ref{thm:binary_improvement_under_prior_ambiguity}]
Let $\mu=\mu^{((\overline{\sigma},\underline{\sigma}),\kappa)}$, with $\kappa$ given by \eqref{equ_thm_binary_improvement_under_prior_ambiguity}. Because $\phi_r$ is strictly concave, and $(\overline{\sigma},\underline{\sigma},\lambda)$ is a Pareto-ranked splitting of $\sigma^{*}$ for all $p\in\mathrm{supp}(\eta)$, it holds that $\kappa>\lambda$. Thus, if $\mu$ is obedient under $\eta$, then an ambiguity neutral sender does strictly better using it than using $\sigma^{*}$.

We next establish its obedience. For each $p\in\mathrm{supp}(\eta)$, let $em^{(p,\mu)}_{\overline{\theta}}$ denote the effective weight on $\overline{\sigma}$ under prior $p$. With $\kappa$ given by \eqref{equ_thm_binary_improvement_under_prior_ambiguity}, it follows that for a minimizing $p^{*}$, $em^{(p^{*},\mu)}_{\overline{\theta}}=\lambda$, and for every $p\in\mathrm{supp}(\eta)$, $em^{(p,\mu)}_{\overline{\theta}}\leq\lambda$. It further implies that, for each $p\in\mathrm{supp}(\eta)$, the experiment
\begin{equation*}
em^{(p,\mu)}_{\overline{\theta}}\overline{\sigma}
+(1-em^{(p,\mu)}_{\overline{\theta}})\underline{\sigma}
\end{equation*}
is a convex combination of $\sigma^{*}$ and $\underline{\sigma}$, and thus, by the obedience of $\sigma^{*}$ and $\underline{\sigma}$ under $p$, is also obedient under $p$. The obedience of $\mu$ under $\eta$ is implied by the following lemma:

\begin{lemma}\label{lem:obedience_sufficient_condition}
If an ambiguous experiment $\mu$ is obedient under $p$ for every $p\in\mathrm{supp}(\eta)$, then it is obedient under $\eta$.
\end{lemma}

\begin{proof}[Proof of Lemma \ref{lem:obedience_sufficient_condition}]
Given $p\in\Delta(\Omega)$ and $\mu$, let
\begin{equation*}
\sigma_p^{\mu}
:=\sum_{\theta}\frac{\mu_{\theta}\phi'_r(u_r(p,\sigma_{\theta}))}
{\sum_{\tilde{\theta}}\mu_{\tilde{\theta}}\phi'_r(u_r(p,\sigma_{\tilde{\theta}}))}\sigma_{\theta}.
\end{equation*}
If $\mu$ is obedient under prior $p$, then $\pi^{(p,\mu)}:=p\times\sigma_p^{\mu}\in\Pi^{*}$. Consider an ambiguous experiment $\mu$ under $\eta$. The induced joint distribution satisfies
\begin{align*}
\pi^{(\eta,\mu)}
&=\frac{\displaystyle\int\sum_{\theta}\mu_{\theta}\phi'_r(u_r(p,\sigma_{\theta}))(p\times\sigma_{\theta})\,\eta(dp)}
{\displaystyle\int\sum_{\theta}\mu_{\theta}\phi'_r(u_r(p,\sigma_{\theta}))\,\eta(dp)}\\
&=\int\frac{\sum_{\theta}\mu_{\theta}\phi'_r(u_r(p,\sigma_{\theta}))}
{\displaystyle\int\sum_{\tilde{\theta}}\mu_{\tilde{\theta}}\phi'_r(u_r(q,\sigma_{\tilde{\theta}}))\,\eta(dq)}\,
\pi^{(p,\mu)}\,\eta(dp).
\end{align*}
In other words, $\pi^{(\eta,\mu)}$ is a probability-weighted average of $\pi^{(p,\mu)}$ for $p\in\mathrm{supp}(\eta)$. Thus, if $\mu$ is obedient under $p$ for every $p\in\mathrm{supp}(\eta)$, by the convexity and closedness of $\Pi^{*}$, it is obedient under $\eta$ as well.
\end{proof}
\end{proof}

\begin{proof}[Proof of Theorem \ref{thm:sender_benefit_prior_ambiguity}]
By the supposition, there exists an obedient ambiguous experiment $\mu$ such that (we explicitly write the dependence on the prior $p$):
\begin{equation*}
U_s(p,\mu)>u_s^{BP}(p).
\end{equation*}

Let $\sigma^{*}$ denote the effective experiment induced by $\mu$ under $p$ and let $\pi^{*}=p\times\sigma^{*}$. By continuity and the fact that $\mathrm{int}(\Pi^{*})\neq\emptyset$, we can perturb $\mu$ slightly to $\tilde{\mu}$ such that
\begin{equation*}
U_s(p,\tilde{\mu})>u_s^{BP}(p),
\end{equation*}
and the joint distribution $\pi^{(p,\tilde{\mu})}\in\mathrm{int}(\Pi^{*})$.

Notice that for any $\eta$ that reduces to $p$, because the sender is ambiguity neutral, the sender's payoff remains the same under $\eta$ as under $p$ provided the receiver is obedient. Because $\pi^{(p,\tilde{\mu})}\in\mathrm{int}(\Pi^{*})$, there exists $\delta>0$ such that for all $q\in\Delta(\Omega)$ with $\Vert q-p\Vert<\delta$, the distribution $\pi^{(q,\tilde{\mu})}$ remains in the interior of $\Pi^{*}$. Thus, by Lemma \ref{lem:obedience_sufficient_condition}, the ambiguous experiment $\tilde{\mu}$ will continue to benefit the sender under any $\eta$ that reduces to $p$ and has support within this $\delta$-neighborhood of $p$.
\end{proof}
\par
\endgroup

\begingroup
\color{black}
\section[A Polar Case: A Maxmin Receiver and Ambiguity Neutral Sender]{A Polar Case: A Maxmin Receiver and\\Ambiguity Neutral Sender}\label{sec_maxmin}

This section analyzes the case of an ambiguity-neutral sender and an infinitely ambiguity averse receiver, represented by the maxmin preferences $U_{r}^{MEU}$. As the comparative statics statements in Theorem \ref{general_improvement} suggest, this is a particularly favorable case for the sender to benefit from ambiguous communication.

More precisely, we show that in this case (a) binary ambiguous experiments are sufficient to exhaust all gains from persuasion, and (b) the sender can attain a payoff arbitrarily close to their best feasible payoff subject to the receiver getting at least the payoff they would obtain if no information were disclosed. A conclusion we draw is that assuming an infinitely ambiguity averse receiver is very powerful and, in our view, unrealistically so, further motivating the analysis in the rest of the paper which allows for more moderate levels of aversion.

Before turning to the analysis, we remark that the opposite cases, of either an infinitely ambiguity averse sender with payoffs $U_{s}^{MEU}$ or an ambiguity neutral receiver, preclude any benefit from ambiguous persuasion. The latter case follows from Theorem \ref{thm_necessary}, while \citet{cheng2020ambiguous} shows that in the former case the sender never benefits from ambiguous communication.

The following lemma relates obedience for an ambiguous experiment to obedience for an experiment. It is thus the analogue of Lemma \ref{lem_effmeasure} for a receiver with preferences $U_{r}^{MEU}$:

\begin{lemma}\label{lem_exante_IC}
An ambiguous experiment $\mu$ is obedient if, and only if, there exists $\lambda \in \Delta(\mathrm{supp}(\mu))$ such that
\begin{equation*}
\mathrm{supp}(\lambda) \subseteq \mathop{\arg\min}_{\sigma \in \mathrm{supp}(\mu)} u_r(\sigma)
\quad\text{and}\quad
\sigma^* := \sum_{\theta}\lambda_{\theta}\sigma_{\theta}
\text{ is obedient}.
\end{equation*}
\end{lemma}

\begin{proof}[Proof of Lemma \ref{lem_exante_IC}]
Fix $\mu$. If $\mu$ is obedient, then $\tau^*$ maximizes $\min_{\sigma_{\theta} \in \mathrm{supp}(\mu)}u_r(\sigma_{\theta},\tau)$. By the minmax theorem, there exists $\lambda \in \Delta(\mathrm{supp}(\mu))$ such that, for every $\tau$,
\begin{equation*}
\sum_{\theta}\lambda_{\theta}u_r(\sigma_{\theta},\tau)
\leq \sum_{\theta}\lambda_{\theta}u_r(\sigma_{\theta})
= \min_{\sigma \in \mathrm{supp}(\mu)}u_r(\sigma).
\end{equation*}
The equality implies that $\lambda$ assigns positive weight only to experiments minimizing the receiver's payoff under obedience. The inequality implies that $\sigma^* = \sum_{\theta}\lambda_{\theta}\sigma_{\theta}$ is obedient.

Conversely, suppose such a $\lambda$ exists. For every $\tau$,
\begin{align*}
\min_{\sigma \in \mathrm{supp}(\mu)}u_r(\sigma,\tau)
&\leq \sum_{\theta}\lambda_{\theta}u_r(\sigma_{\theta},\tau) \leq \sum_{\theta}\lambda_{\theta}u_r(\sigma_{\theta})
= \min_{\sigma \in \mathrm{supp}(\mu)}u_r(\sigma),
\end{align*}
where the second inequality follows from obedience of $\sigma^*$ and the equality from the support restriction on $\lambda$. Thus $\mu$ is obedient.
\end{proof}

Observe that when the argmin in Lemma \ref{lem_exante_IC} is a singleton, $\sigma^*$ equals the receiver's payoff-minimizing experiment from $\mathrm{supp}(\mu)$. More generally, it is a convex combination of the possibly multiple minimizing experiments, with weights supporting obedience. 

Since the receiver can always ignore any recommendations made, they can guarantee themselves the payoff
\begin{equation*}
\underline{u}^{*}_{r} := \max_{a \in A}\sum_{\omega}p(\omega)u_r(a,\omega),
\end{equation*}
which is the payoff they would obtain if no information were disclosed. The consequence of the great flexibility available to the sender given Lemma \ref{lem_exante_IC} is the next theorem, which states that the sender's optimal payoff approaches their highest feasible payoff subject to the receiver getting at least $\underline{u}^{*}_{r}$.\footnote{\textcolor{black}{This ``efficiency subject to an outside option'' approach to identifying the sender's optimal payoff is reminiscent of one developed in \cite{smolin2021} that can be applied to any Bayesian persuasion problem where the receiver has only two actions. In our case, it is the extremity of the receiver's MEU preference that creates the almost complete separation between the experiment providing the sender's payoff and the one providing the receiver's payoff.}} The corresponding communication strategy uses a binary ambiguous experiment with weight $\kappa$ on the better experiment approaching 1, and the worse experiment an obedient one holding the receiver to $\underline{u}^{*}_{r}$.

\begin{samepage}
\begin{theorem}\label{thm_maxmin_receiver_optimal}
Suppose there exists $\hat{\sigma}$ such that $u_r(\hat{\sigma}) > \underline{u}^{*}_{r}$. The value of the following program is the supremum of the payoff that an ambiguity neutral sender can obtain when the receiver has maxmin preferences $U_{r}^{MEU}$:
\begin{align*}
&\max_{\sigma}u_s(\sigma),\\
&\text{s.t. }u_r(\sigma) \geq \underline{u}^{*}_{r}.
\end{align*}
\end{theorem}
\end{samepage}

\begin{proof}[Proof of Theorem \ref{thm_maxmin_receiver_optimal}]
Let $\overline{\sigma}$ attain the value of the program in the theorem for the sender and $\underline{\sigma}$ be an obedient uninformative experiment, so that $u_r(\underline{\sigma}) = \underline{u}^{*}_{r}$. Assume $u_r(\overline{\sigma}) > u_r(\underline{\sigma})$. If the sender chooses $\mu^{((\overline{\sigma},\underline{\sigma}),\kappa)}$, the receiver is obedient by Lemma \ref{lem_exante_IC}, with $\lambda$ concentrated on $\underline{\sigma}$. As $\kappa$ approaches 1, the sender's payoff approaches $u_s(\overline{\sigma})$, while the receiver's payoff remains $\underline{u}^{*}_{r}$.

If $u_r(\overline{\sigma}) = u_r(\underline{\sigma})$, mix $\overline{\sigma}$ with $\hat{\sigma}$. For any $\varepsilon \in (0,1)$,
\begin{equation*}
u_r((1-\varepsilon)\overline{\sigma}+\varepsilon\hat{\sigma}) > u_r(\underline{\sigma}).
\end{equation*}
Apply the preceding construction with $(1-\varepsilon)\overline{\sigma}+\varepsilon\hat{\sigma}$ in place of $\overline{\sigma}$. As $\varepsilon$ approaches 0 and $\kappa$ approaches 1, the payoff for the sender approaches $u_s(\overline{\sigma})$.

The sender cannot do better than this. For any obedient $\mu$, the receiver's payoff is at least $\underline{u}^{*}_{r}$, so $u_r(\sigma_{\theta}) \geq \underline{u}^{*}_{r}$ for every $\sigma_{\theta} \in \mathrm{supp}(\mu)$. Thus $\sum_{\theta}\mu_{\theta}\sigma_{\theta}$ is feasible in the program and, since the sender is ambiguity neutral, delivers the sender the same payoff as $\mu$.
\end{proof}

\begin{remark}[MEU receiver's payoff]\label{rem_MEU_receiver_payoff}
\emph{Theorem \ref{thm_maxmin_receiver_optimal} does not imply that an MEU receiver is held to (or even close to) $\underline{u}^{*}_{r}$ by all nearly optimal strategies. Nevertheless, the receiver's payoff along particular nearly optimal communication strategies in the smooth ambiguity model can differ from their payoff under the MEU construction in the proof. In our introductory example, keep the experiments $\sigma_{\overline{\theta}}$ and $\sigma_{\underline{\theta}}$ fixed, and, for each $\phi_r$, choose the weight $\kappa$ so that the receiver's effective weight on $\sigma_{\overline{\theta}}$ remains $3/4$. Along a sequence of twice continuously differentiable $\phi_r$ for which $-\phi_r''/\phi_r'$ tends uniformly to infinity on $[1/2,3/2]$, both the sender's payoff and the receiver's certainty equivalent approach $3/2$. In contrast, for an MEU receiver, the construction in the proof of Theorem \ref{thm_maxmin_receiver_optimal}, using the same pair of experiments, holds the receiver to $\underline{u}^{*}_{r}=1/2$ while the sender's payoff approaches $3/2$ as $\kappa$ approaches 1. Alternatively, replacing the worse experiment by $\sigma^{BP}$ holds the MEU receiver to the Bayesian persuasion payoff of $5/4$, while the sender's payoff again approaches $3/2$. Thus, in this example, these constructions approach the same sender supremum but give different receiver payoffs. \qedsymbol}
\end{remark}
\medskip

There is a sense in which Theorem \ref{thm_maxmin_receiver_optimal} could be argued to overstate what the sender can achieve. For MEU, the experiment $\sigma^*$ supporting obedience could have a smaller support than the experiments in $\mathrm{supp}(\mu)$. In particular, the joint distribution $p(\omega)\sigma^*(a|\omega)$ may assign zero probability to a state--recommendation pair that has positive probability under $p(\omega)\sum_{\theta}\mu_{\theta}\sigma_{\theta}(a|\omega)$. Obedience under $\sigma^*$ then places no weight on the receiver's payoff at such a pair. For a smooth ambiguity receiver, in contrast, the effective measure assigns positive weight to every experiment in $\mathrm{supp}(\mu)$, so these two joint distributions have the same support.

This issue can be addressed by strengthening obedience to require that there be a worst-case probability distribution over experiments supporting the receiver's obedience whose induced joint distribution over states and recommendations has the same support as that induced by $\mu$. We call this \emph{strong obedience}. It is an additional restriction on MEU obedience, motivated by the support property of smooth ambiguity effective measures.

The following lemma gives the corresponding version of Lemma \ref{lem_exante_IC}:

\begin{lemma}\label{lem_exante_IC_refined}
An ambiguous experiment $\mu$ is strongly obedient if, and only if, there exists $\lambda \in \Delta(\mathrm{supp}(\mu))$ such that
\begin{equation*}
\mathrm{supp}(\lambda) \subseteq \mathop{\arg\min}_{\sigma \in \mathrm{supp}(\mu)}u_r(\sigma),
\end{equation*}
the experiment $\sigma^* := \sum_{\theta}\lambda_{\theta}\sigma_{\theta}$ is obedient, and
\begin{equation*}
\mathrm{supp}\,\sigma_{\theta}(\cdot|\omega)
\subseteq \mathrm{supp}\,\sigma^*(\cdot|\omega)
\quad\text{for all }\sigma_{\theta} \in \mathrm{supp}(\mu)\text{ and }\omega \in \Omega.
\end{equation*}
\end{lemma}

\begin{proof}[Proof of Lemma \ref{lem_exante_IC_refined}]
By Lemma \ref{lem_exante_IC}, a worst-case probability distribution supporting obedience is a $\lambda$ satisfying the first two conditions. Since $p$ has full support, equality of the supports of the joint distributions induced by $\sigma^*$ and by $\mu$ is equivalent to
\begin{equation*}
\mathrm{supp}\,\sigma^*(\cdot|\omega)
= \bigcup_{\sigma_{\theta} \in \mathrm{supp}(\mu)}\mathrm{supp}\,\sigma_{\theta}(\cdot|\omega)
\quad\text{for all }\omega \in \Omega.
\end{equation*}
The left-hand side is always contained in the right-hand side because $\sigma^*$ is a convex combination of experiments in $\mathrm{supp}(\mu)$. The reverse inclusion is precisely the last condition in the lemma.
\end{proof}

The following is the version of Theorem \ref{thm_maxmin_receiver_optimal} using the stronger obedience notion. As before, binary ambiguous experiments suffice to approach the sender's supremum. The additional restriction is that the experiment determining the sender's limiting payoff must be covered, state by state, by an obedient experiment. Recall that $\Sigma^*$ denotes the set of obedient experiments.

\begin{samepage}
\begin{theorem}\label{thm_maxmin_receiver_optimal_refined}
Suppose there exists $\hat{\sigma}$ such that $u_r(\hat{\sigma}) > \underline{u}^{*}_{r}$. The value of the following program is the supremum of the payoff that an ambiguity neutral sender can obtain using strongly obedient ambiguous experiments when the receiver has maxmin preferences $U_{r}^{MEU}$:
\begin{align*}
&\sup_{\sigma,\sigma^*}u_s(\sigma),\\
&\text{s.t. }u_r(\sigma) > \underline{u}^{*}_{r},\quad \sigma^* \in \Sigma^*,\\
&\phantom{\text{s.t. }}\mathrm{supp}\,\sigma(\cdot|\omega)
\subseteq \mathrm{supp}\,\sigma^*(\cdot|\omega)
\quad\text{for all }\omega \in \Omega.
\end{align*}
\end{theorem}
\end{samepage}

\begin{proof}[Proof of Theorem \ref{thm_maxmin_receiver_optimal_refined}]
By replacing $\hat{\sigma}$, if necessary, with an experiment recommending a receiver-optimal action in every state, we can take $\hat{\sigma}$ to be obedient while retaining $u_r(\hat{\sigma}) > \underline{u}^{*}_{r}$. Let $(\overline{\sigma},\sigma^*)$ be any feasible pair in the program. Let $\underline{\sigma}$ be an obedient uninformative experiment, so that $u_r(\underline{\sigma}) = \underline{u}^{*}_{r}$, and let
\begin{equation*}
\tilde{\sigma} = \tfrac12\sigma^*+\tfrac12\hat{\sigma}.
\end{equation*}
Obedience is preserved under convex combinations of experiments. Thus $\tilde{\sigma}$ is obedient, $u_r(\tilde{\sigma}) > \underline{u}^{*}_{r}$, and $\mathrm{supp}\,\overline{\sigma}(\cdot|\omega) \subseteq \mathrm{supp}\,\tilde{\sigma}(\cdot|\omega)$ for every $\omega$.

Define a sequence of experiments
\begin{equation*}
\underline{\sigma}_n = (1-\varepsilon_n)\underline{\sigma}+\varepsilon_n\tilde{\sigma},
\end{equation*}
where $\varepsilon_n \in (0,1)$ and $\varepsilon_n \rightarrow 0$. For all sufficiently large $n$,
\begin{equation*}
u_r(\overline{\sigma}) > u_r(\underline{\sigma}_n) > \underline{u}^{*}_{r}.
\end{equation*}
If the sender chooses $\mu^{((\overline{\sigma},\underline{\sigma}_n),\kappa)}$, the unique payoff-minimizing experiment is $\underline{\sigma}_n$. This experiment is obedient and satisfies
\begin{equation*}
\mathrm{supp}\,\overline{\sigma}(\cdot|\omega)
\subseteq \mathrm{supp}\,\underline{\sigma}_n(\cdot|\omega)
\quad\text{for every }\omega.
\end{equation*}
Lemma \ref{lem_exante_IC_refined}, with $\lambda$ concentrated on $\underline{\sigma}_n$, therefore implies strong obedience. As $\kappa$ approaches 1, the sender's payoff approaches $u_s(\overline{\sigma})$. As $n$ goes to infinity, the receiver's payoff approaches $\underline{u}^{*}_{r}$. Since $(\overline{\sigma},\sigma^*)$ was any feasible pair, this proves that the value of the program can be approached as claimed.

Conversely, take any strongly obedient $\mu$, and let $\sigma^*$ be an obedient experiment satisfying Lemma \ref{lem_exante_IC_refined}. Let $\overline{\sigma} = \sum_{\theta}\mu_{\theta}\sigma_{\theta}$. Since the receiver can always ignore recommendations, $u_r(\overline{\sigma}) \geq \underline{u}^{*}_{r}$. Moreover,
\begin{equation*}
\mathrm{supp}\,\overline{\sigma}(\cdot|\omega)
\subseteq \mathrm{supp}\,\sigma^*(\cdot|\omega)
\quad\text{for every }\omega,
\end{equation*}
and $u_s(\overline{\sigma}) = U_s(\mu)$ because the sender is ambiguity neutral. If $u_r(\overline{\sigma}) > \underline{u}^{*}_{r}$, this pair is feasible in the program. If equality holds, mix each of $\overline{\sigma}$ and $\sigma^*$ with weight $\varepsilon \in (0,1)$ on the obedient $\hat{\sigma}$. The first mixture gives the receiver strictly more than $\underline{u}^{*}_{r}$, while the second is obedient and covers the first state by state. These pairs are feasible in the program, and their sender payoffs converge to $u_s(\overline{\sigma})$ as $\varepsilon$ approaches 0. Thus the sender cannot obtain a payoff above the supremum in the theorem.
\end{proof}
\par
\endgroup

\begingroup
\color{black}
\section{The Insufficiency of Binary Ambiguous Experiments}\label{sec_insufficiency_binary_experiment}
    
    \begin{proposition}\label{prop_insufficiency_binary_experiment}
        It is not always sufficient to consider only binary ambiguous experiments in searching for either a strict benefit from ambiguity or optimal ambiguous persuasion.
    \end{proposition}
    
    We provide a detailed sketch of the proof. For the full argument, including the computational verification of the binary upper bound in the second example, see \citet{cheng2024additional1}. 
    
    \begin{proof}[Sketch of Proof of Proposition \ref{prop_insufficiency_binary_experiment}]
    
    The proof is by construction. We first show an example in which the only optimal ambiguous experiments are more than binary. A modification of this example is then used to provide an example in which the sender may strictly benefit from ambiguous communication even when no binary ambiguous experiment benefits the sender.
    
    \bigskip 
    
    \textbf{Example in which all optimal ambiguous experiments are more than binary.}
    
    Suppose $\phi_{s}(x) = x$ and $\phi_{r}(x) = \ln(x + 5)$. Let $\Omega = \{\omega_{1}, \omega_{2}\}$, with equal prior probabilities $p = (1/2, 1/2)$. There are five actions $\{a_{1}, a_{2}, b_{1}, b_{2}, b_{3}\}$ and the payoff matrix is 
    \begin{center}
        \begin{tabular}{c|cc}
            $(u_{s}, u_{r})$& $\omega_{1}$ & $\omega_{2}$ \\ 
            \hline 
            $a_{1}$ & $3, 3$ & $0, 0$ \\ 
            $a_{2}$ & $-1, -1$ & $3, 3$ \\ 
            $b_{1}$ & $0, 4$ & $-1, -2$ \\ 
            $b_{2}$ & $0, 2$ & $1, 2$ \\ 
            $b_{3}$ & $-2, -4$ & $1, 4$ \\ 
        \end{tabular} 
    \end{center}
    
    The optimal Bayesian persuasion experiment is
    \begin{align*}
        &\sigma_{a}(a_{1}|\omega_{1}) = 4/5, \quad \sigma_{a}(a_{2}|\omega_{1}) = 1/5;\\
        &\sigma_{a}(a_{1}|\omega_{2}) = 2/5, \quad \sigma_{a}(a_{2}|\omega_{2}) = 3/5.
    \end{align*}
    Notice that 
    \begin{align*}
        u_{s}(\sigma_{a}) = u_{r}(\sigma_{a}) = 2. 
    \end{align*}
    
    Let $\sigma_{11}, \sigma_{12}, \sigma_{21}$ and $\sigma_{22}$ denote the extreme experiments where $\sigma_{ij}$ recommends $a_{i}$ and $a_{j}$ deterministically in states $\omega_{1}$ and $\omega_{2}$, respectively. Notice that these extreme experiments are all Pareto-ranked: 
    \begin{align*}
        &u_{s}(\sigma_{11}) = u_{r}(\sigma_{11}) = 3/2;\\
        &u_{s}(\sigma_{12}) = u_{r}(\sigma_{12}) = 3;\\
        &u_{s}(\sigma_{21}) = u_{r}(\sigma_{21}) = -1/2;\\
        &u_{s}(\sigma_{22}) = u_{r}(\sigma_{22}) = 1.
    \end{align*}
    
    Consider the following splitting of $\sigma_{a}$, 
    \begin{align*}
        \sigma_{a} = \frac{1}{5} \sigma_{11} +\frac{3}{5} \sigma_{12} +  \frac{1}{5} \sigma_{21}.
    \end{align*}    
    By \eqref{equ_reverse_effmeasure}, the corresponding ambiguous experiment is
    \begin{equation*}
        \hat{\mu}=\frac{13}{70}\delta_{\sigma_{11}}+\frac{48}{70}\delta_{\sigma_{12}}+\frac{9}{70}\delta_{\sigma_{21}}.
    \end{equation*}
    Its effective measure has weights $(1/5,3/5,1/5)$, so that $\sum_{\theta}em^{\hat{\mu}}_{\theta}\sigma_{\theta}=\sigma_a$. Thus, $\hat{\mu}$ is obedient by Lemma \ref{lem_effmeasure} and yields the sender a payoff of $U_s(\hat{\mu})=159/70\simeq2.27143$, strictly higher than the $2$ under Bayesian persuasion. Therefore, any optimal ambiguous experiment must involve ambiguity and thus use at least two distinct components. 
    
    As $\phi_{s}(x) = x$ and $\phi_{r}(x) = \ln(x + 5)$, by Proposition \ref{cor_linear_phi_inverse_concavification}, in any optimal ambiguous experiment, there cannot exist any further Pareto-ranked splitting of any experiment in the support. 
    
    Observe that $\sigma_{a}$ is the only obedient experiment that never recommends any of the $b$ actions. Furthermore, $\sigma_{a}$ cannot be split into a convex combination of two extreme experiments. Thus, any binary splitting of $\sigma_{a}$ must involve at least one non-extreme experiment. However, since all these extreme experiments are Pareto-ranked, any such non-extreme experiment admits a Pareto-ranked splitting. Therefore, any binary ambiguous experiment constructed from splittings of $\sigma_{a}$ cannot be optimal. 
    
    The proof goes on to show that an optimal ambiguous experiment in this example also cannot be a binary ambiguous experiment that is constructed from a splitting of any other obedient experiment (in particular, any recommending a $b$ action with a positive probability). \medskip 
    
    \textbf{Example in which ambiguous communication benefits the sender, but does not do so when restricted to binary ambiguous experiments}
    
    Suppose $\phi_{s}(x) = x$ and $\phi_{r}(x) = \ln(x + 5)$. Let $\Omega = \{\omega_{1}, \omega_{2}\}$ and the prior $p$ be uniform. There are seven actions $\{a_{1}, a_{2}, b_{1}, b_{2}^{+}, b_{2}^{-}, b_{3}, c\}$. Let the payoff matrix be, for some $x > 2$, 
    \begin{center}
        \begin{tabular}{c|cc}
            $(u_{s}, u_{r})$& $\omega_{1}$ & $\omega_{2}$ \\ 
            \hline 
            $a_{1}$ & $3, 3$ & $0, 0$ \\ 
            $a_{2}$ & $-1, -1$ & $3, 3$ \\ 
            $b_{1}$ & $0, 4$ & $0, -2$ \\ 
            $b_{2}^{-}$ & $0, 5/2$ & $0, 1$ \\
            $b_{2}^{+}$ & $0, 5/4$ & $0, 9/4$ \\
            $b_{3}$ & $0, -4$ & $0, 4$ \\ 
            $c$ & $x, 7/4$ & $x, 7/4$ \\ 
        \end{tabular} 
    \end{center}
    
    Relative to the previous example, we add $c$, replace $b_2$ by $b_2^{-}$ and $b_2^{+}$, and set the sender's payoffs from all $b$ actions to zero. Let $\sigma_c$ denote the experiment that recommends action $c$ deterministically in both states. Because $x>2$, the optimal Bayesian persuasion experiment is $\sigma_c$, yielding the sender a payoff of $x$. The experiment $\sigma_a$ remains obedient, and the three-component ambiguous experiment $\hat{\mu}$ constructed above remains obedient and yields the sender $159/70$.

    Set $x=249/112\simeq2.22321$, as in \citet{cheng2024additional1}. Then
    \begin{equation*}
        U_s(\hat{\mu})=\frac{159}{70}>\frac{249}{112}=u_s^{BP}.
    \end{equation*}
    The accompanying document reports a numerical verification that the sender's payoff from any obedient ambiguous experiment with at most two components is at most $x$. Thus, the sender benefits from ambiguous communication, but no binary ambiguous experiment benefits the sender. The replacement of $b_2$ by $b_2^{-}$ and $b_2^{+}$ makes $c$ a best reply only at the prior $p$, which helps simplify the calculation. 
    \end{proof}
\par
\endgroup

\begingroup
\color{black}
\section{A Local Argument for the Receiver's Gain from Pareto-Ranked Splittings of Obedient Experiments}\label{sec_local_argument_receiver_gain}

Let $\sigma^{*}$ be obedient and such that a Pareto-ranked splitting of it exists. For this local illustration, assume that $\phi_r$ is twice continuously differentiable in a neighborhood of $u:=u_r(\sigma^{*})$ and that $\phi_r''(u)<0$. Then for small enough $\epsilon>0$, one can always find $\overline{\sigma}$ and $\underline{\sigma}$ such that $(\overline{\sigma},\underline{\sigma},1/2)$ is a Pareto-ranked splitting of $\sigma^{*}$ satisfying
\begin{align*}
u_r(\overline{\sigma})=u+\epsilon,\qquad u_r(\underline{\sigma})=u-\epsilon.
\end{align*}
Let
\begin{equation*}
\kappa=\frac{\phi_r'(u-\epsilon)}{\phi_r'(u-\epsilon)+\phi_r'(u+\epsilon)}
\quad\text{and}\quad
\mu=\mu^{((\overline{\sigma},\underline{\sigma}),\kappa)}.
\end{equation*}
Then by Lemma \ref{lem_effmeasure}, $\mu$ is an obedient ambiguous experiment. To show that the receiver's payoff is higher under $\mu$ than under $\sigma^{*}$, we split the receiver's payoff change into two parts, keeping the obedient strategy fixed in both comparisons. Notice that $\kappa$ depends on $\epsilon$. Taking the first-order Taylor expansions of $\phi_r'(u-\epsilon)$ and $\phi_r'(u+\epsilon)$ gives
\begin{equation*}
\kappa=\frac{1}{2}-\frac{\phi_r''(u)}{2\phi_r'(u)}\epsilon+o(\epsilon).
\end{equation*}

First, the receiver's payoff change from replacing $\sigma^{*}$ by the unambiguous experiment $\kappa\overline{\sigma}+(1-\kappa)\underline{\sigma}$, under obedience, is
\begin{align*}
\kappa(u+\epsilon)+(1-\kappa)(u-\epsilon)-u
&=(2\kappa-1)\epsilon=-\frac{\phi_r''(u)}{\phi_r'(u)}\epsilon^2+o(\epsilon^2).
\end{align*}
Thus, the payoff gain from the additional weight on the Pareto-better experiment is second-order in $\epsilon$.

Second, the receiver's payoff change from replacing this unambiguous experiment by the ambiguous experiment $\mu$ is
\begin{equation*}
\phi_r^{-1}\big(\kappa\phi_r(u+\epsilon)+(1-\kappa)\phi_r(u-\epsilon)\big)
-\big[\kappa(u+\epsilon)+(1-\kappa)(u-\epsilon)\big].
\end{equation*}
Taking the second-order Taylor expansions of $\phi_r(u+\epsilon)$ and $\phi_r(u-\epsilon)$ yields
\begin{align*}
&\kappa\phi_r(u+\epsilon)+(1-\kappa)\phi_r(u-\epsilon)=\phi_r(u)+(2\kappa-1)\phi_r'(u)\epsilon
+\frac{1}{2}\phi_r''(u)\epsilon^2+o(\epsilon^2).
\end{align*}
Since $\kappa-1/2$ is of order $\epsilon$, the difference between the right-hand side and $\phi_r(u)$ is of order $\epsilon^2$. Applying $\phi_r^{-1}$ and using $(\phi_r^{-1})'(\phi_r(u))=1/\phi_r'(u)$ therefore gives
\begin{equation*}
U_r(\mu)=u+(2\kappa-1)\epsilon
+\frac{\phi_r''(u)}{2\phi_r'(u)}\epsilon^2+o(\epsilon^2).
\end{equation*}
It follows that the receiver's payoff change due to bearing the resulting ambiguity is
\begin{equation*}
U_r(\mu)-\big[\kappa(u+\epsilon)+(1-\kappa)(u-\epsilon)\big]
=\frac{\phi_r''(u)}{2\phi_r'(u)}\epsilon^2+o(\epsilon^2).
\end{equation*}
In other words, the ambiguity cost is also second-order in $\epsilon$.

Combining the two parts yields
\begin{equation*}
U_r(\mu)-u_r(\sigma^{*})
=-\frac{\phi_r''(u)}{2\phi_r'(u)}\epsilon^2+o(\epsilon^2)>0
\end{equation*}
for all sufficiently small $\epsilon>0$. Therefore, for a small enough symmetric splitting of $\sigma^{*}$, the payoff gain from the additional weight on the Pareto-better experiment is twice the ambiguity cost to leading order, so that the receiver benefits. The additional smoothness and curvature assumptions are used only for this local illustration; Theorem \ref{general_improvement}(ii) gives the exact condition for the receiver's gain without these assumptions.
\par
\endgroup

\begingroup
\color{black}
\section{Proofs of auxiliary results from Appendix \ref{sec_proof_results}}\label{sec_proofs_auxiliary_results}

    \subsection{Proof of Lemma \ref{lem_effective_measure_reduces_payoff}}
    
    \begin{proof}[Proof of Lemma \ref{lem_effective_measure_reduces_payoff}]
    
        Define, for all integers $k \in [1, n]$, 
        \begin{equation*}
            \mathcal{S}_k= \sum_{i=1}^{k} x_i \mu_i \left[ \sum_{j \neq i: j=1}^{k} \mu_j (y_i-y_j)\right]. 
        \end{equation*}
        
        Notice that when $k = n$, we have 
        \begin{align*}
            \mathcal{S}_n	 = &  \sum_{i=1}^n x_i \mu_i \left[  \sum_{j=1; j \neq i}^{n}  \mu_j (y_i-y_j)\right] \\
            = &  \sum_{i=1}^n x_i \mu_i \left[ (1-\mu_i) y_i - \sum_{j=1; j \neq i}^{n} \mu_j y_j\right]  \\
            = &  \sum_{i=1}^n x_i \mu_i \left[ y_i - \sum_{j=1}^{n} \mu_j y_j\right]  \\
            = &  \left( \sum_{j=1}^n \mu_j y_j \right) \left( \sum_{i=1}^n x_i \frac{\mu_i y_i}{\sum_{j=1}^n \mu_j y_j} - \sum_{i=1}^n x_i \mu_i \right). 
        \end{align*}
        Since $\left( \sum_{j=1}^n \mu_j y_j \right)  > 0$, it suffices to show $\mathcal{S}_{n} < 0$. We prove this by induction. Observe that $\mathcal{S}_{1} = 0$. For $k \geq 1$, 
        \begin{align*}
            \mathcal{S}_{k+1}= &  \sum_{i=1}^{k+1} x_i \mu_i \left[  \sum_{j=1; j \neq i}^{k+1}  \mu_j (y_i-y_j)\right] \\
            = & \sum_{i=1}^{k} x_i \mu_i \left[ \sum_{j \neq i: j=1}^{k} \mu_j (y_i-y_j)\right] + \sum_{i=1}^{k} x_i \mu_i[\mu_{k+1}(y_i-y_{k+1})] + \\
            & x_{k+1}\mu_{k+1} \sum_{j=1}^{k}[\mu_j(y_{k+1}-y_j)]\\
            = &  \sum_{i=1}^{k} x_i \mu_i \left[ \sum_{j \neq i: j=1}^{k} \mu_j (y_i-y_j)\right]  +  \sum_{i=1}^{k}\mu_i\mu_{k+1}(x_i-x_{k+1}) (y_i-y_{k+1})\\
            = & \mathcal{S}_{k} + \underbrace{\sum_{i=1}^{k}\mu_i\mu_{k+1}(x_i-x_{k+1}) (y_i-y_{k+1})}_{\leq 0}.
        \end{align*}
        
        For $k = j^*- 1$, 
        \begin{equation*}
            \sum_{i=1}^{j^*-1}\mu_i\mu_{j^*}(x_i-x_{j^*}) (y_i-y_{j^*}) \leq \mu_{i^*}\mu_{j^*}(x_{i^*}-x_{j^*}) (y_{i^*}-y_{j^*}) < 0.
        \end{equation*}
        Therefore, $0 = \mathcal{S}_{1} > \mathcal{S}_{j^*} \geq \mathcal{S}_{n}$. 
    \end{proof}

    \subsection{Proof of Lemma \ref{lem_existence_splitting}}
    \begin{proof}[Proof of Lemma \ref{lem_existence_splitting}]
        Define 
        \begin{equation}\label{sigma_lower_bar}
            \underline{\sigma}^{\lambda} = \frac{1}{1-\lambda} \sigma - \frac{\lambda}{1-\lambda}\hat{\sigma}
        \end{equation}
        where $\lambda \in (0,1)$. Observe that if $\underline{\sigma}^{\lambda}$ is a well-defined experiment, then $\lambda \hat{\sigma}  + (1-\lambda)\underline{\sigma}^{\lambda} = \sigma$, and $u_{s}(\underline{\sigma}^{\lambda}) < u_{s}(\sigma)$, $u_{r}(\underline{\sigma}^{\lambda}) < u_{r}(\sigma)$, so that $(\hat{\sigma}, \underline{\sigma}^{\lambda}, \lambda)$ is a Pareto-ranked splitting of $\sigma$. 
    
        It remains to show that there exists $\lambda \in (0,1)$ such that $\underline{\sigma}^{\lambda}$ is indeed an experiment. In other words, for each $\omega$, $\underline{\sigma}^{\lambda}(\cdot|\omega)$ must be a probability distribution over actions. 
        
        If $|\text{supp}( \sigma(\cdot|\omega))| = 1$, then  $\text{supp}(\hat{\sigma}(\cdot|\omega)) \subseteq \text{supp}(\sigma(\cdot|\omega))$ implies $\text{supp}(\hat{\sigma}(\cdot|\omega)) = \text{supp}(\sigma(\cdot|\omega))$. It follows that $\underline{\sigma}^{\lambda}(\cdot|\omega) = \sigma(\cdot|\omega)$ for all $\lambda \in (0,1)$, and is thus a distribution over actions.
    
        If $|\text{supp}(\sigma(\cdot|\omega))| > 1$, embed $\sigma(\cdot|\omega)$ into the Euclidean space $\R^{|\text{supp}(\sigma(\cdot|\omega))|}$ and notice that $\sigma(\cdot|\omega)$ is in the relative interior of the probability simplex $\Delta(\text{supp}(\sigma(\cdot|\omega)))$. Thus there exists $\epsilon_{\omega} > 0$ such that for all $x \in \R^{|\text{supp}(\sigma(\cdot|\omega))|}$ with $\sum_{i}x_{i} = 1$, if $\Vert x - \sigma(\cdot|\omega) \Vert < \epsilon_{\omega}$, then $x \in \Delta(\text{supp}(\sigma(\cdot|\omega)))$. Since $\text{supp}(\hat{\sigma}(\cdot|\omega)) \subseteq \text{supp}(\sigma(\cdot|\omega))$, one has $\hat{\sigma}(\cdot|\omega) \in \Delta(\text{supp}(\sigma(\cdot|\omega)))$ as well. Then, for all $\lambda \in (0,1)$,
        \begin{equation*}
            \underline{\sigma}^{\lambda}(\cdot|\omega) \in \R^{|\text{supp}(\sigma(\cdot|\omega))|}
            \quad\text{and}\quad
            \sum_{a} \underline{\sigma}^{\lambda}(a|\omega) = 1.
        \end{equation*} Moreover, there exists $\lambda_{\omega} > 0$ such that for all $\lambda \in (0, \lambda_{\omega})$, $\Vert \underline{\sigma}^{\lambda}(\cdot|\omega) - \sigma(\cdot|\omega) \Vert < \epsilon_{\omega}$, and thus $\underline{\sigma}^{\lambda}(\cdot|\omega) \in \Delta(\text{supp}(\sigma(\cdot|\omega)))$, making it a distribution over actions.
    
        Because $\Omega$ is finite, $\underline{\lambda} (\hat{\sigma}, \sigma) \equiv \min_{\omega: |\text{supp}(\sigma(\cdot|\omega))| > 1} \lambda_{\omega} > 0$. Therefore, for all $\lambda \in (0, \underline{\lambda} (\hat{\sigma}, \sigma))$, $\underline{\sigma}^{\lambda}$ is a well-defined experiment. \end{proof}

    \subsection{Proof of Lemma \ref{thm_inverstly_ranked_concavification_experiment_space}}
    
    \begin{proof}[Proof of Lemma \ref{thm_inverstly_ranked_concavification_experiment_space}]
Let $\lambda$ be a solution to \eqref{equ_program_Phi_star}. To prove (i), suppose towards a contradiction that two experiments in its support satisfy $u_s(\sigma_\theta)>u_s(\sigma_{\theta'})>u$ and $u_r(\sigma_\theta)<u_r(\sigma_{\theta'})$, with $\phi_r'(u_r(\sigma_\theta))\neq\phi_r'(u_r(\sigma_{\theta'}))$. Fix $\alpha\in(0,1)$ and let $\sigma=\alpha\sigma_\theta+(1-\alpha)\sigma_{\theta'}$. Concavity of $1/\phi_r'$, positivity of $\phi_s(u_s(\sigma))-\phi_s(u)$, and concavity of $\phi_s$ give
\begin{align*}
\Phi_u(\sigma)-\alpha\Phi_u(\sigma_\theta)
-(1-\alpha)\Phi_u(\sigma_{\theta'})\geq &
\alpha(1-\alpha)
\bigl[\phi_s(u_s(\sigma_\theta))
-\phi_s(u_s(\sigma_{\theta'}))\bigr]\\
&\times
\left[
\frac{1}{\phi_r'(u_r(\sigma_{\theta'}))}
-\frac{1}{\phi_r'(u_r(\sigma_\theta))}
\right]>0.
\end{align*}
The strict inequality follows from strict monotonicity of $\phi_s$, concavity of $\phi_r$, and the distinct receiver derivatives. By the merging argument in the proof of Lemma \ref{lem:inversely_ranked_pareto_ranked}, this contradicts optimality of $\lambda$.

For (ii), the same argument applies: convexity of $1/\phi_r'$ and $\phi_s(u_s(\sigma))-\phi_s(u)<0$ yield the same bound.
\end{proof}
    
    \subsection{Proof of Lemma \ref{thm_Pareto_ranked_concavification_experiment_space}}
    \begin{proof}[Proof of Lemma \ref{thm_Pareto_ranked_concavification_experiment_space}]
Let $\lambda$ be a solution to \eqref{equ_program_Phi_star}. To prove (i), suppose towards a contradiction that some $\sigma_\theta\in\text{supp}(\lambda)$ admits a Pareto-ranked splitting $(\overline{\sigma},\underline{\sigma},\alpha)$ satisfying the stated conditions. Concavity of $1/\phi_r'$, $\phi_s(u_s(\sigma_\theta))-\phi_s(u)\leq 0$, and concavity of $\phi_s$ give
\begin{align*}
\alpha\Phi_u(\overline{\sigma})
+(1-\alpha)\Phi_u(\underline{\sigma})
-\Phi_u(\sigma_\theta) \geq &
\alpha(1-\alpha)
\bigl(u_s(\overline{\sigma})-u_s(\underline{\sigma})\bigr)\\
&\times
\left[
\frac{\phi_s'(u_s(\overline{\sigma}))}
{\phi_r'(u_r(\overline{\sigma}))}
-\frac{\phi_s'(u_s(\underline{\sigma}))}
{\phi_r'(u_r(\underline{\sigma}))}
\right]>0,
\end{align*}
where the last inequality follows from Pareto-ranking and the assumed inequality between derivative ratios. As in the proof of Lemma \ref{lem:inversely_ranked_pareto_ranked}, splitting $\sigma_\theta$ preserves the mean experiment and strictly increases the objective, contradicting optimality.

For (ii), the same argument applies: convexity of $1/\phi_r'$ and $\phi_s(u_s(\sigma_\theta))-\phi_s(u)>0$ yield the same bound.
\end{proof}	
\par
\endgroup

\textcolor{black}{\subsection{Proofs of Proposition \ref{prop_robust_benefit} and Corollary \ref{coro_robust_benefit}}}

\begingroup
\color{black}
\begin{proof}[Proof of Proposition \ref{prop_robust_benefit}]
Throughout, we equip $\Delta_f(\Sigma)$ with the weak topology, take interiors in the Euclidean experiment space described in Section \ref{sec_properties_optimal_ambiguous_communication}, and write $\mathbb{E}_{\lambda}$ for the expectation operator with respect to $\lambda\in\Delta_f(\Sigma)$. Since ambiguous communication benefits the sender,
\begin{equation*}
\left\{\lambda\in\Delta_f(\Sigma):\mathbb{E}_{\lambda}[\sigma]\in\Sigma^{*},\ \mathbb{E}_{\lambda}[\Phi_{u_s^{BP}}(\sigma)]>0\right\}\neq\emptyset
\end{equation*}
by Lemma \ref{lem_characterization_improving_u}. Since $\Sigma^{*}$ is convex and has a non-empty interior, any point in $\Sigma^{*}$ can be approached by points in its interior. More explicitly, fix a $\lambda$ in the above set and $\sigma^0\in\mathrm{int}\,\Sigma^{*}$. For any $\epsilon\in(0,1)$, the distribution $(1-\epsilon)\lambda+\epsilon\delta_{\sigma^0}$ has its mean in $\mathrm{int}\,\Sigma^{*}$, and, for sufficiently small $\epsilon$, its expectation of $\Phi_{u_s^{BP}}$ remains strictly positive. Thus,
\begin{equation*}
\left\{\lambda\in\Delta_f(\Sigma):\mathbb{E}_{\lambda}[\sigma]\in\mathrm{int}\,\Sigma^{*},\ \mathbb{E}_{\lambda}[\Phi_{u_s^{BP}}(\sigma)]>0\right\}\neq\emptyset.
\end{equation*}
As the expectations in this set are continuous in $\lambda$ and $\mathrm{int}\,\Sigma^{*}\times(0,\infty)$ is open, this set is open.

Finally, \eqref{equ_effmeasure} and \eqref{equ_reverse_effmeasure} give a one-to-one correspondence between $\mu\in\Delta_f(\Sigma)$ and $\lambda=em^{\mu}\in\Delta_f(\Sigma)$. This correspondence and its inverse are continuous because $\phi_r'\circ u_r$ is positive and continuous on the compact set $\Sigma$. Therefore, the corresponding set of ambiguous experiments $\mu$ is non-empty and open. Every such $\mu$ is obedient by Lemma \ref{lem_effmeasure} and benefits the sender by the payoff identity in the proof of Lemma \ref{lem_characterization_improving_u}.
\end{proof}

\begin{proof}[Proof of Corollary \ref{coro_robust_benefit}]
By the proof of Proposition \ref{prop_robust_benefit}, there exists an ambiguous experiment $\mu$ that benefits the sender and satisfies $\sum_{\theta}em^{\mu}_{\theta}\sigma_{\theta}\in\mathrm{int}\,\Sigma^{*}$. Fix this $\mu$. By \eqref{equ_effmeasure}, its effective mean $\sum_{\theta}em^{\mu}_{\theta}\sigma_{\theta}$ is continuous in the values $\phi_r'(u_r(\sigma_{\theta}))$, $\sigma_{\theta}\in\mathrm{supp}(\mu)$. Hence it remains in $\Sigma^{*}$ under small perturbations of these values, so that $\mu$ remains obedient by Lemma \ref{lem_effmeasure}. Since $\mu$ is fixed, the sender's payoff $U_s(\mu)>u_s^{BP}$ is unchanged. Finally, within the model's class of differentiable concave functions, small locally uniform perturbations of $\phi_r$ imply small perturbations of $\phi_r'$ at the finitely many relevant payoff values \citep[Theorem 25.7, p. 248]{Rockafellar1970}. Therefore, the same ambiguous experiment continues to benefit the sender under small perturbations of $\phi_r$.
\end{proof}
\par
\endgroup

\bibliographystyle{econ-a}
\bibliography{references.bib}

\end{document}